\documentclass{article}

\usepackage[T1]{fontenc}
\usepackage[utf8]{inputenc}
\usepackage{microtype}
\usepackage{libertine}
\usepackage{amsmath}
\usepackage{amssymb}
\usepackage{amsthm}
\usepackage{mathtools}
\usepackage{dsfont}
\usepackage{bm}
\usepackage{physics}
\usepackage{caption}
\usepackage{esvect}
\usepackage[dvipsnames]{xcolor}
\usepackage{hyperref}
\hypersetup{
  colorlinks,
  linkcolor={red!50!black},
  citecolor={blue!50!black},
  urlcolor={blue!80!black}
}
\usepackage{authblk}
\usepackage[doi=false,isbn=false,url=false,eprint=true,date=year,backend=biber, citestyle=alphabetic, style=alphabetic]{biblatex}
\theoremstyle{plain}
\newtheorem{definition}{Definition}[section]
\newtheorem{lemma}[definition]{Lemma}
\newtheorem{prop}[definition]{Proposition}
\newtheorem{theorem}[definition]{Theorem}
\newtheorem{corol}[definition]{Corollary}
\newtheorem{hypothesis}[definition]{Hypothesis}

\theoremstyle{remark}

\newtheorem{ex}[definition]{Example}
\newtheorem{remark}[definition]{Remark}

\usepackage{bbm}
\usepackage{geometry}

\newcommand{\N}{\mathbf{N}}

\newcommand{\E}{\mathds{E}}

\newcommand{\C}{\mathbb{C}}
\newcommand{\CP}{\mathbb{CP}}

\newcommand{\Mat}{{\rm M}}
\newcommand{\Unit}{\mathbb{U}}
\newcommand{\ee}{\mathrm{e}}
\newcommand{\Sym}{\mathfrak{S}}
\newcommand{\Partition}{\mathcal{P}}
\newcommand{\Weingarten}{{\rm W}}

\newcommand{\pure}{{\rm p}}
\newcommand{\mix}{{\rm m}}
\newcommand{\Phim}{\Phi^{\mix}}
\newcommand{\Phip}{\Phi^{\pure}}
\newcommand{\Km}{\mathcal{K}^{\mix}}
\newcommand{\Kp}{\mathcal{K}^{\pure}}
\newcommand{\Gb}{\overline{\mathcal{G}}}
\newcommand{\G}{\mathcal{G}}
\newcommand{\HS}{{\rm HS}}
\newcommand{\Gau}{\mathcal{N}}

\newcommand{\Positive}{\mathcal{P}}
\newcommand{\WeingCp}[1]{\Weingarten^{\pure}_{{\rm C}, #1}}
\newcommand{\WeingCm}[1]{\Weingarten^{\mix}_{{\rm C}, #1}}
\newcommand{\planar}{\mathbb{G}}

\newcommand{\treeSet}{\mathbb{P}_{\mix}}
\newcommand{\symid}{\mathrm{id}}
\newcommand{\exeps}{\varepsilon}
\newcommand{\FSymp}{\mathcal{F}\Sym_{\pure}}
\newcommand{\FSymm}{\mathcal{F}\Sym_{\mix}}
\newcommand{\FSymg}{\mathcal{F}\Sym_{\Gau}}
\newcommand{\ISymm}{\mathcal{I}\Sym_{\mix}}

\newcommand{\ISymp}{\mathcal{I}\Sym_{\pure}}
\newcommand{\ISymg}{\mathcal{I}\Sym_{\Gau}}
\newcommand{\PPart}{\Partition_{\pure}}
\newcommand{\Mnc}{\mathsf{M}}
\newcommand{\scalp}{\mathsf{S}_{\pure}}
\newcommand{\scalg}{\mathsf{S}_{\Gau}}
\newcommand{\mwalks}{\mathsf{w}}
\newcommand{\NC}{\mathsf{NC}}
\newcommand{\trans}{\mathsf{T}}
\newcommand{\vphim}{\varphi^{\mix}}
\newcommand{\vphip}{\varphi^{\pure}}
\newcommand{\fcumm}{\kappa^{\mix}}
\newcommand{\fcump}{\kappa^{\pure}}

\DeclareMathOperator{\Id}{Id}

\DeclareMathOperator{\Cat}{Cat}
\DeclareMathOperator{\Aut}{Aut}
\usepackage{listofitems}
\newcommand\cycle[2][\,]{%
  \readlist\thecycle{#2}%
  (\foreachitem\i\in\thecycle{\ifnum\icnt=1\else#1\fi\i})%
}

\newcommand{\bsig}{{\bm{\sigma}}}
\newcommand{\btau}{{\bm{\tau}}}
\newcommand{\bnu}{{\bm{\nu}}}

\newcommand{\un}{\mathbbm{1}}

\author[1]{Thomas Buc--d'Alché}
\author[2]{Luca Lionni}
\affil[1]{IRMA, Université de Strasbourg, 7 rue René Descartes,
67084 Strasbourg Cedex, France}
\affil[2]{CNRS, ENS de Lyon, LPENSL, UMR5672, 69342 Lyon cedex 07, France}
\title{Properties of tensorial free cumulants}
\date{}

\begin{document}

\maketitle

\begin{abstract}
In the past two years, several points of view have been proposed to address the question of the generalization of the theory of free probability to random tensors with different invariances, and it is unclear at this point whether they lead to the same notions of tensorial free cumulants and freeness. One way to approach this problem, developed by Collins, Gurau and the second named author for local unitary invariant random tensors, relies on finite size quantities involving averages over the invariance group, and whose asymptotics naturally possess the properties expected for tensorial generalizations of free cumulants of arbitrary orders. At this point, this approach has only been carried out  for certain distributions, and for a subset of the moments that define such theories, and a more systematic and exhaustive study is lacking. 

This is the program initiated in this paper: we link this approach to the one proposed by Nechita and Park; extend a number of their results as well as those of the aforementioned paper  to arbitrary orders of fluctuations, thereby generalizing higher order free cumulants; push further the study of  distributions with larger invariance groups; detail the link with the asymptotics of the free-energies of the tensor HCIZ and BGW integrals; and provide formulae for tensorial free cumulants of products of tensors. 

Another important question is that of the definition of concrete distributions whose tensorial free-cumulants take non-trivial values. We compute the tensorial free cumulants for Gaussian random tensors with non-trivial covariances, and show that they provide such examples. 
\end{abstract}

\setcounter{tocdepth}{2}
\tableofcontents

\section{Introduction}

Throughout this paper, tensors on $D$ inputs are arrays or complex numbers of the form $T=(T_{i_1, \ldots, i_D})$ or $A=(A_{i_1, \ldots, i_D ; j_1,\ldots, j_D})$, where the indices $i_c, j_c$ take value in $[N]=\{1,\ldots, N\}$, $N\in\mathbb{N}^{*}$.\footnote{We sometimes consider more generally indices  $i_c, j_c$ taking value in $[N_c]=\{1,\ldots, N_c\}$,  $N_c\in\mathbb{N}^{*}$, but we then explicitly say so.} Adopting the vocabulary from quantum physics, tensors of the former kind  are called \emph{pure}, while those of the second kind are called \emph{mixed}. 
Recent advances in several domains are calling for a better mathematical understanding of \emph{random} tensors. Applications include quantum information \cite{Hayden2016, aubrun_alice_2017, RMTQI, marginals, cheng2024random,  Penington:2022dhr, Genuine-multi, dartois_injective_2024, carrozza_tensor_2026}, quantum gravity \cite{bonzom_critical_2011, Lionni2018, gurau_random_2017, Gurau:2016cjo, Hayden2016, Penington:2022dhr, Genuine-multi}, data science and statistics \cite{montanari_statistical_2014a,arous_permutation_2024, kunisky_tensor_2024}, condensed matter (e.g. tensor networks) \cite{orus_tensor_2019, Lancien:2019fga},  and statistical physics (e.g. spin glasses) \cite{crisanti_sphericalpspin_1992,castellani_spinglass_2005a,dartois_injective_2024}, random hyper graphs and community detection \cite{friedman_second_1995}, etc.

Local unitary (LU) transformations result from the action of tensor products $U=U_1 \otimes \cdots \otimes U_D$ of $N\times N$ unitary matrices $U_c\in U(N)$ on pure and mixed tensors, respectively by left multiplication $T \rightarrow U T$ and conjugation $A \rightarrow U A U^\dagger$, where:
 \begin{align}
  (U T)_{i_1, \ldots, i_D} & = \sum_{b_1, \ldots, b_D=1}^N (U_1)_{i_1, a_1 }\cdots (U_D)_{i_D, a_D }T_{a_1, \ldots, a_D}\;,\\
 (UA U^\dagger)_{i_1, \ldots, i_D;j_1,\ldots, j_D} &= \sum_{\{a_c, b_c\in [N]\}_{1\le c\le D}} A_{a_1, \ldots, a_D; b_1, \ldots, b_D}   \prod_{c=1}^D (U_c)_{i_c, a_c }  (\bar{U_c})_{j_c, b_c } \;.
 \end{align}
A random tensor is said to be \emph{LU-invariant} if its distributions is invariant under LU transformations. The appropriate notion of (macroscopic) moments in this context are the expectations of a family of polynomials in the components of the tensors called trace-invariants, invariant under LU-transformations, which separate the LU orbits, generate the ring of LU invariant polynomials,  and are asymptotically linearly independent (see \cite{collins_free_2025, carrozza_tensor_2026} where these properties are stated and proven, and in which the a number of original references are given). These trace-invariants are labeled by $D$-tuples of permutations, up to relabeling equivalence. The data of these trace-invariants and of  the LU-orbit generalizes the data of the spectrum and unitary orbit.  A trace-invariant is connected if it does not decompose as a product of smaller trace-invariants (if the tuple that labels it satisfies some transitivity condition).

In quantum physics, a quantum system is modeled by a Hilbert space $\mathcal{H}$, and pure and mixed quantum states  respectively   correspond to normalized elements $\lvert \varphi \rangle$ of $\mathcal{H}$, and Hermitian, positive trace-one operators $\rho$ on  $\mathcal{H}$.  Entanglement \cite{horodecki_quantum_2009,walter_multipartite_2016}
is a correlation  between the different sub-spaces $\mathcal{H}_c$ of a quantum system $\mathcal{H}=\mathcal{H}_1\otimes \cdots \otimes \mathcal{H}_D$, important for its technological applications and implications in e.g.~condensed matter theory and quantum gravity.  Pure and mixed tensors are then seen respectively as the coefficients of the expansions of pure and mixed quantum states on the canonical product basis of $\mathcal{H}=\mathcal{H}_1\otimes \cdots \otimes \mathcal{H}_D$. LU-invariance is important in this context: LU transformations do not affect entanglement, and the orbits under these transformations  define the finest notion of equivalently entangled states \cite{kraus_local_2010,kraus_local_2010a}.
Random quantum states are studied for instance to study generic properties of entanglement \cite{Page:1993df, Hayden:2005sqo, Hayden2016, RMTQI, marginals, dartois_injective_2024, carrozza_tensor_2026}, and for a distribution to be relevant in this context, it must be LU invariant. This motivates the study of LU-invariant random tensors.

Free probability (see e.g.~\cite{nica_lectures_2006}) allows for the study of non-commutative random variables -- elements of non-commutative probability space -- such as $N\times N$  random matrices $M$ in the limit of infinite size, $N\rightarrow \infty$. It was first introduced by Voiculescu in the 80s \cite{voiculescu_symmetries_1985,voiculescu_addition_1986} in the context of operator algebras, before its application to random matrices \cite{voiculescu_limit_1991}. The analogues of independence, convolution, cumulants etc, become freeness, free convolution, and free cumulants, which are central to the theory. The free cumulants $\{\kappa_n\}_{n\in \mathbb{N}^{*}}$ indeed contain the same information as the first order asymptotic moments $\varphi_n = \lim_{N\rightarrow \infty} \mathbb{E}[\Tr(M^n)/N]$, $n\in \N^{*}$, that characterize the limiting eigenvalue distribution of $M$.  
One can indeed compute the former from the latter from free moment-cumulant formulae, and conversely. While for classical commutative random variables, the moment cumulants formulae involve convolutions and Moebius inversion in the lattice of partitions, it is the lattice of non-crossing partitions that plays this role for their free analogues. 
As for their commutative analogues, free cumulants are additive for free random variables, and allow for efficient computations of the asymptotic spectrum of sums or products of random matrices, in particular through the use of generating functions -- the resolvent or Cauchy transform and the R-transform. The latter is the free analogue of the logarithm of the Fourier transform, and it can be derived as a rescaled limit of the logarithm of the HCIZ integral.  Higher order free cumulants of random matrices \cite{collins_second_2007} play a similar role for the infinite size asymptotics of the classical cumulants. 

\

In the past two years, several works have undertaken the vast task of generalizing the theory of free probability for random tensors, LU-invariant \cite{collins_free_2025, nechita_tensor_2025}, or invariant under other groups \cite{Kuninsky-et-al, bonnin_freeness_2024,bonnin_characterization_2026}. 
The starting point is the definition of a generalization of free cumulants, whose knowledge is equivalent to the correlation functions defining the theory, which as stated above are defined in terms of expectations and  cumulants of trace-invariants. 
 From now on, motivated by the applications to quantum information theory, we restrict the discussion to LU-invariance.

 For random matrices, the free-cumulants at arbitrary orders can be obtained as infinite $N$ limits of some finite size precursors \cite{collins_second_2007, capitaine_cumulants_2006,capitaine_geometric_2008, lacroix_finite_2025a}, defined in a natural way from the generating functions of cumulants of the  entries of the matrix. A number of the crucial properties that make free cumulants better suited to study polynomials in independent random matrices can be derived this way.
 These finite $N$ precursors have been generalized in a similar way in \cite{collins_free_2025}
 in full generality for pure and mixed LU-invariant random tensors. Their knowledge is equivalent to that of the classical cumulants of connected trace-invariants. In the same way as for matrices, tensorial free cumulants are then defined in this paper as the infinite size limits of these precursors, and they consequently satisfy generalized versions of the properties of usual free cumulants, such as the additivity for independent invariant random tensors.  These limits are studied in the connected case (for moments) in \cite{collins_free_2025}  for  pure LU invariant random tensors for which the dominant scale in $N$ of the correlations defining the distribution are the same as for the pure complex Gaussian tensor (whose coefficients are i.i.d.~complex Gaussian random variables of variance $1/N^{D-1}$). Most results are derived for the moments that scale the strongest in $N$, called \emph{first order}.  Tensorial free cumulants of mixed LU-invariant random tensors are studied in the same paper under a similar scaling assumption.

In \cite{nechita_tensor_2025}, Nechita and Park study mixed random tensors whose moments scale like random matrices, and they adopt another starting point to define tensorial free cumulants in this context: the moment-cumulant formulae defining these quantities are directly stated ``at the limit'', and then tensorial freeness is shown to hold for random tensors satisfying certain conditions, as well as a number of other properties. Nechita and Park's study is restricted to moments, which can be seen as a restriction to connected trace-invariants.
While it is clear that this notion coincides with  the pure  tensorial free cumulants defined in \cite{collins_free_2025} for the moments that scale the strongest in $N$ (first order), it remains to see whether the frameworks are compatible to higher extents.

At this stage, the examples of random tensors  studied in \cite{collins_free_2025, nechita_tensor_2025}  are mostly distributions built from random vectors or random matrices, with large invariance groups and whose tensorial free cumulants trivialize the moment-cumulant formulae. 
Another important point is therefore to define concrete tensor distributions whose free cumulants  take non-trivial values, in the sense that they for instance allow enumerating elements in the lattice involved in the moment-cumulant formulae that satisfy.  

\

In the present paper, we pursue the study of tensorial free cumulants, in the following directions: 
\begin{itemize}
\item   We show in Sec.~\ref{sec:tensors-matrix-product-scaling} that for LU-invariant random tensors satisfying the scaling assumption of \cite{nechita_tensor_2025},  the finite $N$ precursors of the tensorial free cumulants defined in \cite{collins_free_2025} converge in the connected (or transitive) case to the tensorial free cumulants of Nechita and Park. We push their study further for such distributions by providing the limits in full generality, lifting the restriction to the connected case. This generalizes for such tensors the higher order free cumulants of random matrices \cite{collins_second_2007}. We also generalize the relation linking for matrices the free energy of the HCIZ integral and the R-transform, the generating function of free cumulants \cite{GUIONNET2002461, collins_moments_2003, GUIONNET2005435, collins_second_2007}. This clarifies  why for the first order invariants, some relations that resemble the ones of \cite{nechita_tensor_2025} had been derived in \cite{collins_tensor_2023}, were the asymptotics of the tensor HCIZ integral are derived under the same scaling assumption. 
\item As mentioned above, until now the examples of LU-invariant random tensors studied are mostly pure or mixed random tensors with larger invariance groups. In \cite{nechita_tensor_2025}, a theorem is given which shows that global unitarily invariant random matrices that scale appropriately have tensorial free cumulants that either coincide with their matricial free cumulants, or vanish otherwise. In Sec.~\ref{sec:Coarser-invariance}, we show a similar result for finite $N$ precursors, and for the more general case of a LU-invariant random tensors satisfying a coarser local unitary invariance. Our result encompasses their result asymptotically in the mixed, connected case, and extends it to the disconnected case. It asymptotically provides a universality result in the pure, global unitary invariant case. These results explain why global unitary invariant distributions do not provide examples with non-trivial tensorial free cumulants.
\item We also push the results of \cite{collins_free_2025} further: for pure LU-invariant random tensors that scale like the standard complex Gaussian tensor, we derive in Sec.~\ref{sec:Gaussian-scalings} the formulae defining the cumulants of arbitrary orders, connected or not, thus also generalizing the higher order free cumulants of random matrices for such tensors (in \cite{collins_free_2025} , this was only done in the connected case). We also derive the inverse formulae in the connected case (in \cite{collins_free_2025}, this was only done at first order and few other cases). We study another scaling function, obtained for pure Gaussians with random covariances (next point), or which we derive the formulae defining the cumulants of arbitrary orders and the inverse formulae at arbitrary order, connected or not. For both these scalings, we show a result similar to that relating the free energy of the tensor HCIZ integral and the generating function of first order tensorial free cumulants, where the former is replaced by a tensor version of the BGW integral. 
\item The motivation at the origin of this work was to produce concrete examples  with non-trivial tensorial free cumulants. In order to do so, we study in Sec.~\ref{sec:gauss-cov} pure Gaussian random tensors with non-trivial covariance tensors, random or not. For random matrices, models of this kind are for instance relevant for the Kontsevitch model \cite{kontsevich_intersection_1992}, dually weighted matrix models \cite{kazakov_character_1996,guionnet_character_2005}, matrix models for causal dynamical triangulations \cite{benedetti_imposing_2009}, and models implementing renormalization flows. Our results generalize the main ingredient -- the Gaussian with non-trivial covariance -- to the tensor case\footnote{See also \cite{benedetti_phase_2012}, which defined a tensor generalization that involves $D+1$ coupled random tensors (``colored'' case). } For these distributions, we derive elegant and simple formulae relating the classical and tensorial free cumulants to those of the covariance, in the deterministic and random cases.  We then consider simple examples whose tensorial free cumulants take non-trivial values. 
\item We finally address in Sec.~\ref{sec:gener-form-prod} and  Sec.~\ref{sec:large-n-limit} the question of the tensorial free cumulants for (global) products of tensors\footnote{Global as a matrix or vector product, in opposition to local  product at the level of a single input and output, studied at first order in \cite{collins_free_2025}.  }, which encompass the Gaussians with random covariances. We derive general formulae at finite $N$, and their asymptotics for some of the scaling assumptions mentionned above, which generalize the analogous formulae for the matrix case.
\end{itemize}

\paragraph{Acknowledgements.}
This work was supported by the  ANR JCJC project ``RTFPQuEnt" (ANR-25-CE40-5465).  L.L. also acknowledges support  from the  ANR PRC project  ``TAGADA" (ANR-25-CE40-5672). T.BdA. was supported by ERC Project LDRAM : ERC-2019-ADG Project
884584, and by ERC Project InSpeGMos 101096550.

\section{Notations and prerequisites}

\subsection{Partitions and permutations}
\label{sub:partitions}

Given $n \in \N^{*}$, we denote by $[n]$ the set $\{ 1, \ldots, n\}$. The
cardinal of a finite set $S$ is denoted by $\# S$. In general, we will
denote by $N$ the local dimension of the tensors (the size of the index set), and by $D$ the number of
inputs and outputs -- sometimes called number of colors.

\paragraph{Partitions.}

Many of the results given in the sequel are expressed using partitions
and permutations. We denote the set of partitions on $[n]$ by
$\Partition(n)$.
When discussing pure tensors, we will consider partitions of
$[n] \cup [\bar{n}]$, where $[\bar{n}]$ denotes a distinct copy of
$[n]$, containing the ``barred'' integers $\bar{1}, \ldots, \bar{n}$. We denote by $\Partition(n, \bar{n})$ the set of partitions of $[n] \cup [\bar{n}]$.
\begin{definition}[Pure partitions]\label{def:pure-partition}
  Let $n \in \N^{*}$. The set of pure partitions is defined by
  \begin{equation*}
    \Partition_{\pure}(n) = \Bigl\{ \Pi \in \Partition(n, \bar{n}) \colon \forall S \in \Pi, \# S\cap [n] = \# S \cap [\bar{n}]\Bigr\}.
  \end{equation*}
\end{definition}
In general, we will use the convention that lower case $\pi$ denote
usual partitions and upper case $\Pi$ denote pure partitions.

We define the distinguished empty partitions
\begin{equation*}
  0_{n} = \Bigl\{\{i\} \colon i\in[n]\Bigr\},  \quad 0_{n, \bar{n}} = \Bigl\{\{i\} \colon i\in[n]\cup[\bar{n}]\Bigr\},
\end{equation*}
and full partitions
\begin{equation*}
  1_{n} = \{[n]\}, \quad 1_{n, \bar{n}} = \{[n] \cup [\bar{n}]\}.
\end{equation*}
These partitions are especially useful when considering the partial
order $\leq$ on the lattice of partitions of a set. Given $\pi_{1}$
and $\pi_{2}$ two partitions of a given set, we say that $\pi_{1}$ is
finer than $\pi_{2}$ (or that $\pi_{2}$ is coarser than $\pi_{1}$), denoted
by $\pi_{1} \leq \pi_{2}$, if for all $S_{1} \in \pi_{1}$ there exists
$S_{2} \in \pi_{2}$ such that $S_{1} \subset S_{2}$. The join
$\pi_{1} \vee \pi_{2}$ of two partitions $\pi_{1}$ and $\pi_{2}$ is the finest
partition coarser than both $\pi_{1}$ and $\pi_{2}$.

Finally, if $\pi \leq \pi'$, then for all $S \in \pi'$ the partition
$\pi\vert_{S}$ is the partition of $S$ whose blocks are the blocks of
$\pi$ contained in $S$.

\begin{remark}[Join between $\Partition(n)$ and $\PPart(n)$]\label{rem:joint-part-ppart}
  In Section \ref{sec:gener-form-prod}, we shall study expressions
  that involve both partitions $\pi_{1} \in \Partition(n)$ and pure
  partitions $\Pi_{2} \in \PPart(n)$. We shall abuse notation and see
  $\pi_{1}$ as a partition in $\Partition(n, \bar{n})$ by adjoining
  the blocks $\{\bar{1}\}, \ldots, \{\bar{n}\}$. We can then consider
  the partition $\pi_{1} \vee \Pi_{2} \in \Partition(n, \bar{n})$.

  However, note that if $\Pi' \in \PPart(n)$ and $\Pi'' \in \Partition(n)$
  satisfy $\Pi' \leq \Pi''$, then we immediately have
  $\Pi'' \in \PPart(n)$. This implies that
  $\pi_{1} \vee \Pi_{2} \in \PPart(n)$.
\end{remark}

\paragraph{Permutations.}

We denote the set of permutations on $[n]$ by $\Sym_{n}$ and its
identity element by $\symid_{n}$. For $\sigma \in \Sym_{n}$, we denote by
$\# \sigma$ the number of cycles of $\sigma$, and write $\ell(\sigma^{(i)})$ for the
length of a cycle $\sigma^{(i)}$ of $\sigma$. We define $\Pi(\bm{\sigma})$ to be
the partition of $[n]$ induced by the disjoint cycles of $\sigma$. We
extend this definition to $D$-tuples of permutations
$\bm{\sigma} = (\sigma_{c})_{c = 1, \ldots, D} \in \Sym_{n}^{D}$ by
\begin{equation*}
  \Pi(\bm{\sigma}) = \Pi(\sigma_{1}) \vee \cdots \vee \Pi(\sigma_{D}).
\end{equation*}
We define for convenience
$\Pi(\bm{\sigma}, \bm{\rho}) = \Pi(\bm{\sigma}) \vee \Pi(\bm{\rho})$ for
$\bm{\sigma}, \bm{\rho} \in \Sym_{n}^{D}$. We write $K(\bsig)=\# \Pi(\bsig)$
and say that $\bsig$ is connected if $K(\bsig)=1$.

When considering pure tensors, we may see a permutation
$\sigma \in \Sym_{n}$ as a map $[n] \to [\bar{n}]$. We thus define the pure
partition
\begin{equation}\label{eq:def-pure-part-sym}
  \Pi_{\pure}(\sigma) = \Bigl\{\{i, \overline{\sigma(i)}\} \colon i \in [n]\Bigr\}.
\end{equation}
This definition is naturally extended to $D$-tuples $\bm{\sigma} \in \Sym_{n}^{D}$ through
\begin{equation*}
  \Pi_{\pure}(\bm{\sigma}) = \Pi_{\pure}(\sigma_{1}) \vee \cdots \vee \Pi_{\pure}(\sigma_{D}).
\end{equation*}
Similarly as before, we set $K_{\pure}(\bm{\sigma}) = \# \Pi_{\pure}(\bm{\sigma})$, and we extend this notation to the case of several families. In particular,  given an additional permutation $\eta \in \Sym_{n}$, $K_{\pure}(\bm{\sigma}, \eta) = \# \Bigl(\Pi_{\pure}(\bm{\sigma}) \vee \Pi_{\pure}(\eta)\Bigr)$.

A pure partition $\Pi \in \PPart(n)$ may be restricted to a partition in
$\Partition(n)$: we write
\begin{equation}\label{def:restriction-pure-n}
  \Pi_{[n]} = \Bigl\{ S \cap [n] \colon S \in \Pi \Bigr\}.
\end{equation}
If we fix a permutation $\eta \in \Sym_{n}$, this construction gives a bijection
\begin{equation}\label{eq:bijection-permutation}
  \begin{cases}
    \Partition(n) &\to \PPart(n)\\
    \pi &\mapsto \Bigl\{ S \cup \eta(S) \colon S \in \pi \Bigr\}.
  \end{cases}
\end{equation}
In particular, using this bijection with either $\symid_{n}$ or $\eta$ we have
\begin{equation*}
  \Pi(\bm{\sigma}\eta^{-1}) \simeq \Pi_{\pure}(\bm{\sigma}\eta^{-1}, \symid_{n}) \simeq \Pi_{\pure}(\bm{\sigma}, \eta).
\end{equation*}
By considering the number of blocks of these permutation, we get
$K(\bsig\eta^{-1})= K_\mathrm{p}(\bsig, \eta)$.

If $\Pi_\mathrm{p}(\bsig) \leq \Pi$, it is enough to know the blocks
$B$ of $\Pi_{[n]}$ to reconstruct the blocks $G$ of $\Pi$: each $B$ is
supplemented by the $\sigma_c(i)$, where $1\le c \le D$ and
$i\in {[n]}$.

\paragraph{Distances and genus.}

The symmetric group $\Sym_{n}$ may be seen as a metric space by
considering the graph distance associated to the Cayley graph of the
symmetric group generated by the transpositions. More explicitly,
given $\sigma \in \Sym_{n}$ we define $|\sigma|$ to be the minimal number of
term in a factorization of $\sigma$ as a product of transposition. Note
that $|\sigma| = n - \# \sigma$. We then define the distance on $\Sym_{n}$
\begin{equation}\label{eq:permutation-distance}
  d(\sigma, \rho) = |\sigma\rho^{-1}| \quad \text{ for } \sigma, \rho \in \Sym_{n}.
\end{equation}
These notations extend to the case of $D$-tuples of permutations. We
write for $\bm{\sigma} \in \Sym_{n}^{D}$
\begin{equation}\label{eq:def-size-tuple-perm}
  |\bm{\sigma}| = \sum_{c = 1}^{D} |\sigma_{c}| \qquad \text{ and } \qquad \# \bm{\sigma} = \sum_{c = 1}^{D} \# \sigma_{c}.
\end{equation}
The distance $d$ may be extended to a distance on $\Sym_{n}^{D}$ through
\begin{equation}\label{eq:permutation-distance-family}
  d(\bm{\sigma}, \bm{\rho}) = \sum_{c = 1}^{D}d(\sigma_{c}, \rho_{c}) \quad \text{ for } \sigma, \rho \in \Sym_{n}^{D}.
\end{equation}
If one of the tuples of permutations is of the form $\btau=(\tau, \ldots, \tau)$, we use the notation $d(\bm{\sigma}, \tau)$.

Given two permutations $\sigma, \tau \in \Sym_{n}$, we write $\tau \preceq \sigma$ to
mean that the triangular inequality between $\symid_{n}, \tau$, and $\sigma$ is an equality:
\begin{equation}\label{eq:def-preceq-permutation}
  d(\symid_{n}, \tau) + d(\tau, \sigma) = d(\symid_{n}, \sigma).
\end{equation}
Similarly, we write $\bm{\tau} \preceq \bm{\sigma}$ to mean
\begin{equation}\label{eq:def-preceq-fam-perm}
  d(\bm{\symid_{n}}, \bm{\tau}) + d(\bm{\tau}, \bm{\sigma}) = d(\bm{\symid_{n}}, \bm{\sigma}).
\end{equation}

If $\sigma, \tau\in \Sym_{n}$, the genus $g$ of $(\sigma, \tau)$ is defined by Euler's relation:
\begin{equation}\label{eq:Euler}
  \#(\sigma\tau^{-1} ) + \#(\sigma) + \#(\tau) -n = 2 K(\sigma, \tau) - 2 g(\sigma, \tau)\;.
\end{equation}

\begin{remark}\label{rem:preceq-nc}
  As in \cite[Equation (2.6)]{collins_free_2025}, we notice that
  \eqref{eq:Euler} may be rewritten as
  \begin{equation}
  \label{eq:triangular-genus}
    d(\bm{\sigma}, \bm{\tau})  + d(\bm{\tau}, \bm{\symid_{n}}) - d(\bm{\symid_{n}}, \bm{\sigma}) = 2\sum_{c = 1}^{D}\Bigl(g(\sigma_{c}, \tau_{c}) + \# \sigma_{c} - K(\sigma_{c}, \tau_{c})\Bigr).
  \end{equation}

  Hence, $\bm{\tau} \preceq \bm{\sigma}$ is equivalent to having for all
  $1 \leq c \leq D$,
  \begin{equation*}
    g(\sigma_{c}, \tau_{c}) = 0 \quad \text{ and } \quad \Pi(\tau_{c}) \leq \Pi(\sigma_{c}).
  \end{equation*}
\end{remark}

\paragraph{Other non-negative combinatorial quantities.}
Consider $\pi, \pi', \tilde \pi\in \Partition(n)$ with $ \tilde{\pi} \leq \pi$
and $\tilde{\pi} \leq \pi'$. We use the quantity (already introduced in
\cite{lionni_higher_2022}):
\begin{equation}
\label{eq:def-of-L}
    L\bigl[\pi, \pi'; \tilde \pi\bigr] = \#(\tilde \pi) - \#( \pi) - \#(\pi') + \#(\pi \vee \pi') \ge 0\;.
\end{equation}
The inequality is sharp, in the sense that for every $\pi \leq \pi'$, there exist some $\tilde{\pi} \leq \pi$ such that $\pi, \pi', \tilde \pi$ saturate the inequality.

Similarly if
$\pi_1, \ldots, \pi_D, \pi', \tilde \pi_1, \ldots, \tilde \pi_D\in \Partition(n)$ satisfying
for any $1\le c\le D$ the relations $\tilde{\pi}_{c} \leq \pi_c$ and
$\tilde \pi_{c} \leq \pi'$, then as remarked in \cite{collins_tensor_2023a}:
\begin{equation}
\label{eq:def-of-LD}
    L_D\bigl[\{\pi_c\}, \pi'; \{\tilde \pi_c\}\bigr] = \sum_{c=1}^D \bigl(\#(\tilde \pi_c) - \#(\pi_c)\bigr) - \#(\pi') + \#( \pi' \vee \pi_1 \vee \ldots \vee \pi_D) \ge 0\;.
\end{equation}
The inequality is sharp in the same sense as above.

The lower bound in \eqref{eq:def-of-LD} may be upgraded in particular case.
\begin{lemma}\label{lem:particular-LD}
  Let $n \in \N^{*}$, $\eta \in \Sym_{n}$, and $\bm{\sigma} \in \Sym_{n}^{D}$. We have that
  \begin{equation*}
    L_{D}\Bigl(\{\Pi(\sigma_{c}, \eta)\}, \Pi(\bm{\sigma}); \{\Pi(\sigma_{c})\}\Bigr) \geq (D - 1)(K(\bm{\sigma}) - K(\bm{\sigma}, \eta)).
  \end{equation*}
\end{lemma}
\begin{proof}
  We are going to show that for all $1 \leq c \leq D$, we have
  \begin{equation*}
    \# \sigma_{c} - \# \Pi(\sigma_{c}, \eta) \geq K(\bm{\sigma}) - K(\bm{\sigma}, \eta).
  \end{equation*}
  This directly implies the result.

  Let $\eta_{0} = \symid, \eta_{1}, \ldots, \eta_{l} = \eta$ where $l = |\eta|$ and
  $|\eta_{i}\eta_{i+1}^{-1}| = 1$. Fix $c \in [D]$. For all $i$, if
  $\Pi(\eta_{i}\eta_{i+1}^{-1}) \leq \Pi(\sigma_{c}, \eta_{i})$ then
  \begin{equation*}
    \# \sigma_{c} - \# \Pi(\sigma_{c}, \eta_{i+1}) = \# \sigma_{c} - \# \Pi(\sigma_{c}, \eta_{i})
  \end{equation*}
  and
  \begin{equation*}
    K(\bm{\sigma}) - K(\bm{\sigma}, \eta_{i+1}) = K(\bm{\sigma}) - K(\bm{\sigma}, \eta_{i}).
  \end{equation*}
  Otherwise,
  \begin{equation*}
    \# \sigma_{c} - \# \Pi(\sigma_{c}, \eta_{i+1}) = \# \sigma_{c} - \# \Pi(\sigma_{c}, \eta_{i}) +1
  \end{equation*}
  and
  \begin{equation*}
    K(\bm{\sigma}) - K(\bm{\sigma}, \eta_{i+1}) - K(\bm{\sigma}) - K(\bm{\sigma}, \eta_{i}) \in \{0, 1\},
  \end{equation*}
  depending on whether $\Pi(\eta_{i}\eta_{i+1}^{-1}) \leq \Pi(\bm{\sigma}, \eta_{i})$ or not. We get
  \begin{equation*}
    \begin{split}
    \# \sigma_{c} - \# \Pi(\sigma_{c}, \eta)
    &= \sum_{i=0}^{l-1}\Bigl(\# \sigma_{c} - \# \Pi(\sigma_{c}, \eta_{i+1}) - \# \sigma_{c} + \# \Pi(\sigma_{c}, \eta_{i}))\\
    &\geq \sum_{i=0}^{l-1}\Bigl(K(\bm{\sigma}) - K(\bm{\sigma}, \eta_{i+1}) - K(\bm{\sigma}) + K(\bm{\sigma}, \eta_{i}))
    = K(\bm{\sigma}) - K(\bm{\sigma}, \eta)\qedhere
    \end{split}
  \end{equation*}
\end{proof}

\paragraph{Moebius functions.}

We will constantly make use of the Moebius inversion in the lattice of
partitions $(\Partition(n), \leq)$. The Moebius function on the lattice of partitions is
\begin{equation}\label{eq:def-moebius-partition}
  \mu_\pi=(-1)^{\#(\pi)-1}(\#(\pi) - 1)! \quad \text{ for } \pi \in \Partition(n).
\end{equation}
It enjoys the inversion property
\begin{equation}\label{eq:moebius-inversion-partition}
  \sum_{\substack{\pi \in \Partition(n)\\\pi_{1} \leq \pi \leq \pi_{2}}}\Bigl(\prod_{S \in \pi_{2}}\mu_{\pi\vert_{S}}\Bigr) = \delta_{\pi_{1}, \pi_{2}}.
\end{equation}


\begin{remark}\label{rem:lattice-Pp}
  A priori the lattice $(\Partition_{\pure}(n), \leq)$ is a
  sub-lattice of the lattice of partitions of $[n] \cup [\bar{n}]$.
  Hence, its Moebius function is a priori different. In the sequel,
  the functions that depend on a partition $\Pi$ of $[n] \cup [\bar{n}]$
  will always be zero if $\Pi \notin \Partition_{\pure}(n)$.
  Furthermore, if $\Pi' \in \Partition_{\pure}(n)$ satisfies
  $\Pi' \leq \Pi$, then $\Pi \in \Partition_{\pure}(n)$. It implies in
  particular that if $\Pi' \in \Partition_{\pure}(n)$,
  \begin{equation*}
    \sum_{\substack{\Pi \in \Partition_{\pure}(n)\\\Pi' \leq \Pi}}\mu_{\Pi_{[n]}} = \sum_{\substack{\Pi \text{ partition of }[n] \cup [\bar{n}]\\\Pi' \leq \Pi}}\mu(\Pi) = \delta_{\Pi', 1_{n, \bar{n}}}.
  \end{equation*}
\end{remark}

A related function to $\mu$ is the Moebius function $\Mnc$ of the
lattice $\NC(n)$ of non-crossing partitions of $n$ elements. We do not define this
lattice as it will not be used, but still define the function
\begin{equation}\label{eq:def-moeb-nc}
  \Mnc(\sigma) = \prod_{\tilde{\sigma} \text{ cycle of } \sigma}(-1)^{\ell(\tilde{\sigma}) - 1}\Cat_{\ell(\tilde{\sigma}) - 1} \quad \text{ for } \sigma \in \Sym_{n},
\end{equation}
where $\Cat_{n} = \frac{1}{n+1}\binom{2n}{n}$ is the $n$-th Catalan number. This function can be extended to $\Sym_{n}^{D}$ by
\begin{equation}\label{eq:def-moeb-nc-family}
  \Mnc(\bm{\sigma}) = \prod_{c = 1}^{D}\Mnc(\sigma_{c}) \quad \text{ for } \bm{\sigma} \in \Sym_{n}^{D}.
\end{equation}

\subsection{Trace-invariants and moments of LU-invariant random tensors}
\label{sub:trace-inv}

\paragraph{Indices of tensors}

We distinguish between mixed and pure tensors. A mixed tensor $A$ is a
collection of complex numbers
$A = (A_{i_{1}, \ldots, i_{D};j_{1}, \ldots, j_{D}})$, where for
$1 \leq c \leq D, i_{c}, j_{c} \in [N_{c}]$ and $N_{1}, \ldots, N_{D} \in \N^{*}$. A pure tensor is a collection of
complex numbers
$(T, \bar{T}) = (T_{i_{1}, \ldots, i_{D}}, \bar{T}_{j_{1}, \ldots, j_{D}})$,
where for $1 \leq c \leq D, i_{c}, j_{c} \in [N_{c}]$. In most cases
$\bar{T}_{i_{1}, \ldots, i_{D}}$ is the complex conjugate of
$T_{i_{1}, \ldots, i_{D}}$ but it needs not be the case in general.

Except in Section \ref{sec:gauss-cov}, we will assume
$N_{1} = \cdots = N_{D} = N$. This can be done without loss of
generality as a $N_{1} \times \cdots \times N_{D}$ tensor can always
be embedded in a $N \times \cdots \times N$ tensor, up to adding
zeroes.

As a shortened notation, we will often write
$A_{\bm{i};\bm{j}} = A_{i_{1}, \ldots, i_{D};j_{1}, \ldots, j_{D}}$
and $T_{\bm{i}} = T_{i_{1}, \ldots, i_{D}}$ with
$\bm{i} = (i_{1}, \ldots, i_{D}), \bm{j} = (j_{1}, \ldots, j_{D}) \in [N]^{D}$.
A family with $n \in \N^{*}$ such multi-indices will be given by a
function $\bm{i} = (i_{1}, \ldots, i_{D}) \colon [n] \to [N]^{D}$. We define the composition of
a $D$-tuple of permutations $\bm{\sigma} \in \Sym_{n}^{D}$ with such a
function by
\begin{equation}\label{eq:def-composition}
  \bm{i} \circ \bm{\sigma} = (i_{1} \circ \sigma_{1}, \ldots, i_{D} \circ \sigma_{D}).
\end{equation}
Given a family $\bm{A} = (A^{(1)}, \ldots, A^{(n)})$ of mixed tensors or $(\bm{T}, \overline{\bm{T}}) =  \bigl( (T^{(1)}, \bar{T}^{(1)}), \ldots, (T^{(n)}, \bar{T}^{(n)})\bigr)$ of pure tensors, we let
\begin{equation}\label{eq:short-product-entries}
  \bm{A}_{\bm{i};\bm{j}} = \prod_{k = 1}^{n}A^{(k)}_{\bm{i}(k);\bm{j}(k)} \quad \text{ or } \quad \bm{T}_{\bm{i}} = \prod_{k=1}^{n}T^{(k)}_{\bm{i}(k)}.
\end{equation}

\paragraph{Labeled trace-invariants.}

Let $n \in \N^{*}$ and $\bm{\sigma} \in \Sym_{n}^{D}$. The trace-invariants
of a family of mixed tensors $\bm{A} = (A^{(1)}, \ldots, A^{(n)})$ or
of a family of pure tensors $\bm{T} = (T^{(1)}, \ldots, T^{(n)})$ are
defined using notation \eqref{eq:short-product-entries} by
\begin{align}\label{def:trace-invariants}
  \Tr_{\bsig}(\bm{A})&= \sum_{\bm{j} \colon [n] \to [N]^{D}} \bm{A}_{\bm{j} \circ \bm{\sigma}; \bm{j}} =  \sum_{\bm{i}, \bm{j} \colon [n] \to [N]^{D}} \
                       \prod_{k=1}^n A^{(k)}_{\bm{i}(k); \bm{j}(k)}  \delta_{\bm{i}(k), \bm{j} \circ \bm{\sigma}(k)}\; ,\\
  \Tr_{\bsig}(\bm{T}, \bar{\bm{T}})&= \sum_{\bm{j} \colon [n] \to [N]^{D}} \bm{T}_{\bm{j}\circ \bm{\sigma}} \  \overline{\bm{T}}_{\bm{j}} =\sum_{\bm{i} \colon [n] \to [N]^{D}} \
                                     \prod_{k=1}^n T^{(k)}_{\bm{i}(k)}\bar T^{(k)}_{ \bm{j}(k)}  \delta_{\bm{i}(k), \bm{j} \circ \bm{\sigma}(k)}
                                     \; .
\end{align}
We write $\Tr_{\bm{\sigma}}(A)$ and $\Tr_{\bm{\sigma}}(T, \bar{T})$
when considering families of identical tensors. They are
generalizations of the usual trace of a matrix, recovered by taking
$D = 1$ and $\bm{\sigma} = (\sigma_{1})$ with one cycle in the mixed
trace-invariant. Going further, trace-invariants are generalizations
of the products of traces of powers of a matrix in the mixed case.
This is recovered by taking $D =
1$. 

In the sequel, if $S\subset[n]$, we will sometimes abuse notation and
consider $\Tr_{\bm{\sigma}_{\vert_{S}}}(T, \bar{T})$. By this we mean the
trace-invariant associated to the $D$-tuple of permutations
$(g_{S, 1}^{-1} \circ \sigma_{1} \circ f_{S}, \ldots, g_{S, D}^{-1} \sigma_{D} \circ f_{S})$
where $f_{S}$ is the strictly increasing mapping $[\# S] \to S$ and
$g_{S, c}$ is the strictly increasing mapping $[\# S] \to \sigma_{c}(S)$.
We use the analogous notation in the mixed case.

\paragraph{Graphical representation.}
  A trace-invariant can be represented graphically as a colored graph
  with labelled vertices. Given $\bm{\sigma} \in \Sym_{n}^{D}$, the
  graph has vertex set $[n] \cup [\bar{n}]$ and each vertex is
  incident to $D$ colored edges as follows. For all $1 \leq c \leq D$
  and $v \in [n]$, there is an edge of color $c$ between the vertices
  $v$ and $\sigma_{c}(v)$. Note that the graph is bipartite: the
  vertices labelled by elements of $[n]$ are connected to vertices
  labelled by $[\bar{n}]$ only.

\begin{figure}[!ht]
\centering
\includegraphics[scale=0.7]{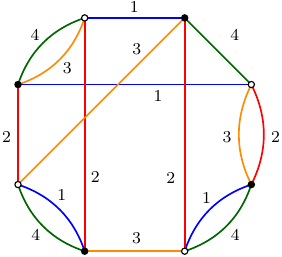}\hspace{3cm}\includegraphics[scale=1.2]{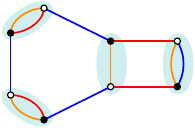}
\caption{Left: a 4-colored graph. Right: a melonic 3-colored graph (the light-blue regions represent the canonical pairs).
}
\end{figure}

When considering trace-invariants of a family of identical tensors,
the family of all trace-invariants
$(\Tr_{\bm{\sigma}})_{\bm{\sigma} \in \Sym_{n}^{D}}$ contains redundant
observables. Indeed, relabeling the vertices appearing in the
graphical representation yields identical quantities. We thus define
equivalence relations corresponding to this invariance by relabeling.

We define the two equivalence relations $\sim_\mix$ and  $\sim_\mathrm{p}$ on $\Sym_{n}^{D}$ as:
\begin{equation}
\label{eq:relabeling}
    \begin{split}
        \bsig  \sim_\mix \bsig' \qquad &\Leftrightarrow \qquad \exists \eta\in \Sym_{n}, \quad \bsig = \eta \bsig ' \eta^{-1}\;,\\
                \bsig  \sim_\mathrm{p} \bsig' \qquad &\Leftrightarrow \qquad \exists \eta, \nu\in \Sym_{n}, \quad \bsig = \eta \bsig ' \nu\;.
    \end{split}
\end{equation}

The sequence of numbers
$(\E \Tr_{\bm{\sigma}}(A) \colon \bm{\sigma} \in \Sym_{n}^{D}/\sim_{\mix})$ and
$(\E \Tr_{\bm{\sigma}}(T, \bar{T}) \colon \bm{\sigma} \in \Sym_{n}^{D}/\sim_{\pure})$
play the role of macroscopic moments.

Given a $D$-tuple of permutation $\bm{\sigma} \in \Sym_{n}^{D}$, its
automorphism group for the relations $\sim_{\mix}$ and $\sim_{\pure}$
are $\Aut_{\mix}(\bm{\sigma})$ and $\Aut_{\pure}(\bm{\sigma})$, defined by
\begin{equation}\label{eq:def-automorphism}
  \Aut_{\mix}(\bm{\sigma}) = \{\eta \in \Sym_{n} \colon \bm{\sigma} = \eta \bm{\sigma} \eta^{-1}\}
  \quad \text{ and } \quad
  \Aut_{\pure}(\bm{\sigma}) = \{(\eta, \nu) \in \Sym_{n}^{2} \colon \bm{\sigma} = \eta \bm{\sigma} \nu^{-1}\}.
\end{equation}
Hence, by the orbit-stabilizer theorem, the number of $D$-tuples of
permutations $\bm{\sigma'}$ such that $\bm{\sigma} \sim_{\mix} \bm{\sigma'}$ is
$\frac{n!}{\# \Aut_{\mix}(\bm{\sigma})}$, and the number of $\bm{\sigma'}$ such
that $\bm{\sigma} \sim_{\pure} \bm{\sigma'}$ is
$\frac{(n!)^{2}}{\# \Aut_{\pure}(\bm{\sigma})}$.

\paragraph{Melonic graphs.}A family of trace-invariants play a distinguished role in the study of
random tensors: trace-invariants related to melonic graphs (see for
instance \cite{bonzom_diagrammatics_2017,fusy_combinatorial_2020}).
Those graphs are defined recursively as follows:
\begin{enumerate}
  \item\label{itm:single-melon} either it is the only graph with vertices labelled by $1$ and
  $\bar{1}$ and $D$ colored edges between them;
  \item\label{itm:many-melon} or it is a colored bipartite graph as introduced above, which contains at least two vertices $v$
        and $\bar{v}$ connected by $D-1$ edges, and such that when
        removing $v$ and $\bar{v}$ and connecting the two edges
        connected to other vertices, we obtain a melonic graph.
\end{enumerate}
This recursive definition makes it clear that there is a distinguished
pairing of vertices labelled by $[n]$ with vertices labelled by
$[\bar{n}]$: in case \ref{itm:single-melon}, it is the pairing
$1 \mapsto \bar{1}$, and in case \ref{itm:many-melon} it is the
pairing defined recursively as being the one sending the label of $v$
to the label of $\bar{v}$. Such a pairing defines a permutation
$\eta \in \Sym_{n}$ called the canonical pairing of the melonic graph.

The following quantity will play an important role\footnote{It is well-defined, since for every $n \in \N^{*}$ and $\bm{\rho} \in \Sym_{n}$, there exists  $\eta \in \Sym_{n}$ such that
$    K_{\pure}(\bm{\rho}, \eta) = K(\bm{\rho}\eta^{-1}) = 1.$ Indeed, letting $\gamma = \cycle{1, 2, \ldots, n}$, the permutation $\eta = \gamma^{-1}\rho_{1}$  satisfies $K(\rho_{1}\eta^{-1}) = 1$ and thus  $K_{\pure}(\bm{\rho}, \eta) = K(\bm{\rho}\eta^{-1}) = 1$.
}:
\begin{equation}
\label{eq:scaling-gauss}
  \delta(\btau) = \min_{\substack{\eta \in \Sym_{n}\\K_{\pure}(\btau, \eta) = 1}}d(\bm{\sigma}, \eta),
\end{equation}

\begin{theorem}
\label{thm:degree}
Consider $\bsig\in \Sym_{n}^{D}$. Then, one has that
\begin{equation}
    \Omega(\bsig) = K(\bsig) + n(D-1) - \# \bm{\sigma} \ge 0,
\end{equation}
with equality if and only if $\bsig$ is melonic and the vertices in the canonical pairs have the same label (i.e.~the canonical pairing is the identity). Similarly (see e.g. Thm.~5.3 and Thm.~5.7 of  \cite{collins_free_2025}):
\begin{equation}
   \delta(\bsig) \ge n-1 + (D-1) (K_\pure(\bsig) - 1)\;,
\end{equation}
with equality if and only if $\bsig$ is melonic. Furthermore, if $\bsig$ is melonic, its canonical pairing is the unique $\eta\in \Sym_n$ such that $d(\bsig, \eta)=\delta(\bsig)$ and $K_\pure(\bsig, \eta) = K_\pure(\bsig)$.
\end{theorem}

The $\bsig\in \Sym_{n}^{D}$ for which $\Omega(\bsig)$ takes a fixed positive
value can also be characterized and counted asymptotically, see
\cite{fusy_combinatorial_2020}.


\subsection{Classical cumulants}
The classical cumulants of a family $(x_{i})_{i\geq1}$ are defined using the Moebius function of the lattice of partition \eqref{eq:def-moebius-partition} through the formula
\begin{equation}\label{eq:class-cum}
  k_p(x_1, \ldots, x_p) = \sum_{\pi \in \mathcal{P}(p)} \mu_\pi \prod_{B\in \pi} \E\bigl[\prod_{i\in B}x_i\bigr].
\end{equation}
The classical cumulants enjoy the inverse relations:
\begin{equation}\label{eq:classical-mom-cum}
  \E[x_1 \cdots x_p] = \sum_{\pi \in \mathcal{P}(p)}  \prod_{B\in \pi} k_{\# B}(\{x_i\}_{i\in B}).
\end{equation}
We now define the cumulants in the case of mixed and pure tensors.

\paragraph{Mixed version:} We denote the classical cumulants of $p$ connected trace-invariants:
\begin{equation}
  \Phim_{ \bsig} [A]=k_p \bigl(\Tr_{\bsig^{(1)}}(A), \ldots, \Tr_{\bsig^{(p)}}(A)\bigr),
\end{equation}
where $\bm{\sigma}^{(1)}, \ldots,
\bm{\sigma}^{(p)}$ are the connected components of $\bsig$, i.e. if $\Pi(\bm{\sigma})
= \{S_1, \ldots, S_p\}$, we have for all $1\leq i \leq p,
\bm{\sigma}^{(i)} = \bm{\sigma}\vert_{S_i}$. In particular $K(\bm{\sigma})
= p$. From \eqref{eq:class-cum}, one has:
\begin{equation}
\label{eq:class-cum-trace-ind}
    \Phim_{\bsig} [A]=\sum_{\pi \ge \Pi(\bsig)}\mu_\pi \prod_{S\in \pi}\E\bigl[\Tr_{\bsig\vert_{S}}(A)\bigr].
\end{equation}

Note that if $A'$ is a deterministic tensor and $A=A'$ or $A=UA'U^\dagger$ with $U=U_1\otimes \cdots \otimes U_D$, $U_c$ Haar distributed, one has
\begin{equation}
\label{eq:mat-deter-corr}
    \Phim_{\bsig} [A] = \delta_{\Pi(\bsig), 1_n}  \Tr_{\bsig} (A')\;.
\end{equation}
The multiplicative extension
$\Phim_{\pi, \bm{\sigma}}[A] = \prod_{S \in \pi}\Phim_{\bm{\sigma}\vert_{S}}[A]$
satisfies
\begin{equation}
\Phim_{1_n, \bsig} [A] = \Phim_{\bsig} [A] \quad \text{ and } \quad  \Phim_{\Pi(\bsig), \bsig}[A] = \prod_{i=1}^p \Phim_{\bsig^{(i)}}[A].
\end{equation}

\paragraph{Pure version:}
In a similar fashion as in the mixed case, the classical cumulants of
a pure random tensor are defined for $\bm{\sigma} \in \Sym_{n}^{D}$ by
\begin{equation}\label{eq:def-Phip}
  \Phip_{\bsig} [T, \bar T] = k_p \bigl(\Tr_{\bsig^{(1)} } (T, \bar T), \ldots, \Tr_{\bsig^{(p)} } (T, \bar T) \bigr),
\end{equation}
where now the $\bsig^{(i)}$ are the \emph{pure} connected components of
$\bsig$: writing $\Pi_{\pure}(\bm{\sigma}) = \{S_{1}, \ldots, S_{p}\}$, we
set $\bm{\sigma}^{(i)} = \bm{\sigma}\vert_{S_{i}\cap [n]}$. The multiplicative
extension
$\Phip_{\Pi, \bsig} [T,\bar T] = \prod_{S\in \Pi_{[n]}} \Phip_{ \bsig\vert_{S}} [T,\bar T]$
satisfies
\begin{equation}
  \Phip_{1_{n, \bar n}, \bsig} [T,\bar T] = \Phip_{\bsig} [T,\bar T] \quad \text{ and } \quad \Phip_{\Pi_{\pure}(\bsig), \bsig}[T,\bar T] = \prod_{i=1}^p \Phip_{\bsig^{(i)}}[T,\bar T].
\end{equation}

\subsection{Examples of LU-invariant random tensors}

Two kinds of examples of LU-invariant random tensors naturally come to mind:

\begin{itemize}
\item Tensors taken uniformly at random in the LU-orbit of a deterministic tensor, of the form $A=UA'U^\dagger$ in the mixed case, where $A'$ is a deterministic mixed tensor with $D$ inputs and $U=U_1\otimes \cdots \otimes U_D$, $U_c$ Haar distributed, and  $T=UT'$ in the pure case, where  $T'$ is a deterministic pure tensor with $D$ inputs.

\item Unitarily invariant random matrices (GUE, Wishart...) or vectors (Gaussian, Haar distributed...) with subdivided index sets.

{\bf Standard complex Gaussians:} Fixing the number of indices $D$, the example that has been studied the most in the literature is the pure standard complex Gaussian -- or pure  ``Ginibre'' -- random tensor  $(T_\un, \bar T_\un)$, where the components of $T_\un$ are i.i.d.~centered Gaussian distributed complex random variables of variance $1/N^{D-1}$:
\begin{equation}
\E\bigl[(T_\un)_{i_1, \ldots, i_D}(T_\un)_{j_1, \ldots, j_D} \bigr] = N^{1-D}\prod_{c=1}^D \delta_{i_c, j_c}\;.
\end{equation}
One has in that case for any $\bsig \in \Sym_n^D$, $n\in \N^{*}$, and $\bm{i}, \bm{j} \colon [n] \to [N]^{D}$ (see e.g.~\cite{gurau_universality_2014}):
\begin{equation}
\label{eq:Pure-C-Gaussian}
\begin{split}
  \E\left[\prod_{k=1}^n (T_\un)_{\bm{i}(k)} (\bar T_\un)_{\bm{j}(k)}\right] &= N^{n(1-D)} \sum_{\eta\in \Sym_{n}}  \prod_{k = 1}^n\delta_{\bm{i}(k), \bm{j}(\eta(k))}\;,\\[+1ex]
  \E[\Tr_\bsig(T_\un, \bar T_\un)] &= N^{n(1-D)} \sum_{\eta\in \Sym_{n}} N^{\sum_{c=1}^D \#(\sigma_c \eta^{-1})}\;,\\
  \Phip_\bsig[T_\un, \bar T_\un] &= N^{n(1-D)} \sum_{\substack{{\eta\in \Sym_{n}}\\{K_p(\bsig, \eta)=1}}} N^{\sum_{c=1}^D \#(\sigma_c \eta^{-1})}\;.
 \end{split}
\end{equation}

{\bf Wishart tensors:} We will call Wishart tensor of parameters $(N^D, N^{D'})$ a Wishart random matrix of parameters $(N^D, N^{D'})$  whose indices are subdivied in $D$ inputs and $D$ outputs. More precisely, given a standard complex Gaussian $(T_\un, \bar T_\un)$ as above with $D+D'$ inputs, then the coefficients of the Wishart tensor $W = (W_{\bm{i}; \bm{j}})$ of parameters $(N^D, N^{D'})$ are:
\begin{equation}
W_{\bm{i}; \bm{j}} = \sum_{\bm{i}'\in [N]^{D'}}(T_\un)_{\bm{i},\bm{i}'} (\bar{T}_\un)_{\bm{j},\bm{i}'}, \end{equation}
where $\bm{i}, \bm{j}\in [N]^D$.
Such Wishart tensors have been discussed in the case $D'=1$ in \cite[Section 5.5]{collins_free_2025}. We get for any $\bsig \in \Sym_n^D$, $n\in \N^{*}$, and $\bm{i}, \bm{j} \colon [n] \to [N]^{D}$:
\begin{equation}
\label{eq:Wishart}
\begin{split}
  \E\Bigl[\prod_{k=1}^n W_{\bm{i}(k); \bm{j}(k)}\Bigr]
  &= N^{n(1-D)} \sum_{\eta\in \Sym_{n}}  N^{-D'|\eta|} \ \prod_{k = 1}^n\delta_{\bm{i}(k), \bm{j}(\eta(k))}\;,\\[+1ex]
  \E[\Tr_\bsig(W)] &= N^{n(1-D)} \sum_{\eta\in \Sym_{n}} N^{\sum_{c=1}^D \#(\sigma_c \eta^{-1}) - D'|\eta|}\;,\\
  \Phip_\bsig[W] &= N^{n(1-D)} \sum_{\substack{{\eta\in \Sym_{n}}\\{K(\bsig, \eta)=1}}} N^{\sum_{c=1}^D \#(\sigma_c \eta^{-1}) - D'|\eta|}\;.
 \end{split}
\end{equation}
\end{itemize}

This notion of Wishart tensor extends to different parameters. Given $D, D' \geq 1$, $N \in \N^*$, and $\bm{N'} = (N'_1, \ldots, N'_{D'}) \in (\N^*)^{D'}$, a Wishart tensor $W$ of parameters $(N^D, \bm{N'})$ is defined as follows. Let $(T_\un, \bar{T}_\un)$ be a standard Gaussian tensor of size $N \times \cdots \times N \times N'_1 \times \cdots \times N'_{D'}$ with $D + D'$ inputs, i.e.\ a tensor of this size with i.i.d.\ complex Gaussian entries of variance $N^{1-D-D'}$. We set
\begin{equation*}
    W_{\bm{i};\bm{j}} = \sum_{\bm{i'} \in [N'_1] \times \cdots \times [N'_{D'}]}(T_\un)_{\bm{i},\bm{i'}}(\bar{T}_\un)_{\bm{j},\bm{i'}}, \quad \text{ where } \bm{i},\bm{j} \in [N]^{D}.
\end{equation*}
For convenience, we set $r_{c'} = \frac{N'_c}{N}, 1 \leq c' \leq D'$.
For any $n \in \N^*$, $\bm{\sigma} \in \Sym_n^D$, and $\bm{i}, \bm{j} \colon [n] \to [N]^{D}$:
\begin{equation}
\label{eq:Wishart-multi}
\begin{split}
  \E\Bigl[\prod_{k=1}^n W_{\bm{i}(k); \bm{j}(k)}\Bigr]
  &= N^{n(1-D)} \sum_{\eta\in \Sym_{n}} \Bigl(\prod_{c'=1}^{D'}r_c\Bigr)^{\# \eta} \ N^{-D'|\eta|} \ \prod_{k = 1}^n\delta_{\bm{i}(k), \bm{j}(\eta(k))}\;,\\[+1ex]
  \E[\Tr_\bsig(W)] &= N^{n(1-D)} \sum_{\eta\in \Sym_{n}} \Bigl(\prod_{c'=1}^{D'}r_c\Bigr)^{\# \eta} \ N^{\sum_{c=1}^D \#(\sigma_c \eta^{-1}) - D'|\eta|}\;,\\
  \Phip_\bsig[W] &= N^{n(1-D)} \sum_{\substack{{\eta\in \Sym_{n}}\\{K(\bsig, \eta)=1}}} \Bigl(\prod_{c'=1}^{D'}r_c\Bigr)^{\# \eta} \ N^{\sum_{c=1}^D \#(\sigma_c \eta^{-1}) - D'|\eta|}\;.
 \end{split}
\end{equation}

\subsection{Random tensors in the limit of infinite size}
\label{subsub:asympt-def-of-distrib}

In Random Matrix Theory, the limiting distribution as the size $N$ of
a random matrix $X_{N}$ goes to infinity is defined through the limit
of the moments, appropriately renormalized
$(\frac{1}{N} \E \Tr X_{N}^{k})_{k \geq 0}$. A richer description of the
$N \to \infty$ limit is provided by considering the leading
asymptotics of the cumulants of the renormalized traces, i.e.\ the quantities
\begin{equation}\label{eq:cumulants-RMT}
  k_{n}\Bigl(\frac{1}{N}\Tr X_{N}^{k_{1}}, \ldots, \frac{1}{N}\Tr X_{N}^{k_{n}}\Bigr)
\end{equation}
for $n \geq 0$ and $k_{1}, \ldots, k_{n} \geq 0$.

Here, we study the analogous quantities for random tensors. The
analogous quantities to the one in \eqref{eq:cumulants-RMT} are the
$(\Phim_{\bm{\sigma}}[A])_{n \in \N^{*}, \bm{\sigma} \in \Sym_{n}^{D}}$ in the
mixed case and
$(\Phip_{\bm{\sigma}}[T, \bar{T}])_{n \in \N^{*}, \bm{\sigma} \in \Sym_{n}^{D}}$
in the pure case. Contrary to the matrix case, it is not clear in the
tensor case how one is to renormalize the classical cumulants to
obtain the sequence of numbers defining the distribution in the limit
of infinite size.

A way to choose the normalization is to follow \cite{collins_second_2007} and \cite{collins_free_2025} and define the
distribution asymptotically at first order by the data of the dominant
contributions in $N$, of the classical cumulants $\Phip_\bsig$ that
scale the strongest in $N$. Given a sequence
$(T_{N}, \bar{T}_{N})_{N \geq 1}$ of pure or $(A_{N})_{N \geq 1}$ of mixed
random tensor, if $\Phip_\bsig(A_{N})$ or
$\Phip_{\bm{\sigma}}(T_{N}, \bar{T}_{N})$ is non-zero for $N$ big enough
for this distribution, we let:
\begin{equation}\label{eq:def-asymptotic-cum}
  \Phim_\bsig[A_{N}] = N^{r(\bsig)}\Bigl( \vphim_\bsig(a) + o(1)\Bigr) \quad \text{ or } \quad \Phip_\bsig[T_{N}, \bar{T}_{N}]= N^{r(\bsig)}\Bigl( \vphip_\bsig(t, \bar{t}) +o(1)\Bigr)\;,
\end{equation}
where either $\vphim_\bsig(a) \neq 0$ in the mixed case or
$\vphip_{\bm{\sigma}}(t, \bar{t}) \neq 0$ in the pure case is
independent on $N$. In the sequel, we relax the requirement that those
quantities are non-zero, and consider tensors that scale \emph{at
  most} as $N^{r(\bm{\sigma})}$, i.e.\ such that the scaling assumption
\eqref{eq:def-asymptotic-cum} is satisfied with $\vphim$ or $\vphip$ a finite
quantity that does not depend on $N$ but may or may not be zero. If the coefficient $\vphim_\bsig(a)$ or $\vphip_{\bm{\sigma}}(t, \bar{t})$ is non-vanishing, we will say that the scaling $r$ is \emph{sharp} for this distribution. Note
that in the sequel, we will omit the index $N$ keeping track of the
size of the tensors.

These quantities $\varphi_{\bm{\sigma}}$ are the \emph{asymptotic
  cumulants} of the sequences of random tensors. With this
terminology, we depart from the one of \cite{collins_second_2007} in
which these quantities where called asymptotic moments. This notion of
asymptotic cumulants depend on the scaling $r$. We shall take as
scaling exponent natural quantities that appear when considering
tensor products of matrices in the mixed case, and Gaussian tensors in
the pure case. Note that the Gaussian scaling will have special
properties and we will be led to consider a second pure scaling, see
Section \ref{sec:two-scalings-pure}.

One expects a sharp scaling $r$ to be bounded on $\Sym_n^D$ by its maximum
$r_\mathrm{M}$ (up to a global rescaling in $N$ of the components of
the tensor). For a given sequence of distributions $(T_{N}, \bar{T}_{N})_{N \geq 1}$ or $(A_{N})_{N \geq 1}$ with sharp\footnote{One can instead make the weaker assumption that for every $n\in \N^{*}$ there is at least one first order $\bsig$ with non-vanishing coefficient.} scaling \eqref{eq:def-asymptotic-cum}, we then define the asymptotic distribution in the limit of infinite size $N$  at first order
by the data of the $\varphi_\bsig$, for $\bsig$ such that
$r(\bsig)=r_\mathrm{M}$.\footnote{One could also make a choice of scaling, and define a notion of asymptotic distribution that depends on this scaling. The precise choice of definition will be unimportant for many results presented in the sequel.} Similarly, we define the distribution at
order $k\ge 1$ by the data of the $\varphi_\bsig$ for $\bsig$
satisfying
$$
r_\mathrm{M}-r(\bsig) \le k-1\;.
$$
For classical random matrix ensembles ($D=1$), $r(\sigma)=2-\#(\sigma)$, $r_\mathrm{M}=1$, obtained for $\sigma$ cyclic (with a single cycle).

One can go further in the definition of distributions of random tensors in the limit of infinite size. For instance, the first order distribution of a pure random tensor for which $r(\bsig)$ is the same as for a complex Gaussian takes value in a generalization of non-commutative probability spaces, see \cite{collins_free_2025}.

\subsection{The Gaussian and Wishart scalings}
\label{sub:Gaussian-scaling}

\paragraph{The Gaussian scaling.}
The moments of a standard complex Gaussian tensor are given in \eqref{eq:Pure-C-Gaussian}.
Theses moments are of order
$N^{n - \delta^{\bullet}(\bm{\sigma})}$ where
\begin{equation}\label{eq:scale-gauss-moment}
  \delta^{\bullet}(\bm{\sigma}) = \min_{\substack{{\eta \in \Sym_{n}}}} d(\bsig, \eta).
\end{equation}

The classical cumulants of such a Gaussian tensor are then of order
$N^{n - \delta(\bm{\sigma})}$ where $\delta$ has been defined in \eqref{eq:scaling-gauss}:
\[
  \delta(\bm{\sigma}) = \min_{\substack{\eta \in \Sym_{n}\\K_{\pure}(\bm{\sigma}, \eta) = 1}}d(\bm{\sigma}, \eta),
\]
and are precisely
\begin{equation}
\label{eq:asympt-cum-gaussienne}
\begin{split}
  &\Phip_{\bm{\sigma}}[T_\un, \bar T_\un]  = N^{n - \delta(\bm{\sigma})}\Bigl( \vphip_{\bm{\sigma}}(t_{1}, \bar{t}_{1}) + o(1) \Bigr),\\
  &\vphip_{\bm{\sigma}}(t_{1}, \bar{t}_{1})  = \# \Bigl\{ \eta \in \Sym_{n} \colon K_{\pure}(\bm{\sigma}, \eta) = 1, \delta(\bm{\sigma}) = d(\bm{\sigma}, \eta)\Bigr\}.
  \end{split}
\end{equation}

From Thm.~\ref{thm:degree}, the first order trace-invariant correspond to purely connected melonic $\bsig$, and in that case, $\vphip_\bsig(t_1, \bar t_1)=1$.

\

The multiplicative extension $\Phip_{\Pi}$ has its $N \to \infty$
asymptotics described in terms of the following quantity, defined for
$\bm{\sigma} \in \Sym_{n}^{D}$:
\begin{equation}
\label{eq:multiplicative-scaling-gaussien}
\delta(\Pi, \bm{\sigma}) = \sum_{B\in \Pi_{[n]}} \delta(\bsig_{\lvert_{B}}) = \min_{\substack{\eta \in \Sym_{n}\\\Pi_{\pure}(\bm{\sigma}, \eta) = \Pi}}d(\bm{\sigma}, \eta),
\end{equation}
and the $\vphip_{\Pi, \bsig}(t_1, \bar t_1)$ count the elements in the sets
\begin{equation}\label{eq:set-scal-gauss}
  \scalg(\Pi, \bm{\sigma}) = \Bigl\{ \eta \in \Sym_{n} \colon \Pi_{\pure}(\bm{\sigma}, \eta) = \Pi, \delta(\Pi, \bm{\sigma}) = d(\bm{\sigma}, \eta) \Bigr\} \quad \text{ for } \Pi \in \PPart(n) \text{ and } \bm{\sigma} \in \Sym_{n}^{D}.
\end{equation}
If $\Pi=1_{n, \bar n}$, we will use the notation $  \scalg(\bm{\sigma})  = \scalg(1_{n, \bar n}, \bm{\sigma}) $.

\paragraph{The Wishart scaling.}The moments of a Wishart tensor of parameters $(N^D, N^{D'})$ are given in \eqref{eq:Wishart}. For $\bsig\in \Sym_n^D$, we let  $\hat\bsig_{D'}$ be $\bsig$ supplemented by $D'$ copies of $\mathrm{id}_n$, that is:
\[
\hat\bsig_{D'} = (\bsig,  \mathrm{\bf id})=(\sigma_1, \ldots, \sigma_D, \mathrm{id}_n,\ldots, \mathrm{id}_n )\in \Sym_{n}^{D+D'}\;.
\]

With these notations, from \eqref{eq:Wishart}, the moments of $W$ are of order $N^{n-\delta^\bullet(\hat\bsig_{D'})}$, and since $K_\pure(\hat\bsig_{D'}, \eta) = K(\bsig, \eta)$, the classical cumulants satisfy:
\begin{equation}
\label{eq:Wishart-asympt}
  \Phip_\bsig[W] = N^{n - \delta(\hat\bsig_{D'})} \vphip_{\hat\bsig_{D'}}(t_{1}, \bar{t}_{1}) \;.
\end{equation}

\begin{corol}[of  Thm.~\ref{thm:degree}]
\label{cor:melonic-wishart}
 Consider $D'\ge 1$ and $\bsig\in \Sym_n^D$ with $K(\bsig)$ fixed to $K$, one has
\[
n-\delta(\hat \bsig_{D'}) \le 1-(D+D'-1)(K - 1)\;,
\]
with equality if and only if $\hat \bsig_{D'}$ is melonic. For $D'\ge 2$, this occurs if and only if $\bsig$ is melonic and  the canonical pairing is the identity.
\end{corol}
\proof The first statements result from Thm.~\ref{thm:degree} together with the fact that $K_\pure(\hat\bsig_{D'})=K(\bsig)$.  If  $D'\ge 2$, $\hat \bsig_{D'}$ is melonic if and only if  $\bsig$ is melonic and the canonical pairing is the identity.  \qed

\

The first order therefore again consists in  melonic $\bsig$, which are purely connected for $D'\ge 1$, but are allowed to be connected but not purely connected for $D'=1$, see Sec.~5.5 of \cite{collins_free_2025}.

\ 

For $D'=D$, we may write the scaling $n-\delta(\hat \bsig_{D'})$ in a way which does not involve a minimum. 

\begin{prop}\label{prop:expr-n-delta-no-min}
Consider $\bsig\in \Sym_n^D$, $n\in \N^{*}$. Then:
\[
n-\delta(\hat \bsig_{D}) = 2D\bigl(1-K(\bsig)\bigr) + \#\bsig - n(D-1).
\]
\end{prop}
The proof relies on the following lemma:
\begin{lemma}
\label{lem:bound-for-wishart}
Consider $n\in \N^{*}$, $\eta\in \Sym_n$ and $\bsig\in \Sym_n^D$, such that $K(\bsig, \eta)=1$. Then:
\begin{equation*}
  d(\bm{\sigma}, \eta)  + d(\eta, \bm{\symid_{n}}) - d(\bm{\symid_{n}}, \bm{\sigma}) \ge 2D\big(K(\bsig) - 1\bigr)\;,
\end{equation*}
and for every $\bsig$, there exists $\eta\in \Sym_n$ such that equality holds and $K(\bsig, \eta)=1$.
\end{lemma}
\begin{proof}[Proof of Proposition \ref{prop:expr-n-delta-no-min}]
Consider $\eta \in \Sym_n$ such that $K(\bsig, \eta)=1$. From Lemma~\ref{lem:bound-for-wishart}:
\begin{equation*}
    \begin{split}
        n- &d(\hat \bsig_{D}, \eta)\\
        &= n- d(\bm{\sigma}, \bm{\symid_{n}}) - 2D\bigl(K(\bsig) - 1\bigr) - \Bigl(d(\bm{\sigma}, \eta) + d(\eta, \bm{\symid_{n}}) - d(\bm{\sigma}, \bm{\symid_{n}}) - 2D\bigl(K(\bsig) - 1\bigr) \Bigr)\\
   &\leq \# \bm{\sigma}   - 2D\bigl(K(\bsig) - 1\bigr) - n(D-1).
    \end{split}
\end{equation*}
Still from Lemma~\ref{lem:bound-for-wishart}, there exists a $\eta\in \Sym$ such that $K(\bsig, \eta)=1$ and the quantity between parenthesis vanishes. This concludes the proof.
\end{proof}
\begin{proof}[Proof of Lemma \ref{lem:bound-for-wishart}]
For $\btau=(\eta, \ldots, \eta)$, \eqref{eq:triangular-genus} reads:
 \[d(\bm{\sigma}, \eta) + d(\eta, \bm{\symid_{n}}) - d(\bm{\sigma}, \bm{\symid_{n}} ) = 2\sum_{c=1}^D g(\sigma_c, \eta) + \sum_{c=1}^D\bigl[\#(\sigma_c) - K(\sigma_c, \eta) \bigr],\]
and using  \eqref{eq:def-of-LD}:
\[
    L_D\Bigl[\{\Pi(\sigma_{c}, \eta)\}, \Pi(\bm{\sigma}); \{\Pi(\sigma_{c})\}\Bigr] = \sum_{c=1}^D \Bigl(\#(\sigma_c) - \#(\Pi(\sigma_c, \eta)\Bigr) - K(\bsig)  + 1 \ge 0\;.
\]
we write:
$$
    d(\bm{\sigma}, \eta)  + d(\eta, \bm{\symid_{n}}) - d(\bm{\symid_{n}}, \bm{\sigma}) =  2\big(K(\bsig) - 1\bigr) +2\sum_{c=1}^D g(\sigma_c, \eta) +  2 L_D\bigl[\{\Pi(\sigma_c, \eta)\}, \Pi(\bsig); \{\Pi(\sigma_c)\}\bigr] \;,
$$
which is lower bounded by $2D(K(\bm{\sigma}) - 1)$ by Lemma
\ref{lem:particular-LD}. It reaches this value if we choose $\eta$ as
follows. We construct it inductively. If $K(\bm{\sigma}) = 1$, then we can
choose $\eta = \symid_{n}$. Otherwise, we are going to construct
permutations $\eta_{0} = \symid_{n}, \eta_{1}, \ldots, \eta_{K(\bm{\sigma}) - 1}$
such that for all $0 \leq i < K(\bm{\sigma})$, $|\eta_{i}| = i$,
$K(\bm{\sigma}, \eta_{i}) = K(\bm{\sigma}) - i$, $g(\sigma_{c}, \eta_{i}) = 0$ for all
$1 \leq c \leq D$, and
\begin{equation*}
  L_D\Bigl[\{\Pi(\sigma_{c}, \eta_{i})\}, \Pi(\bm{\sigma}); \{\Pi(\sigma_{c})\}\Bigr]
  = (D - 1)i.
\end{equation*}
These assumptions are clearly satisfies for $\eta_{0} = \symid_{n}$.
Assume we have constructed $\eta_{i}$ with $i < K(\bm{\sigma}) - 1$. Hence,
$K(\bm{\sigma}, \eta_{i}) \geq 2$. Let $p, q \in [n]$ such that $p$ and $q$
belongs to different blocks of $\Pi(\bm{\sigma}, \eta_{i})$. Set
$\eta_{i+1} = \eta \cycle{p, q}$. We get that $|\eta_{i+1}| = |\eta_{i}| + 1$ and  $K(\bm{\sigma}, \eta_{i+1}) = K(\bm{\sigma}, \eta_{i}) - 1$, that for all $1 \leq c \leq D$,
\begin{equation*}
  \begin{split}
  g(\sigma_{c}, \eta_{i+1})
  &= K(\sigma_{c}, \eta_{i+1}) + \frac{1}{2}\Bigl( n - \# \sigma_{c} - \# \eta_{i+1} - \#(\sigma_{c}\eta_{i+1}^{-1})\Bigr)\\
  &= K(\sigma_{c}, \eta_{i}) - 1 + \frac{1}{2}\Bigl( n - \# \sigma_{c} + 1 - \# \eta_{i+1} - \#(\sigma_{c}\eta_{i+1}^{-1}) + 1\Bigr)
  = g(\sigma_{c}, \eta_{i})
  = 0,
  \end{split}
\end{equation*}
we used that since $p$ and $q$ are in different blocks of $K(\sigma_{c}, \eta_{i})$, we have that $\# \sigma_{c}\eta_{i+1}^{-1} = \# \sigma_{c}\eta_{i}^{-1} - 1$. It remains to show the property about $L_{D}$: we have $\# \Pi(\sigma_{c}, \eta_{i+1}) = \# \Pi(\sigma_{c}, \eta_{i}) - 1$ so that
\begin{equation*}
  L_D\Bigl[\{\Pi(\sigma_{c}, \eta_{i+1})\}, \Pi(\bm{\sigma}); \{\Pi(\sigma_{c})\}\Bigr]
  =   L_D\Bigl[\{\Pi(\sigma_{c}, \eta_{i})\}, \Pi(\bm{\sigma}); \{\Pi(\sigma_{c})\}\Bigr] + D - 1
  = (D - 1)(i + 1).
\end{equation*}
Finally,
\begin{equation*}
  L_D\Bigl[\{\Pi(\sigma_{c}, \eta)\}, \Pi(\bm{\sigma}); \{\Pi(\sigma_{c})\}\Bigr]
  = (D - 1)(K(\bm{\sigma}) - 1),
\end{equation*}
and Lemma \ref{lem:particular-LD} gives that this is the smallest possible value.

Finally, with this choice of $\eta$, we have
\begin{equation*}
  d(\bm{\sigma}, \eta)  + d(\eta, \bm{\symid_{n}}) - d(\bm{\symid_{n}}, \bm{\sigma}) =  2\big(K(\bsig) - 1\bigr) +  2(D - 1)(K(\bm{\sigma}) - 1)
  = 2D(K(\bm{\sigma}) - 1).\qedhere
\end{equation*}
\end{proof}

\subsection{Weingarten functions}

The Weingarten calculus allows the computation of moments of entries
of unitary matrices distributed according to the Haar measure. The
fundamental result is the following.
\begin{theorem}[Weingarten formula, {\cite{weingarten_asymptotic_1978,samuel_integrals_1980,collins_moments_2003}}]\label{thm:Weingarten-formula}
  Let $n, N \in \N^{*}$ and $i, j, i', j' \colon [n] \to [N]$. We have
  \begin{equation*}
    \int_{\Unit(N)} U_{i(1)j(1)} \cdots U_{i(n)j(n)} \overline{U_{i'(1)j'(1)} \cdots U_{i'(n)j'(n)}} \ \dd U
    = \sum_{\rho, \sigma \in \Sym_{n}} \delta_{i, i' \circ \rho}\delta_{j, j' \circ \sigma}\Weingarten_{N}(\sigma\rho^{-1}),
  \end{equation*}
  where $\Unit(N)$ is the unitary group of $N \times N$ unitary
  matrices, $\dd U$ denotes the Haar measure on $\Unit(N)$, and
  $\Weingarten_{N}$ is the \emph{Weingarten function} defined for
  $N \geq n$ by
  \begin{equation*}
    \Weingarten_{N}(\sigma) = \int_{\Unit(N)}U_{1,1} \cdots U_{n,n}\overline{U_{1,\sigma(1)} \cdots U_{n,\sigma(n)}} \ \dd U \quad \text{ for } \sigma \in \Sym_{n}.
  \end{equation*}
\end{theorem}

In the sequel, we shall use the notation
\begin{equation}\label{eq:not-prod-wg}
  \Weingarten_{N}(\bm{\nu}) = \prod_{c = 1}^{N}\Weingarten_{N}(\nu_{c}) \quad \text{ for } \bm{\nu} \in \Sym_{n}^{D}.
\end{equation}

A remarkable identity is that for any $n\le N$ and any $\tau\in \Sym_{n}$:
\begin{equation}\label{eq:inverse-Weingarten}
    \sum_{\nu\in \Sym_{n}} N^{\#(\nu \tau^{-1})} \Weingarten_N(\nu) = \delta_{\tau, \mathrm{id}}\;.
\end{equation}

\

It will be useful to apply to the Weingarten function the same Moebius
transformation used to go from the moments to the classical cumulants.
Given $\bnu \in \Sym_{n}^{D}$ and $\pi \in \Partition(n)$ with
$\Pi(\bm{\nu}) \leq \pi$, we define the Weingarten cumulant functions
(\cite{collins_tensor_2023a} and \cite{lionni_higher_2022}) as:
\begin{equation}
\label{eq:cumulant-weingarten}
    \WeingCm{N}[\pi, \bnu] = \sum_{\substack{\pi' \in \Partition(n)\\\pi \leq \pi'}} \mu_{\pi'}  \prod_{B\in \pi'}\prod_{c=1}^D \Weingarten_{N}(\nu_c{_{|_B}}).
\end{equation}
Given $\Pi\in \PPart(n)$ with $\Pi(\bm{\nu}) \leq \Pi_{[n]}$, we define the
pure analogue:
\begin{equation}\label{eq:cumulant-weingarten-pure}
    \WeingCp{N}[\Pi, \bnu] = \sum_{\substack{\Pi' \in \PPart(n)\\\Pi \leq \Pi'}} \mu_{\Pi'}  \prod_{B\in \Pi'_{[n]}}\prod_{c=1}^D \Weingarten_{N}(\nu_c{_{|_{B}}}).
\end{equation}

\begin{remark}\label{rem:WgC-mix-to-pure}
  We notice that if $\Pi \in \PPart(n)$ satisfies
  $\Pi(\bm{\nu}) \leq \Pi_{[n]}$, we immediately have
  \begin{equation*}
    \WeingCp{N}[\Pi, \bm{\nu}] = \WeingCm{N}[\Pi_{[n]}, \bm{\nu}].
  \end{equation*}

  Conversely, given $\pi \in \Partition(n)$ with $\Pi(\bm{\nu}) \leq \pi$,
  we can make use of the bijection mentioned in Section
  \ref{sub:partitions}. We let $\pi_{\pure}\in\PPart(n)$ which is
  such that both $(\pi_{\pure})_{[n]}$ and
  $(\pi_{\pure})_{[\bar n]}$ coincide with $\pi$ and similarly for
  any $\pi'\ge \pi$. As pure partitions,
  $\pi'_{\pure}\ge \pi_{\pure}\ge \Pi_{\pure}(\bnu)\vee \Pi_{\pure}(\mathrm{id})$.
  Then:
  \begin{equation}
    \WeingCm{N}[\pi, \bnu] =  \WeingCp{N}[\pi_{\pure}, \bnu] \;.
  \end{equation}
\end{remark}

In order to compute the asymptotics of the cumulant Weingarten
function, we introduce the sets of monotone walks.
\begin{definition}[Monotone walk]\label{def:monotone-walk}
  Let $n \in \N^{*}$, $r \in \N$, and $\sigma \in \Sym_{n}$. The set of monotone walks
  to $\sigma$ with $r$ steps is
  \begin{equation*}
    \mwalks^{r}(\sigma) = \biggl\{ (\tau_{1}, \ldots, \tau_{r}) \in \Sym_{n}^{r} \colon \exists (a_{i}, b_{i}), \forall i, \tau_{i} = \cycle{a_{i}, b_{i}}; a_{i} < b_{i}; b_{1} \leq \cdots \leq b_{r}; \tau_{r} \cdots \tau_{1} = \sigma\biggr\}.
  \end{equation*}

  The planar monotone Hurwitz number is then
  \begin{equation*}
    \gamma(\sigma) = \# \Bigl\{\bm{\tau} \in \mwalks^{n + \# \sigma - 2} \colon \Pi(\sigma) \vee \Pi(\tau_{1}) \vee \cdots \vee \Pi(\tau_{r}) = 1_{n}\Bigr\}.
  \end{equation*}
\end{definition}

\begin{remark}\label{rem:Riemann-Hurwitz}
  A monotone walk can be seen as a subsets of ramified coverings of
  the sphere $S' \to \CP^{1}$ with labelled sheets where one point has
  ramification profile given by $\sigma$ and $r$ points have simple
  ramification profile specified by the $\tau_{i}$'s. Hence,
  Riemann-Hurwitz formula gives
  \begin{equation}\label{eq:Riemann-Hurwitz}
    \chi(S') = n - r + \#\sigma.
  \end{equation}
  This implies that whenever
  $\Pi(\sigma) \vee \Pi(\tau_{1}) \vee \cdots \vee \Pi(\tau_{r}) = 1_{n}$,
  \begin{equation*}
    n - r + \# \sigma \leq 2,
  \end{equation*}
  with equality if and only if $S'$ has the topology of a sphere. This
  explains why $\gamma(\sigma)$ is called the \emph{planar} monotone Hurwitz number.
\end{remark}

The asymptotics of the cumulant Weingarten functions are then computed
in Theorem \ref{thm:asympt-cumulant-weingarten-funct}. This Theorem
was shown with different notation in \cite[Theorem
3.3]{collins_tensor_2023a}, and in the mixed case for case $D = 1$ in
\cite{collins_tensor_2023}. See also \cite[Section
2.5]{lionni_higher_2022}.
\begin{theorem}\label{thm:asympt-cumulant-weingarten-funct}
  For any $\bnu \in \Sym_{n}^{D}$, $\pi \in \Partition(n)$ that
  satisfies $\Pi(\bm{\nu}) \leq \pi$, and $\Pi \in \PPart(n)$ that
  satisfies $\Pi(\bm{\nu}) \leq \Pi_{[n]}$, the cumulant Weingarten
  function admits the following asymptotics:
  \begin{equation}\label{eq:connected-Weingarten}
    \begin{split}
      \WeingCm{N}[\pi, \bnu] &= N^{\# \bm{\nu} - 2(\#\pi - 1) - 2nD} \Bigl(\Gamma[\pi, \bnu]+\order{N^{-1}}\Bigr)\\
      \WeingCp{N}[\Pi, \bnu] &= N^{\# \bm{\nu} - 2(\#\Pi - 1) - 2nD} \Bigl(\Gamma[\Pi_{[n]}, \bnu]+\order{N^{-1}}\Bigr)\\
  \end{split}
\end{equation}
where
\begin{equation}
\label{eq:def-gamma}
    \Gamma[\pi, \bnu] = (-1)^{|\bm{\nu}|}\sum_{\substack{\pi_{1}, \ldots, \pi_{D} \in \Partition(n)\\ \forall c, \Pi(\nu_{c}) \leq \pi_{c}\\[+0,5ex]{\pi \vee \pi_1 \vee \ldots \vee \pi_D = 1_n}\\[+0,5ex]{L_D [\{\pi_c\},\; \pi\;;\; \{\Pi(\nu_c)\}]=0}}} \prod_{c=1}^D \prod_{B_c\in \pi_c} \gamma({\nu_c}_{\lvert_{B_c}})\;,
\end{equation}
where $L_D$ has been defined in \eqref{eq:def-of-LD} and $\gamma$ in Definition \ref{def:monotone-walk}.
\end{theorem}
\begin{remark}[Connected case]
  If $\pi = 1_{n}$, we get that
  \begin{equation*}
    \Gamma[1_{n}, \bm{\nu}]
    = \prod_{c = 1}^{D}\Bigl( (-1)^{|\nu_{c}|}\prod_{S \in \Pi(\nu_{c})}\gamma(\nu_{c}\vert_{S})\Bigr)
    = \prod_{c = 1}^{D}\Bigl( (-1)^{|\nu_{c}|}\prod_{S \in \Pi(\nu_{c})}\Cat_{\# S - 1}\Bigr)
    = \Mnc(\bm{\nu}).
  \end{equation*}
\end{remark}
Since almost identical results have been proved in \cite{collins_tensor_2023a,collins_tensor_2023}, we defer the proof to App.~\ref{sec:proof-theorem-ref}.

\subsection{Tensorial free cumulants from finite size precursors}

The quantities we define as free cumulants are obtained by taking
limits of natural objects obtained for finite $N$ from the tensorial
generalization of the Harish-Chandra--Itzykson--Zuber and
Brézin--Gross--Witten integrals.

\subsubsection{Moments and cumulants of the tensor HCIZ and BGW integrals}

The tensor Harish-Chandra--Itzykson--Zuber (HCIZ) integral was defined in \cite{collins_tensor_2023a,collins_tensor_2023} as
\begin{equation}\label{eq:def-HCIZ}
  \mathcal{I}_{D,N}[A ; B] = \int [dU]\  \mathrm{e}^{\Tr(B^{\trans} U A U^\dagger)}\;,
\end{equation}
where $B^{\trans}$ is the transpose of $B$, and $U=U_1 \otimes \cdots\otimes U_D$,  $U_c$ is Haar distributed, and $[dU] = dU_1 \cdots dU_D$.

As shown in \cite{collins_tensor_2023a}, the moments of
$\mathcal{I}_{D,N}[A;B]$ satisfy:
\begin{equation}\label{eq:HCIZ-mom-G}
  \int [dU] \left[\Tr(B^{\trans} U A U^\dagger)\right]^n = \sum_{\btau\in \Sym_{n}^{D}} \Tr_{\btau}[B]\; \G_\btau[A]\;, \quad \text{where} \quad \G_{\bm{\tau}}[A] = \sum_{\bnu \in \Sym_{n}^{D}}\Tr_{\bnu}(A)\Weingarten_{N}(\bm{\nu}\bm{\tau}^{-1})\;.
\end{equation}
These last relations can be inverted for $n\le N$ using \eqref{eq:inverse-Weingarten}, so that
for $n\le N$ the $\G_{\btau}$ determine the trace-invariants:
\begin{equation}
\label{eq:trace-determinist-and-G}
\Tr_{\bsig}(A)
 =  \sum_{\btau \in \Sym_{n}^{D}} N^{     nD - d(\bsig,\btau)    } \G_{\btau} [A] \;.
\end{equation}

In the pure case, letting $J\cdot T=\sum_{\bm{i}} J_{\bm{i}} T_{\bm{i}}$, one rather defines:
\begin{equation}\label{eq:def-BGW}
    \mathcal{J}_{D,N}[T, \bar T ; J, \bar J ] = \int [dU]\  \mathrm{e}^{J\cdot U T + \bar J\cdot \bar U \bar T}\;,
\end{equation}
which can be seen as a particular case of tensor  Brezin-Gross-Witten (BGW) integral \cite{gross_possible_1980,brezin_external_1980}
\begin{equation}
   \mathcal{J}_{D,N}[B] =  \int [dU]\  \mathrm{e}^{\Tr(B^\dagger U  + B U^\dagger)}\;,
\end{equation}
for $\bar B=J\otimes T$. A consequence of \cite[Theorem
4.10]{collins_free_2025} is that the moments of
$\mathcal{J}_{D,N}[J, \bar J ; T, \bar T]$ expand as:
\begin{equation}
\label{eq:BGW-mom-G}
   \int [dU]\ \left(J\cdot T\right)^n \left(\bar J\cdot \bar T\right)^n = \sum_{\btau\in \Sym_{n}^{D}} \Tr_{\btau}[J, \bar J]\; \G_\btau[T, \bar T]\;.
\end{equation}
The inverse relation \eqref{eq:trace-determinist-and-G} is similar in the pure case.

\

In \cite{lacroix_finite_2025a}, for the matrix case ($D=1$), Lacroix and Zuber call the $\G_{\tau} [A]$, $\tau\in \Sym_{n}$, $n\in \N^\star$ the ``finite
$N$ precursors of free cumulants'' for matrices: it is another basis
for LU-invariant polynomials, and in the large $N$ limit their
rescaled limits converge to products of (first-order) free cumulants
for the disjoint cycles of $\sigma$. The quantity $\Gb$ corresponds to the "$\mathbb{U}$-cumulants" of Capitaine and Casalis \cite{capitaine_cumulants_2006,capitaine_geometric_2008}, in the case of families of random matrices.

 \

For $\bsig\in \Sym_{n}^{D}$ and $\pi \ge \Pi(\bsig)$, we let
${\G}_{\pi, \bsig} [A] = \prod_{B\in \pi}  {\G}_{\btau_{|_B}} [A] $,
and we define the classical cumulants associated to the $\G_\bsig$:
\begin{equation}
\label{eq:cumulants-of-HCIZ}
\mathcal{L}^\mix_{\bsig} [A]  = \sum_{\pi\ge \Pi(\bsig)} \mu_\pi {\G}_{\pi, \bsig} [A] = \sum_{\btau\in \Sym_{n}^{D}} \Tr_\btau(A)\  \WeingCm{N}\bigl[\Pi(\bsig, \btau), \bsig\btau^{-1}\bigr] \;.
\end{equation}
It only differs from the $\G$ in the disconnected case:
\begin{equation}
   \textrm{If }\  \Pi(\bsig)=1_n, \qquad  \mathcal{L}^{\mix}_\bsig[A] =  {\G}_{\bsig} [A]\;.
\end{equation}

They are the coefficients of the expansion of the free-energy of $\mathcal{I}_{D,N}[A ; B]$ on the trace-invariants of~$B$ \cite{collins_tensor_2023a}:
\begin{equation}
\label{eq:free-energy-HCIZ}
    \frac{\partial^n}{\partial z^n} \log \mathcal{I}_{D,N}[A ; zB] \Bigr\lvert_{z=0} =  \sum_{\bsig\in \Sym_{n}^{D}} \Tr_{\bsig}(B)\; \mathcal{L}^{\mix}_\bsig[A] \;.
\end{equation}

\

If for instance  $A$ is the  $N^D\times N^D$ identity matrix $\un$:
\begin{equation}
\label{eq:mixed-ex-identity-matrix}
    \Tr_\bsig(\un) = N^{\#\bm{\sigma}}\;, \qquad  \G_\bsig[\un] = \delta_{\bsig, \bm{\mathrm{id}}}\;,
    \qquad
    \mathcal{L}^{\mix}_\bsig[\un] = \delta_{\bsig, \bm{\mathrm{id}}}\ \delta_{\Pi(\bsig), 1_n} = \delta_{n, 1}\;.
\end{equation}

\

We also define the pure analogue:
\begin{equation}
\label{eq:cumulants-of-HCIZ-pure}
\mathcal{L}^{\pure}_{\bsig} [T, \bar T]  = \sum_{\Pi\ge \Pi_{\pure}(\bsig)} \mu_\Pi\ {\G}_{\Pi, \bsig} [T, \bar T] = \sum_{\btau\in \Sym_{n}^{D}} \Tr_\btau(T, \bar T)\  \WeingCp{N}\bigl[\Pi_{\pure}(\bsig, \btau), \bsig\btau^{-1}\bigr] \;.
\end{equation}
In particular, in the purely connected case:
\begin{equation}
   \textrm{If }\  \Pi_\mathrm{p}(\bsig)=1_{n,\bar n}, \qquad  \mathcal{L}^{\pure}_\bsig[T, \bar T] =  {\G}_{\bsig} [T, \bar T]\;,
\end{equation}
but the ${\G}_{\bsig} $ and $\mathcal{L}^{\pure}_\bsig$ usually differ otherwise.

The computations of \cite[Section A.2]{collins_free_2025} in the
particular case where $T=UT'$ with $U=U_1\otimes \cdots \otimes U_D$, $U_c$ a tensor
product of Haar distributed unitary matrices (in analogy to
\cite{collins_tensor_2023a}), show that the quantities
$\mathcal{L}^{\pure}_\bsig$ appear as coefficients of the expansion of the
free-energy of $\mathcal{J}_{D,N}[T, \bar T ; J, \bar J ]$ on the
trace-invariants of $J, \bar J$:
\begin{equation}
\label{eq:free-energy-BGW}
\frac{\partial^n}{\partial z^n}\frac{\partial^n}{\partial \bar z^n} \log  \mathcal{J}_{D,N}[T, \bar T ; zJ, \bar z \bar J ] \Bigr\lvert_{z=\bar z = 0} =  \sum_{\bsig\in \Sym_{n}^{D}} \Tr_{\bsig}(J, \bar J)\; \mathcal{L}^{\pure}_\bsig[T, \bar T] \;.
\end{equation}

\subsubsection{Mixed finite size precursors}
\label{subsub:mixed-precursors}

For random tensors, we recall the construction of
\cite{collins_free_2025}. For $A$ random, $\bsig\in \Sym_{n}^{D}$, and
$\pi \ge \Pi(\bsig)$, we let 
\begin{equation} 
\Gb_{\btau} [A] = \E\bigl[\G_{\btau} (A)\bigr], 
\qquad \text{ and }\qquad 
\Gb_{\pi, \btau} [A] = \prod_{B\in \pi} \Gb_{\btau_{|_B}} [A]\;,
\end{equation} 
and define the finite $N$ precursors of the free cumulants in the mixed
case as the classical cumulants associated to the quantities
$\Gb_{\bm{\tau}}[A]$:
\begin{equation}
\label{eq:HOFC-def-C-short}
\Km_\bsig[A]  = \sum_{\substack{\pi \in \Partition(n)\\\Pi(\bsig) \leq \pi}} \mu_\pi \Gb_{\pi, \bsig} [A]\;,
\end{equation}
which expresses $\Km_\bsig[A]$ in terms of the
$\E\bigl[\Tr_{\btau_{|_G} } (A)\bigr]$ for $G\in \pi$. This can be
rewritten in terms of the classical cumulants as
\begin{equation}
    \label{eq:HOFC-def-C}
    \Km_\bsig[A]  = \sum_{\btau \in \Sym_{n}^{D}} \sum_{\substack{\pi \in \Partition(n)\\\Pi(\bm{\tau}) \leq \pi}} \Phim_{\pi, \btau}[A]\  \WeingCm{N}\bigl[\pi\vee \Pi(\bsig), \bsig\btau^{-1}\bigr]\;,
  \end{equation}
  where $\WeingCm{N}$ has been defined in
  \eqref{eq:cumulant-weingarten}.

\begin{lemma}\label{lem:inverse-Km}
  Let $n \in \N^{*}$ and $\bm{\sigma} \in \Sym_{n}^{D}$. The inverse
  relation of \eqref{eq:HOFC-def-C-short} in terms of the moments are:
  \begin{equation}
    \label{eq:moments-from-precursors-mixed}
    \E\left[\Tr_{\bsig}(A)\right] = \sum_{\btau \in \Sym_{n}^{D}} N^{nD - d(\bsig, \btau)} \sum_{\substack{\pi \in \Partition(n)\\ \Pi(\bm{\tau}) \leq \pi}}  \
    \Km_{\pi,\btau}[A] \;,
  \end{equation}
 in terms of the quantities $\Gb$:
  \begin{equation}\label{eq:inv-G-K-mix}
    \Gb_{\bm{\sigma}}[A] = \sum_{\substack{\pi \in \Partition(n)\\\Pi(\bm{\sigma}) \leq \pi}}\Km_{\pi, \bsig}[A]\;,
  \end{equation}
in terms of the classical cumulants:
  \begin{equation}
    \label{eq:Phi-in-terms-of-precursors-mixed}
    \Phim_{ \bsig} [A]  = \sum_{\btau \in \Sym_{n}^{D}} N^{nD - d(\bsig, \btau)}  \sum_{\substack{{ \pi \ge \Pi(\btau)}\\{\Pi(\bsig)\vee\pi = 1_n}}}  \  \Km_{\pi,\btau}[A] \;.
  \end{equation}
\end{lemma}
This result is an application of Moebius inversion. Arguments of this
sort are commonplace when discussing free cumulants, see for instance
\cite{nica_lectures_2006}. For the convenience of the reader, we give
a short argument.
\begin{proof}
  We start by proving \eqref{eq:inv-G-K-mix}. By summing over partitions, we have
  \begin{equation*}
    \sum_{\substack{\pi \in \Partition(n)\\\Pi(\bm{\sigma}) \leq \pi}}\Km_{\pi, \bsig}[A]  = \sum_{\substack{\pi \in \Partition(n)\\\Pi(\bm{\sigma}) \leq \pi}}\sum_{\substack{\pi' \in \Partition(n)\\\Pi(\bsig) \leq \pi' \leq \pi}} \Bigl(\prod_{S \in \pi}\mu_{\pi'\vert_{S}}\Bigr) \Gb_{\pi', \bsig} [A]\;,
  \end{equation*}
  We exchange the sums on $\pi$ and $\pi'$, and use Moebius inversion \eqref{eq:moebius-inversion-partition},
  this immediately gives \eqref{eq:inv-G-K-mix}.

  To get \eqref{eq:moments-from-precursors-mixed}, we use
  \eqref{eq:trace-determinist-and-G} and \eqref{eq:inv-G-K-mix}:
  \begin{equation*}
    \E\Bigl[\Tr_{\bm{\sigma}}[A]\Bigr] = \sum_{\bm{\tau} \in \Sym_{n}^{D}} N^{nD - d(\bm{\sigma}, \bm{\tau})}\Gb_{\bm{\tau}}[A] = \sum_{\bm{\tau} \in \Sym_{n}^{D}} N^{nD - d(\bm{\sigma}, \bm{\tau})}\sum_{\substack{\pi \in \Partition(n)\\\Pi(\bm{\sigma}) \leq \pi}}\Km_{\pi, \bsig}[A].
  \end{equation*}

  Finally, one get \eqref{eq:Phi-in-terms-of-precursors-mixed} from
  from \eqref{eq:class-cum-trace-ind} and
  \eqref{eq:moments-from-precursors-mixed}, and the fact that
  $\sum_{\tilde{\pi} \leq \pi}\mu_{\tilde{\pi}}=\delta_{\pi, 1_n}$.
\end{proof}

\

In the connected case, one recovers the expectations of the quantities derived from the tensor HCIZ integral:
\begin{equation}
\label{eq:quantities-match-connected-mixed}
   \textrm{If }\  \Pi(\bsig)=1_n, \qquad  \Km_\bsig[A] =  \Gb_{\bsig} [A] = \E\bigl[\mathcal{L}^{\mix}_\bsig[A]\bigr]\;,
\end{equation}
but  these relations do not usually hold anymore for non-connected $\bsig$.
If $A'$ is a deterministic tensor and $A=A'$ or $A=UA'U^\dagger$ with $U=U_1\otimes \cdots \otimes U_D$, $U_c$ Haar distributed, the relation between  $\Km_\bsig$ and $\mathcal{L}^{\mix}_\bsig$ extends to arbitrary $\bsig$, connected or not:
\begin{equation}
\label{eq:finite-precursor-log-version-deterministic}
    \Km_{\bsig} [A] = \mathcal{L}^\mix_{\bsig} [A']\;.
\end{equation}

\subsubsection{Pure finite size precursors}

In the pure case, recalling that  $\Gb$ is $\Gb_{\Pi, \bsig} [T,\bar T]=\prod_{B\in \Pi_{[n]}}  \Gb_{\bsig_{|_{B}}} [T,\bar T]$, we have:
\begin{equation}
\label{eq:HOFC-def-TT-short}
     \Kp_\bsig[T,\bar T]  = \sum_{\substack{\Pi \in \PPart(n)\\\Pi_{\pure}(\bm{\sigma}) \leq \Pi}} \mu_\Pi \,  \Gb_{\Pi, \bsig} [T,\bar T]\;.
\end{equation}
If $\bsig$ is purely connected, we get the simpler relation
$\Kp_\bsig[T,\bar T] = \Gb_{\bsig} [T, \bar T]$. The relation
\eqref{eq:HOFC-def-TT-short} can be rewritten in terms of the
classical cumulants as
\begin{equation}\label{eq:HOFC-def-TT}
  \Kp_\bsig[T,\bar T]  = \sum_{\btau \in \Sym_{n}^{D}} \sum_{\substack{\Pi \in \PPart(n)\\ \Pi_{\pure}(\bm{\tau}) \leq \Pi}} \Phip_{\Pi, \btau}[T, \bar T]\
  \WeingCp{N}\bigl[\Pi\vee\Pi_{\pure}(\bsig) , \bsig\btau^{-1}\bigr] \;,
\end{equation}
where $\WeingCp{N}$ has been defined in \eqref{eq:cumulant-weingarten-pure}.
\begin{lemma}\label{lem:inverse-Kp}
  Let $n \in \N^{*}$ and $\bm{\sigma} \in \Sym_{n}^{D}$. The inverse
  relation of \eqref{eq:HOFC-def-C-short}  in terms of the moments are:
  \begin{equation}
    \label{eq:moments-from-precursors-pure}
    \E\left[\Tr_{\bsig}(T,\bar T)\right]  = \sum_{\btau \in  \Sym^D_{n} } N^{nD - d(\bsig, \btau)}  \sum_{\substack{{\Pi\in \PPart(n)}\\\Pi_{\pure}(\bm{\tau}) \leq \Pi}}  \
    \Kp_{\Pi,\btau}[T,\bar T]\;,
  \end{equation}
  in terms of the quantities $\Gb$:
  \begin{equation}\label{eq:inv-G-K-pure}
    \Gb_{\bm{\sigma}}[T, \bar{T}] = \sum_{\substack{\Pi \in \PPart(n)\\\Pi_{\pure}(\bm{\sigma}) \leq \Pi}}\Kp_{\Pi, \bsig}[T, \bar{T}]\;,
  \end{equation}
  in terms of the classical cumulants:
  \begin{equation}
\label{eq:Phi-in-terms-of-precursors-pure}
    \Phip_{ \bsig} [T, \bar T]  = \sum_{\btau \in \Sym_{n}^{D}}N^{nD - d(\bsig, \btau)}   \sum_{\substack{\Pi \in \PPart(n)\\\Pi_{\pure}(\bm{\tau}) \leq \Pi\\{\Pi_{\pure}(\bsig)\vee\Pi = 1_{n, \bar n}}}}  \  \Kp_{\Pi,\btau}[T, \bar T]  \;.
\end{equation}

\end{lemma}
We omit the proof of Lemma \ref{lem:inverse-Kp} as once one relies on Remark \ref{rem:lattice-Pp}, the proof is almost identical to the one of Lemma \ref{lem:inverse-Km}.

\

In the purely connected case, one recovers the expectations of the quantities derived from the tensor BGW integral:
\begin{equation}
\label{eq:quantities-match-connected-PURE}
   \textrm{If }\  \Pi_\mathrm{p}(\bsig)=1_{n, \bar n}, \qquad  \Kp_\bsig[T, \bar T] =  \Gb_{\bsig} [T, \bar T] = \E\bigl[\mathcal{L}^{\pure}_\bsig[T, \bar T]\bigr]\;.
\end{equation}
These relations do not usually hold anymore if  $\bsig$ is not purely connected.  If $T'$ is a deterministic tensor and $T=T'$ or $T=UT'$ with $U=U_1\otimes \cdots \otimes U_D$, $U_c$ Haar distributed, one has for arbitrary $\bsig$:
\begin{equation}
  \Kp_{\bsig} [T, \bar T] = \mathcal{L}^{\pure}_{\bsig} [T', \bar T'] \;.
\end{equation}

\subsubsection{Tensorial free cumulants}

We define the free cumulant for mixed and pure tensors as the properly
rescaled limits of the finite-$N$ precursors to the free cumulants,
$\Km$ in the mixed case and $\Kp$ in the pure case. The limits
$\fcumm$ and $\fcump$ enjoy several properties that make them natural
quantities to study the limiting distibution. Firstly, they are
related to the $\vphim$ and the $\vphip$ discussed in Section
\ref{subsub:asympt-def-of-distrib} by inversion relations. The precise
form these relations take will be discussed below, but we note that
they are obtained as a $N \to \infty$ limit of the relations given in
Lemmata \ref{lem:inverse-Km} and \ref{lem:inverse-Kp}. Secondly, they
satisfy additivity properties. These properties are stated in \cite[Proposition 4.9 and Equation
(4.25)]{collins_free_2025}.
\begin{prop}[Additivity at finite $N$]\label{prop:add-finite-N}
  Let $A_{1}$ and $A_{2}$ be two independent LU-invariant random mixed tensors, and
  $(T_{1}, \bar{T}_{1})$ and $(T_{2}, \bar{T}_{2})$ be two independent
  LU-invariant random pure tensors. Assume that all the moments of these tensors
  exist. Then, for all $n \in \N^{*}$ and $\bm{\sigma} \in \Sym_{n}^{D}$,  there exist $N_n$ such that for every $N\ge N_n$:
  we have in the mixed case
  \begin{equation*}
   \Km_{\bm{\sigma}}[A_{1} + A_{2}] = \Km_{\bm{\sigma}}[A_{1}] + \Km_{\bm{\sigma}}[A_{2}],
  \end{equation*}
  and in the pure case
  \begin{equation*}
   \Kp_{\bm{\sigma}}[T_{1} + T_{2}, \bar{T}_{1} + \bar{T}_{2}] = \Kp_{\bm{\sigma}}[T_{1}, \bar{T}_{1}] + \Kp_{\bm{\sigma}}[T_{2}, \bar{T}_{2}].
  \end{equation*}
\end{prop}
Proposition \ref{prop:add-finite-N} implies additivity properties of
the free cumulants $\fcumm_{\bm{\sigma}}$ and $\fcump_{\bm{\sigma}}$ at all
orders, whether $\bm{\sigma}$ is connected or not.

\subsection{Microscopic versus macroscopic}
In this subsection, we show that it is equivalent to define the distribution of a LU-invariant random tensor by the data of its macroscopic moments -- the expectations of trace-invariants -- or by the data of the microscopic moments -- the joint moments of the tensor components.

We recall the following (the expression is the same in the pure case replacing $A$ by $(T, \bar T$): \begin{equation}\label{eq:recall-bar-G}\Gb_{\btau} [A]  = \sum_{\bnu \in \Sym_{n}^{D}}\E[\Tr_{\bnu}(A)]\,  \prod_{c=1}^D W^{(N)} ( \nu_{c} \tau_{c}^{-1}) \;. \end{equation}
\begin{prop}
\label{prop:micro-vs-macro}
If $A$ (respectively $(T, \bar T)$) is a mixed (respectively pure) LU-invariant random tensor with $D$ inputs, then the joint moments of the tensor components can be computed from the expectations of trace-invariants as:
\begin{align}
\label{eq:tensor-components-LU-inv}
\E\left[\prod_{k=1}^n A_{\bm{i}(k), \bm{j}(k)}\right] & = \sum_{\btau\in \Sym_{n}^{D}}
\prod_{k=1}^n\delta_{\bm{i}(k), \bm{j}\circ \bm{\tau}(k)}
\Gb_{\btau} [A]\;, \\
\label{eq:tensor-components-LU-inv-pure}
\E\left[\prod_{k=1}^n T_{\bm{i}(k)}\bar T_{\bm{j}(k)}\right] & = \sum_{\btau\in \Sym_{n}^{D}} \prod_{k=1}^n
\delta_{\bm{i}(k), \bm{j}\circ \bm{\tau}(k)}
\Gb_{\btau} [T,\bar T] \;.
\end{align}
For every $D$ and if $n\le N$, the coefficients $\Gb_{\btau} [A]$ and $\Gb_{\btau} [T,\bar T]$ in these expressions are respectively unique  up to $\sim_\mix$ and $\sim_\mathrm{p}$, see  \eqref{eq:relabeling}.
\end{prop}

Similarly, for the classical cumulants of the tensor components, one has the same formulae with  $\Km_\btau[A]$ and $\Kp_\btau[T,\bar T] $ on the right hand side:
\begin{align}
\label{eq:tensor-components-LU-inv-cum}
k_n\left(\left\{A_{\bm{i}(k), \bm{j}(k)} \right\}_{1\le k \le n}\right) & = \sum_{\btau\in \Sym_{n}^{D}}
\prod_{k=1}^n\delta_{\bm{i}(k), \bm{j}\circ \bm{\tau}(k)}
\Km_{\btau} [A]\;, \\
\label{eq:tensor-components-LU-inv-pure-cum}
k_n\left(\left\{ T_{\bm{i}(k)}\bar T_{\bm{j}(k)}\right\}_{1\le k \le n}\right) & = \sum_{\btau\in \Sym_{n}^{D}} \prod_{k=1}^n
\delta_{\bm{i}(k), \bm{j}\circ \bm{\tau}(k)}
\Kp_{\btau} [T,\bar T] \;.
\end{align}

\begin{proof}
  The formulae in the statement are obtained using the LU-invariance
  and the Weingarten integration formula, Theorem
  \ref{thm:Weingarten-formula}. Alternatively, one can use the
  generating function of microscopic moments, the LU-invariance of
  $A$, as well as \eqref{eq:HCIZ-mom-G}:
  \begin{equation*}
    \begin{split}
        \E\left[\prod_{k=1}^n A_{\bm{i}(k), \bm{j}(k)}\right] 
        &=   \Bigl(\prod_{k=1}^n  \frac{\partial}{\partial B_{\bm{i}(k), \bm{j}(k)}} \Bigr) \log \E\left[e^{\Tr(B^{\trans}A)}\right]\Bigr\lvert_{B=0}\\
        &=  \Bigl(\prod_{k=1}^n  \frac{\partial}{\partial B_{\bm{i}(k), \bm{j}(k)}} \Bigr)  \sum_{\bsig\in \Sym_{n}^{D}} \Tr_{\bsig}(B)\; \Gb_\bsig[A] \;.
    \end{split}
  \end{equation*}

  Knowing \eqref{eq:tensor-components-LU-inv}, we derive the
  expectation of \eqref{eq:trace-determinist-and-G}, namely:
\begin{equation}
\label{eq:recall-macro-from-G}
\E[\Tr_{\bsig}(A)]
=  \sum_{\btau \in \Sym_{n}^{D}} N^{     nD - d(\bsig,\btau)    } \Gb_{\btau} [A] \;.
\end{equation}
If \eqref{eq:tensor-components-LU-inv} were to hold for another
coefficient $F_\bsig$, one would obtain \eqref{eq:recall-macro-from-G}
with $\G_{\btau}$ replaced by $F_\btau$. But for $n\le N$,
the inversion formula \eqref{eq:inverse-Weingarten} imposes
$\G_{\btau}=F_\btau$.
\end{proof}

\

One may also see the unicity of the expansion in the trace-invariants
associated to elements of the quotients $\Sym_{n}^{D}/\sim_{\mix}$ and
$\Sym_{n}^{D}/\sim_{\pure}$ from the following formulae, stated in
\cite[Proposition 4.11]{collins_free_2025}. For each $c\in \{1,\ldots D\}$,
choose $\sigma_c\in \Sym_{n}$, and distinct $1\le j_c(1), \ldots, j_c(n)\le N$,
then the finite size precursors $\Kp_{\bsig}$ may be expressed
in terms of the joint cumulants of entries of the tensors with
distinct indices as:
\begin{equation}\label{eq:micro-Kpm}
\Km_{\bsig} [A] = k_n\left(\left\{A_{\bm{j}\circ \bm{\sigma}(k)\; ;\;   \bm{j}(k)  }  \right\}_{1\le k \le n}\right) \;, \quad \text{ and }\quad
\Kp_{\bsig} [T,\bar T] =   k_n\left(\left\{ T_{\bm{j}(\bsig (k)} , \bar T_{ \bm{j}(k) } \right\}_{1\le k \le n}\right)\;,
\end{equation}
and in the same way, $\Gb_{\bsig}$ may be expressed in terms of the joint moments of the entries of the tensors with distinct indices as:
\begin{equation}\label{eq:micro-Gb}
\Gb_{\bsig} [A] = \E\Bigl[\prod_{k=1}^n A_{\bm{j} \circ \bm{\sigma}(k)\; ;\;   \bm{j}(k)  } \Bigr]\;, \qquad \text{ and }
\qquad
\Gb_{\bsig} [T,\bar T] =   \E\Bigl[\prod_{k=1}^n T_{\bm{j}\circ \bm{\sigma}(k)} \bar T_{ \bm{j}(k) } \Bigr]\;.
\end{equation}

\

\section{Random tensors with coarser local unitary invariances}
\label{sec:Coarser-invariance}

Prop.~\ref{prop:micro-vs-macro} shows that the joint moments of the tensor components can be computed from the macroscopic moments, and that for every $D$, the coefficients $\Gb_{\bsig}$ in these expressions are unique up to relabeling. If the random tensor has a larger LU invariance group, then the same will be true for the macroscopic moments of the larger LU-invariance group. This drastically constrains the possible values that  $\Gb_{\bsig}$ and ${\Kp}_{\bsig}$ can take and entirely solves the question of determining the tensorial free cumulants of unitarily invariant random matrices  or unitarily invariant random vectors with subdivided index sets. A consequence of these drastic simplifications is that for all random tensors invariant under global unitary transformations (which implies the local unitary invariance),  the $\Phip_\bsig$ are expressed through expansions over a single permutation. This implies that such distributions have the same scaling behavior, for instance.

\subsection{Coarser invariance and the vanishing of finite size precursors}
\label{sub:coarser1}

Consider $ D'\le D$ and a surjective mapping $\xi \colon [D] \to [D']$. Such a mapping defines a partition $\ker \xi$ given by
\begin{equation}
    \ker \xi = \Bigl\{ \xi^{-1}(\{c'\}) \colon 1 \leq c' \leq D'\Bigr\}.
\end{equation}
The mapping $\xi$ and the partition $\ker \xi$ prescribe a grouping of the $D$ colors in $D'$ blocks. We denote by $(1_D)$ the mapping sending $[D]$ to $\{1\}$.

 Consider $A$ and $(T, \bar{T})$ some tensors with $D$ inputs each taking value in $[N]$:  $A=\{A_{\bm{i};\bm{j}}\}_{\bm{i},\bm{j}\in [N]^D}$ and $T=\{T_{\bm{i}}\}_{\bm{i}\in [N]^D}$. Consider as well a tensor $M=\{M_{\bm{i'};\bm{j'}}\}_{\{i'_{c'},j'_{c'}\in [N_{c'}]\}_{1\le c'\le D'}}$ with $D'$ inputs but the same total number of entries:   $N_{c'}=N^{p_{c'}}$, where letting $p_{c'}= \# \xi^{-1}(\{c'\})$ is the number of colors sent to $c'$ by $\xi$.

For each $c'$, consider a bijective map $f_{c'}: [N]^{p_{c'}} \rightarrow [N_{c'}]$, and the bijective map
\[
f_\xi: [N]^D \simeq [N]^{p_{1}} \times \cdots \times [N]^{p_{D'}}  \rightarrow [N_1]\times \cdots \times [N_{D'}]
\]
induced by  the $\{f_{c'}\}$. If $A$ is a tensor with $D$ inputs each taking value in $[N]$, we let  $A_{\xi}$ be the tensor with $D'$ inputs obtained reorganizing the entries of $A$ as $(A_{\xi})_{\bm{i'};\bm{j'}} = A_{f_\xi^{-1}(\bm{i'});f_\xi^{-1}(\bm{j'})}$. Note that $A_\xi$ depends on the choice of bijections $f_c$. If for instance $\ker \xi=1_D$ and  if $f_{1_D}(i_1, \ldots, i_D)= i$ and $f_{1_D}(j_1, \ldots, j_D)= j$, for some $i_1, \ldots, i_D, j_1, \ldots, j_D \in [N]$ and some $i, j\in [N^D]$, then  $A_{(1_D)}\in\mathbb{C}^{N^D\times N^D}$, and:
\[
(A_{(1_D)})_{i;j} = A_{i_1, \ldots, i_D;j_1, \ldots, j_D}\;.
\]
Reciprocally, if $M$ is a tensor with $D'$ inputs, such that the input with color $c'$ takes value in $[N_{c'}]$, we can reorganize its entries by subdividing its indices: $(M_{\xi^{-1}})_{\bm{i};\bm{j}} = M_{f_\xi(\bm{i});f_\xi(\bm{j})}$.

With these notations, we may then define the multiplication of $M$ and  $A$  as:
\begin{equation}
    (MA)_{\bm{i'};\bm{j'}} := (MA_{\xi})_{\bm{i'};\bm{j'}}  = \sum_{\bm{a'} \in [N_1] \times \cdots \times [N_{D'}]}  M_{\bm{i'};\bm{a'}} A_{f_\xi^{-1}(\bm{a'})  ; f_\xi^{-1}(\bm{j'})} \;,
\end{equation}
and similarly for $AM$ and for $MT$.

If $A$ (respectively $(T, \bar T)$) is a mixed (respectively pure)  random tensor with $D$ inputs, we say that it is \emph{LU-invariant with respect to $\xi$} if for any $U=U_{1}\otimes  \cdots \otimes  U_{D'}$, $U_{c'}\in \Unit(N_{c'})$, $UAU^\dagger$ and $A$ have the same distribution (respectively $(UT, \bar U\bar T)$). We say that $A$ (respectively $(T, \bar T)$) is \emph{global unitary invariant}  if this holds for $\xi=1_D$. If a random tensor is LU-invariant for $\xi\colon[D] \to [D']$, then it is LU-invariant for any $\xi'$ with a finer kernel $\ker \xi' \leq \ker \xi$.

The following map canonically defines a one-to-one correspondence between $\Sym_{n}^{D'}$ and the following subset of   $\Sym_{n}^{D}$ :
\begin{equation}
   g_\xi: \bsig' = (\sigma'_1, \ldots, \sigma'_{D'})\in \Sym_{n}^{D'} \qquad\mapsto \qquad \bsig = (\sigma'_{\xi(1)}, \ldots, \sigma'_{\xi(D)})\in \Sym_{n}^{D}\;.
\end{equation}
This corresponds to adding copies of the $\sigma'_c$'s as prescribed by $\xi$. For instance, if $\xi=(1_D)$, then for $\eta\in \Sym_{n}$, $g_{1_D}(\eta) = (\eta, \ldots, \eta)\in \Sym_{n}^{D}$.

If $A$ (respectively $(T, \bar T)$) is LU-invariant with respect to $\xi\colon[D] \to [D']$, then the macroscopic moments describing the distribution seen with respect to this invariance are identified by $\bsig'\in \Sym_{n}^{D'}$, up to $\sim_{\mix}$ and $\sim_{\pure}$, and correspond to (see \eqref{def:trace-invariants}):
\begin{equation}
    \Tr_{\bsig'}(A) := \Tr_{\bsig'}(A_{\xi}) = \Tr_{g_\xi(\bsig')}(A)
 \; ,
\end{equation}
and similarly for $(T, \bar T)$. We define $\Gb_{\bsig'} [A]$, $\mathcal{K}^\mix_{\bsig'}[A]$, $\Gb_{\bsig'} [T,\bar T]$, $\Kp_{\bsig'}[T, \bar T]$ in the same way.

If for instance $\xi=(1_D)$ and $\gamma=\cycle{1, 2, \cdots, n}\in \Sym_{n}$, then $A_{({1_D})}\in\mathbb{C}^{N^D\times N^D}$, $g_{(1_D)}(\gamma) = (\gamma, \ldots, \gamma)\in \Sym_{n}^{D}$,  and one has:
\begin{equation}
    \Tr(A^n) = \Tr_{\gamma}(A) : = \Tr(A_{(1_D)}^n)= \Tr_\gamma(A_{(1_D)}) = \Tr_{(\gamma, \ldots, \gamma)}(A)\;.
\end{equation}
If for instance $A$ is a $N^D\times N^D$ Wishart random matrix, then $\E[\Tr(A^n)]=N^D \mathrm{Cat}_n (1+o(1))$ and $\mathcal{K}^\mix_{\gamma}[A]= N^D (1 + o(1))$.

\begin{remark}
  If $A$ is LU-invariant for $\xi$ then the law of $A_\xi$ does not depend on the choice of $f_1, \ldots, f_{D'}$. Indeed, the groups $\Unit(N_{c'}), 1 \leq c' \leq D'$ contain permutation matrices of size $N_{c'} \times N_{c'}$, and any two choices of $f_{c'}$ are related by such a permutation.
\end{remark}

\begin{theorem}
Consider a random tensor $A$ (respectively $(T, \bar T)$) with $D$ inputs, and which is  LU-invariant with respect to  $\xi \colon [D] \to [D']$, then for any $n\le N$ and any $\bsig \in \Sym_{n}^{D}$:
\begin{align}
\label{eq:large-group-precursors-mixed}
    \mathcal{K}^\mix_\bsig[A] &= \sum_{\bsig'\in \Sym_{n}^{D'}}  \delta_{\bsig, g_\xi(\bsig') }  \mathcal{K}^\mix_{\bsig'}[A_{\xi}] \; ,\\
        \Kp_\bsig[T, \bar T] &= \sum_{\bsig'\in \Sym_{n}^{D'}}  \delta_{\bsig, g_\xi(\bsig') }  \Kp_{\bsig'}[T_{\xi}, \bar T_{\xi}] \; ,
        \label{eq:large-group-precursors-pure}
\end{align}
where $\delta_{\bsig, g_\xi(\bsig') } = \prod_{c = 1}^{D}\delta_{\sigma_c, \sigma_{\xi(c)}}$. As a consequence, the same relations \eqref{eq:large-group-precursors-mixed} relate the tensorial free cumulants $ \fcumm_{\bsig'}(a), \bsig'\in \Sym_{n}^{D'}$, and $ \fcumm_{\bsig}(a) $, $\bsig\in \Sym_{n}^{D}$, and similarly in the pure case.
\end{theorem}
\begin{proof}
  From the unicity of Prop.~\ref{prop:micro-vs-macro}, one has for $\bm{j} \colon [n] \to [N_1] \times \cdots \times [N_{D'}]$ with distinct indices:
  \begin{equation*}
      \Gb_{\bsig}[A]
      = \E\Bigl[A_{\bm{i} \circ \bm{\sigma}; \bm{i}}\Bigr]
      = \E\Bigl[(A_\xi)_{f_\xi \circ (\bm{i} \circ \bm{\sigma}); f_\xi \circ \bm{i}}\Bigr]
      = \sum_{\bm{\sigma'} \in \Sym_n^{D'}}\delta_{f_\xi\circ (\bm{i} \circ \bm{\sigma}), (f_\xi \circ \bm{i}) \circ \bm{\sigma'}}\Gb_{\bm{\sigma'}}[A_\xi],
  \end{equation*}
  and
  \begin{equation*}
      \delta_{f_\xi\circ (\bm{i} \circ \bm{\sigma}), (f_\xi \circ \bm{i}) \circ \bm{\sigma'}}
      = \prod_{k=1}^{n} \prod_{c = 1}^{D'}\delta_{f{c'}(\bm{i}_{\xi^{-1}(c')}(\sigma_{\xi^{-1}(c')}(k))), f_{c'}(\bm{i}_{\xi^{-1}(c)}(\sigma'_{c'}))}
       = \prod_{c = 1}^{D'}\delta_{\sigma_{\xi^{-1}(c')},\sigma'_{c'}}
       = \delta_{\sigma,g_\xi(\sigma')}.
  \end{equation*}
  Thus, for all $\bsig \in \Sym_n^D$ and $\pi \in \Partition(n)$ with $\Pi(\bsig) \leq \pi$:
\begin{equation*}
        \Gb_\bsig[A] = \sum_{\bsig'\in \Sym_{n}^{D'}}  \delta_{\bsig, g_\xi(\bsig') }      \Gb_{\bsig'}[A]
        \quad \text{ and } \quad
        \Gb_{\pi,\bsig}[A] = \sum_{\substack{\bsig'\in \Sym_{n}^{D'}\\\Pi(\bsig') \leq \pi}}  \delta_{\bsig, g_\xi(\bsig') } \Gb_{\pi, \bsig'}[A]
\end{equation*}
and similarly for the pure case. The formulae in the statement follow from the definitions of $\Km_\bsig$ and $\Kp_\bsig$.
\end{proof}

\

With these notations, the macroscopic moments and cumulants \eqref{eq:moments-from-precursors-mixed} and \eqref{eq:Phi-in-terms-of-precursors-mixed} of $A$ or $T, \bar T$ can be expressed as:
\begin{align}
\E\left[\Tr_{\bsig}(A)\right] &= \sum_{\btau'\in \Sym_{n}^{D'}} N^{nD - d(\bsig, \btau')}  \sum_{ \pi \ge \Pi(\btau')} \
\mathcal{K}^{\mix }_{\pi,\btau'}[A] \;,\\
    \Phim_{ \bsig} [A]  &= \sum_{\btau' \in S^{D'}_n} N^{nD - d(\bsig, \btau')}  \sum_{\substack{{ \pi \ge \Pi(\btau')}\\{\Pi(\bsig)\vee\pi = 1_n}}}  \  \mathcal{K}^{\mix }_{\pi,\btau'}[A]  \;,
\end{align}
and similarly in the pure case.

\subsection{The mixed and global unitary invariant case}

\paragraph{Macroscopic moments and cumulants.}
\begin{lemma}
For a \emph{global unitary invariant} tensor $A$, the finite size precursors of the free cumulants of $A_{(1_D)}$ determine the macroscopic moments and cumulants of $A$:
\begin{align}
\label{eq:E-for-global-univ-mixed}
\E\left[\Tr_{\bsig}(A)\right] &= \sum_{\eta\in \Sym_{n}} N^{\sum_{c=1}^D \#(\sigma_c\eta^{-1})}  \sum_{ \pi \ge \Pi(\eta)} \
\mathcal{K}^{\mix }_{\pi,\eta}[A] \;,\\
\label{eq:Phi-for-global-univ-mixed}
    \Phim_{ \bsig} [A]  &= \sum_{\eta \in \Sym_{n}}  N^{\sum_{c=1}^D \#(\sigma_c\eta^{-1})}   \sum_{\substack{{ \pi \ge \Pi(\eta)}\\{\Pi(\bsig)\vee\pi = 1_n}}}  \  \mathcal{K}^{\mix }_{\pi,\eta}[A] \;,
\end{align}
\end{lemma}

This results shows that  the finite size moments and cumulants of any global unitary mixed random tensor resemble the ones of the the Wishart tensors  \eqref{eq:Wishart}, in the sense that they involve an expansion over a single permutation $\eta$, but with a modified weight given by the finite $N$ precursors of the free cumulants.

For the Wishart tensor $W$ of parameters $(N^D, N^{D'})$,
we identify for $\eta\in \Sym_n$:
\begin{equation}
\label{eq:precursors-wishart-general}
    \Kp_{\eta}[W] = N^{D' - n(D + D' - 1)} \delta_{K(\eta), 1}\;,
\end{equation}
for which the rightmost sum in \eqref{eq:E-for-global-univ-mixed} is one (for $\pi=\Pi(\eta)$), and the rightmost sum in \eqref{eq:Phi-for-global-univ-mixed} is $\delta_{K(\bsig, \eta), 1}$, so that one recovers \eqref{eq:Wishart}. The rescaled limit is given by
\begin{equation}
     \fcumm_{\eta}(w) =   \delta_{K(\eta), 1}\;,
\end{equation}
in accordance with (6.46) of \cite{collins_free_2025} for $D'=1$.

\paragraph{Asymptotics.} For $D=1$, if $\bm{A} = (A^{(1)}, \ldots, A^{(n)})$ is a unitarily invariant family of $N\times N$ random matrices in classical random matrix ensembles (GUE, Wishart with parameters of the same order, etc), then the expected  scaling behavior for the joint cumulants  $\Phim_{ \sigma} [\bm{A}]$  (up to a global rescaling of the components of the matrices) is:
\begin{equation}
\label{eq:matrix-scaling-D1}
    \Phim_{ \sigma} [\bm{A}] = N^{2-\#(\sigma)} \Bigl(\vphim_\sigma(\bm{a})+o(1)\Bigr)\;.
\end{equation}
Furthermore, it has been shown \cite{collins_second_2007} that the finite size precursors of the free cumulants of a unitarily invariant random matrix (from classical ensembles) scale as:
\begin{equation}
\label{eq:matrix-scaling-CUM}
    \Km_{ \sigma} [A] = N^{2-\#(\sigma) - n} \Bigl(\fcumm_\sigma(a)+o(1)\Bigr)\;.
\end{equation}

We have seen in Sec.~\ref{sub:Gaussian-scaling} that the classical cumulants of the Wishart tensors $W$ of parameters $(N^D, N^D)$ scale as $N^{n-\delta(\hat \bsig_D)}$, 
where as shown in Prop.~\ref{prop:expr-n-delta-no-min}, the scaling can be expressed as:
\[
n - \delta(\hat\bsig_D)= 2D(1 - K(\bsig))  + \#\bsig - n(D-1).
\]
The same scaling determines an upper bound on the scales  of the classical cumulants for any global unitary invariant random tensors that scales as \eqref{eq:matrix-scaling-D1} (with $N$ replaced by $N^D$), as we now show. From \eqref{eq:Wishart-asympt}, this scaling is sharp for the Wishart random tensor of parameters $(N^D, N^D)$. We recall the notation:
\[
\hat\bsig_D = (\bsig,  \mathrm{\bf id})=(\sigma_1, \ldots, \sigma_D, \mathrm{id}_n,\ldots, \mathrm{id}_n )\in \Sym_{n}^{2D}:
\]

The following theorem extends Prop.~6.2 of \cite{nechita_tensor_2025}. Their results only hold for the asymptotic moments, corresponding to the \emph{connected} $\bsig\in \Sym_n^D$, while our theorem stand for any $\bsig\in \Sym_n^D$.

\begin{theorem}
\label{thm:asymptotics-mixed-global-inv}
Consider a mixed global unitary invariant random tensor $A$ with $D\ge 1$ inputs such that $N^{D-1} A$ scales as \eqref{eq:matrix-scaling-D1}. Then for every $\bsig\in \Sym_n^D$, $n \ge 1$:
\begin{align}
\label{eq:scaling-phi-global-inv}
   \Phim_\bsig[A] &= N^{n - \delta(\hat\bsig_D)}
   \bigl(\vphim_\bsig (a) + o(1)\bigr)\;,\\
   \mathcal{K}^{\mix}_\bsig[A] &= N^{n - \delta(\hat\bsig_D) - nD} \bigl(\fcumm_\bsig(a) + o(1)\bigr)\;,
   \label{eq:scaling-K-global-inv}
\end{align}
where the tensorial free cumulants of $A$ coincide with its (matricial) free cumulants:
\begin{equation}
\label{eq:limit-kappa-global-inv}
    \fcumm_\bsig(a) = \left\{
    \begin{array}{ll}
       \fcumm_{\sigma}(a) & \textrm{ if } \bsig=(\sigma, \ldots, \sigma) \textrm{ for some  } \sigma\in \Sym_{n} \\
        0&  \textrm{ otherwise }
    \end{array}
\right. \; .
\end{equation}
The asymptotic cumulants of $A$ are computed from the free cumulants as:
\begin{equation}
\label{eq:limit-phi-global-inv}
    \vphim_\bsig (a)= \sum_{\substack{{\eta \in \Sym_n}\\{K(\bsig, \eta )=1}\\{d(\hat\bsig, \eta) = \delta(\hat\bsig_D)}}} \fcumm_{\Pi(\eta),\eta}(a)\;.
  \end{equation}
\end{theorem}

\begin{proof}
  Starting from \eqref{eq:Phi-for-global-univ-mixed}, one sees from \eqref{eq:matrix-scaling-CUM} that the term corresponding to $\pi$ in the rightmost sum scales as $N^{2D\#(\pi) - D\#(\eta) -nD - n(D-1)}$, so that the term that dominates the sum over $\pi\ge \Pi(\eta)$ is $\pi= \Pi(\eta)$. This imposes the condition $\Pi(\bsig, \eta)=1_n$:
\[
      \Phim_{ \bsig} [A]  = \sum_{\substack{{\eta \in \Sym_{n}}\\{K(\bsig, \eta)=1}}}  N^{\sum_{c=1}^D \#(\sigma_c\eta^{-1}) + D\#(\eta) - 2nD + n} \bigl( \fcumm_{\Pi(\eta),\eta}(a) + o(1)\bigr) \;.
\]
Writing the exponent of $N$ as $n - d(\hat \bsig, \eta)$ allows deducing  \eqref{eq:scaling-phi-global-inv} and \eqref{eq:limit-phi-global-inv}. The relations    \eqref{eq:scaling-K-global-inv} and \eqref{eq:limit-kappa-global-inv} follow from \eqref{eq:large-group-precursors-mixed}.
\end{proof}

\

With the same assumptions, we define the rescaled trace of  $A$ as:
\begin{equation}
\label{eq:rescaled-norm-mixed}
    \vphim(a) = \vphim_{\mathrm{id}_1}(a) = \lim_{N\rightarrow \infty} \frac 1 {N} \E\bigl[\Tr(A)\bigr]\;.
\end{equation}

\begin{corol}
Consider a mixed global unitary invariant random tensor $A$ with $D\ge 1$ inputs such that $N^{D-1} A$ scales sharply as \eqref{eq:matrix-scaling-D1}. Then for any  $\bsig\in \Sym_n^D$ with $K(\bsig)$ fixed to $K$, the $\Phi_\bsig^\mix$ that have the strongest scale $N^{1-(2D-1)(K - 1)}$ are the melonic  $\bsig$ with canonical pairing  the identity. Furthermore, if $\bsig$ is melonic with $K(\bsig)=1$, then:
\[
  \vphim_\bsig (a)= (\vphim(a))^n\;.
\]
\end{corol}
\proof The first statement is a direct consequence of Cor.~\ref{cor:melonic-wishart} for $D'=D$. The second statement results from the fact that there is a unique $\eta$ such that $d(\hat \bsig, \eta)=\delta(\hat \bsig)$, corresponding to the canonical pairing of $\hat \bsig$, that is, $\eta=\mathrm{id}_n$. From \eqref{eq:limit-phi-global-inv}, one has therefore
$ \vphim_\bsig (a)= (\fcumm_{\mathrm{id}_1}(a))^n$, and from \eqref{eq:limit-phi-global-inv} again for $n=1$: $\fcumm_{\mathrm{id}_1}(a)= \vphim(a)$. \qed

\

\begin{remark}
One can adapt the theorem above to include Wishart random matrices whose parameters are not of the same order. For parameters $(N^D, N^{D'})$, \eqref{eq:matrix-scaling-D1} and \eqref{eq:matrix-scaling-CUM} must be modified to (see \cite{Lionni2018} and \eqref{eq:precursors-wishart-general})
 \begin{align*}
     \Phim_{ \sigma} [A] &= N^{D + D' - D\#(\sigma)} \Bigl(\vphim_\sigma(a)+o(1)\Bigr)\;, \\
    \Km_{ \sigma} [A] &= N^{D + D' - D\#(\sigma) - nD} \Bigl(\fcumm_\sigma(a)+o(1)\Bigr)\;.
\end{align*}
which are indeed the scaling satisfied by $N^{D'-1}W$, where $W$ is a Wishart random matrix of parameters $(N^D, N^{D'})$. For a global unitary invariant random tensor satisfying these assumptions, \eqref{eq:limit-kappa-global-inv} still holds, and \eqref{eq:scaling-phi-global-inv} and \eqref{eq:scaling-K-global-inv} are modified to:
\begin{align*}
   \Phim_\bsig[A] &= N^{n - \delta(\hat\bsig_{D'})} \bigl(\vphim_\bsig (a) + o(1)\bigr)\;,\\
   \mathcal{K}^{\mix}_\bsig[A] &= N^{D'  -D(K(\bsig) -1)- n(D+D'-1) } \bigl(\fcumm_\bsig(a) + o(1)\bigr)\;.
\end{align*}
\end{remark}

\subsection{Universality in the pure and global unitary invariant case}

\paragraph{Macroscopic moments and cumulants.}
\begin{lemma}
For a \emph{global unitary invariant} pure random tensor $T$, the finite size precursors of the free cumulants of the random vector $T_{f_{1_D}}$ determine the macroscopic moments and cumulants of the tensor $T$:
\begin{align}
    \label{eq:E-for-global-univ-pure}
\E\left[\Tr_{\bsig}(T, \bar T)\right] &= \sum_{\eta\in \Sym_{n}} N^{\sum_{c=1}^D \#(\sigma_c\eta^{-1})}  \sum_{ \Pi \ge \Pi_\mathrm{p}(\eta)} \
\Kp_{\Pi,\eta}[T, \bar T] \;,\\
    \Phip_{ \bsig} [T, \bar T]  &= \sum_{\eta \in \Sym_{n}}  N^{\sum_{c=1}^D \#(\sigma_c\eta^{-1})}   \sum_{\substack{{ \Pi \ge \Pi_\mathrm{p}(\eta)}\\{\Pi_\mathrm{p}(\bsig)\vee\Pi = 1_{n, \bar n}}}}  \  \Kp_{\Pi,\eta}[T, \bar T] \;.
    \label{eq:Phi-for-global-univ-pure}
\end{align}
\end{lemma}

Here again, it shows  that  the finite size moments and cumulants of any global unitary pure random tensor  resemble the ones of the complex Gaussian \eqref{eq:Pure-C-Gaussian}, but with a modified weight given by the finite $N$ precursors of the free cumulants.

\paragraph{Asymptotics.} For $D=1$, if $T\in \mathbb{C}^{N^D}$ is a unitarily invariant  random vector, then $\Phip_{ \sigma} [T, \bar T] = \Phip_{ \mathrm{id}_n} [T, \bar T]$ since it is a class function for $\sim_\pure$. Up to a global rescaling in $N$ of the components of the vector, if there is no conditioning on the norm of $T$, then we expect  the cumulants  $\Phip_{\mathrm{id}_n} [T, \bar T]$ to scale as:
\begin{equation}
\label{eq:vector-scaling-D1}
    \Phip_{ \mathrm{id}_n} [T, \bar T] = N^D \bigl(\vphip_{ \mathrm{id}_n}(t, \bar t)+o(1)\bigr)\;,
\end{equation}
\begin{lemma}
Under this assumption, one has the following:
\begin{align}
\label{eq:vector-scaling-CUM}
    \Kp_{ \mathrm{id}_n} [T, \bar T] &= N^{D(1-n)} \bigl(\fcump_{\mathrm{id}_n}(t, \bar t)+o(1)\bigr)\;, \\
    \fcump_{\mathrm{id}_n}(t, \bar t) & = \sum_{\tau \in \Sym_{n}} \sum_{\substack{\Pi \succeq \Pi_{\pure}(\tau)     \\  \#( \Pi) = n + 1  -  \#(\tau)  \\ \Pi\vee\Pi_{\pure}(\mathrm{id}_n) = 1_{n, \bar n} }}\
 \vphip_{\Pi, \tau}(t, \bar t)\ \Gamma[\Pi\vee\Pi_{\pure}(\mathrm{id}_n)_{[n]}, \tau]\;.
\end{align}
\end{lemma}
\begin{proof} Starting from \eqref{eq:HOFC-def-TT} for $D=1$, and taking into account the scaling assumption \eqref{eq:vector-scaling-D1} and the asymptotics of the cumulant Weingarten function \eqref{eq:connected-Weingarten}, one has:
\begin{equation}
\begin{split}
  \Kp_{\mathrm{id}_n}[T,\bar T]  &= \sum_{\tau \in \Sym_{n}} \sum_{\substack{\Pi \in \PPart(n)\\ \Pi_{\pure}(\tau) \leq \Pi}}
  N^{D\#\Pi   +  D\#(\tau) - 2D(\#\Pi\vee\Pi_{\pure}(\mathrm{id}_n) - 1) - 2nD}\\ &\hspace{4cm}\times\Bigl( \vphip_{\Pi, \tau}(t, \bar t)\Gamma[\Pi\vee\Pi_{\pure}(\mathrm{id}_n)_{[n]}, \tau]+o(1)\Bigr)\;,
  \end{split}
\end{equation}
Using \eqref{eq:def-of-L} one may therefore rewrite the exponent of $N$ as
\begin{align}
-DL\bigl[\Pi, \Pi_\pure(\tau, \mathrm{id}_n); [n]\bigr] - D(\#\Pi\vee\Pi_{\pure}(\mathrm{id}_n) - 1) + D(1-n)\;, \ \mathrm{where:}\\
L\Bigl[\Pi, \Pi_\pure(\tau, \mathrm{id}_n); [n]\Bigr] = n - \#( \Pi) - \#(\tau) + \#(\Pi\vee\Pi_{\pure}(\mathrm{id}_n)) \ge 0\;.
\end{align}
The terms that dominate the sum over $\Pi$ are for $\Pi=\Pi_\pure(\tau)$ satisfying $\Pi_\pure(\tau, \mathrm{id}_n)=\Pi(\tau)=1_n$
\begin{equation}
  \Kp_{\mathrm{id}_n}[T,\bar T]  =  N^{ D(1-n)}  \sum_{\tau \in \Sym_{n}} \sum_{\substack{\Pi \succeq \Pi_{\pure}(\tau)    \\  \#( \Pi) = n + 1  -  \#(\tau)  \\ \Pi\vee\Pi_{\pure}(\mathrm{id}_n) = 1_{n, \bar n} }}\
 \Bigl( \vphip_{\Pi, \tau}(t, \bar t)\Gamma[\Pi\vee\Pi_{\pure}(\mathrm{id}_n)_{[n]}, \tau]+o(1)\Bigr)\;,
\end{equation}
which concludes the proof.
\end{proof}

Similarly to \eqref{eq:rescaled-norm-pure}, we define -- under this assumption -- the rescaled norm of $T$:
\begin{equation}
\label{eq:rescaled-norm-pure}
    \vphip(t, \bar t) = \vphip_{\mathrm{id}_1}(t, \bar t) = \lim_{N\rightarrow \infty} \frac 1 {N} \E\bigl[\lVert T \rVert^2\bigr]\;.
\end{equation}
\begin{theorem}
\label{thm:asymptotics-pure-global-inv}
Consider a pure global unitary invariant random tensor $T\in \mathbb{C}^{N^D}$  and such that the pure tensor $(N^{D-1} T, N^{D-1}\bar T)$ scales as \eqref{eq:vector-scaling-D1}. Then for every $\bsig\in \Sym_n^D$, $n \ge 1$:
\begin{align}
\label{eq:scaling-phi-global-inv-pure}
   \Phip_\bsig[T, \bar T] &= N^{n - \delta(\bsig)
   } \bigl(\vphip_\bsig (t, \bar t) + o(1)\bigr)\;,\\
   \mathcal{K}^\pure_\bsig[T, \bar T] &= N^{D -  n( 2D  - 1)} \bigl(\fcump_\bsig(t, \bar t) + o(1)\bigr)\;,
   \label{eq:scaling-K-global-inv-pure}
\end{align}
where the tensorial free cumulants of $A$ coincide with its (vectorial) free cumulants:
\begin{equation}
\label{eq:limit-kappa-global-inv-pure}
    \fcump_\bsig(t, \bar t) = \left\{
    \begin{array}{ll}
      \fcump_{\mathrm{id}_n(t,\bar t) }& \textrm{ if } \bsig=(\sigma, \ldots, \sigma) \textrm{ for some  } \sigma\in \Sym_{n}\\
        0&  \textrm{ otherwise }
    \end{array}
\right. \; .
\end{equation}
Furthermore, the asymptotically normalized pure random tensor $T/ \vphip(t, \bar t) $ has the same distribution as the complex Gaussian $T_\un$:
\begin{equation}
\label{eq:limit-phi-global-inv-pure}
    \vphip_\bsig (t, \bar t)= \bigl(\vphip(t, \bar t) \bigr)^n \vphip_\bsig (t_1, \overline{t_1})\;,
  \end{equation}
  so that the scaling in \eqref{eq:rescaled-norm-pure} is sharp if $\vphip(t, \bar t)\neq 0$.
\end{theorem}
\begin{proof}
  Starting from \eqref{eq:Phi-for-global-univ-pure}, one sees from \eqref{eq:vector-scaling-CUM} that the term corresponding to $\Pi$ in the rightmost sum scales as $N^{D\#(\Pi) -nD + n(1-D)}$, so that the term that dominates the sum over $\Pi\ge \Pi_\pure(\eta)$ is $\Pi= \Pi_\pure(\eta)$, with $\#(\Pi_\pure(\eta))=n$. This imposes the condition $\Pi_\pure(\bsig, \eta)=1_{n, \bar n}$:
\[
      \Phip_{ \bsig} [T, \bar T]  = \sum_{\substack{{\eta \in \Sym_{n}}\\{K_\pure(\bsig, \eta)=1}}}  N^{\sum_{c=1}^D \#(\sigma_c\eta^{-1})  + n(1-D)} \bigl( \fcump_{\Pi_\pure(\eta),\eta}(t, \bar t) + o(1)\bigr) \;.
\]
We may write the exponent of $N$ as $n - d(\bsig, \eta)$, which implies \eqref{eq:scaling-phi-global-inv-pure}. For any $\eta\in \Sym_n$, $\fcump_{\Pi_\pure(\eta),\eta}(t, \bar t) = \fcump_{\mathrm{id}_1}(t, \bar t)^n = \vphip(t, \bar t)^n$, so that
\[
      \lim_{N\rightarrow \infty} N^{\delta(\bsig) - n}\Phip_{ \bsig} [T, \bar T]  =  \vphip(t, \bar t)^n \sum_{\substack{{\eta \in \Sym_{n}}\\{K_\pure(\bsig, \eta)=1}}}  1 \;,
\]
where we recognize $\vphip_\bsig (t_1, \overline{t_1})$, \eqref{eq:asympt-cum-gaussienne}. This shows \eqref{eq:limit-phi-global-inv-pure}.  The relations   \eqref{eq:scaling-K-global-inv-pure} and \eqref{eq:limit-kappa-global-inv-pure} follow from \eqref{eq:large-group-precursors-pure}.
\end{proof}

\

For the standard complex Gaussian, we identify from \eqref{eq:E-for-global-univ-pure} and \eqref{eq:Phi-for-global-univ-pure} for $\sigma\in \Sym_n$:
\begin{equation}
\label{eq:precursor-standard-gaussian}
    \Kp_{\sigma}[T_\un, \bar{T_\un}] = N^{1-D} \delta_{n,1}\;,
\end{equation}
for which the rightmost sum in \eqref{eq:E-for-global-univ-pure} is one (for $\Pi=\Pi_\pure(\eta)$), and the rightmost sum in \eqref{eq:Phi-for-global-univ-pure} is $\delta_{K_\pure(\bsig, \eta), 1}$, so that we indeed recover \eqref{eq:Pure-C-Gaussian}. Since $(N^{D-1} T_\un, N^{D-1}\bar T_\un)$ scales as \eqref{eq:vector-scaling-D1}, we find:
\begin{equation}
 \fcump_\bsig(t_1, \bar t_1) = \delta_{n,1}\;,
\end{equation}
in accordance with (6.3) of \cite{collins_free_2025}.

\section{Tensorial free cumulants and the matrix product scaling}
\label{sec:tensors-matrix-product-scaling}

In this section, we study mixed tensors whose classical cumulants have the same scaling in $N$ as tensor products of independent random matrices, called ``matrix product scaling''. This scaling was considered  by Nechita and Park in \cite{nechita_tensor_2025}. In this paper, the authors introduce a notion of tensorial free cumulants for all connected $\bsig\in \Sym_n^D$, and study their properties. As far as we understand, this notion is postulated, and not obtained as a limit of finite size quantities.  Here, we will show that these relations are obtained as limits of the finite size relations of Sec.~\ref{subsub:mixed-precursors} for connected invariants, and extend them to asymptotic cumulants for arbitrary non-connected $\bsig\in \Sym_n^D$, in analogy to higher order free cumulants \cite{collins_second_2007}.

The same scaling was considered in \cite{collins_tensor_2023} for random tensor taking uniformly at random in LU orbits, for which the tensorial free cumulants $\kappa_\bsig^\mix$ coincide with the $\lambda^\mix_\bsig$ derived from the tensor HCIZ integral, and computed for any tensor in the orbit. In that case, the same formulae as in \cite{nechita_tensor_2025} had been  obtained in \cite{collins_tensor_2023} asymptotically for first order (melonic) invariants. In the general case of a mixed random tensor satisfying the matrix product scaling ansatz, the aforementioned fact implies a generalization of the relation that links the asymptotics of the free energy of the HCIZ integral and the primitive of the R-transform for random matrices. We detail this below.

\subsection{The matrix product scaling}

We consider a mixed LU-invariant random tensor $A$ scaling as:
\begin{equation}
\label{eq:matrix-scaling}
    \Phim_{ \bsig} [A] = N^{2(1-K(\bsig)) + \# \bsig} \Bigl(\vphim_\bsig(a) + o(1)\Bigr),
\end{equation}
where we recall that $\#\bm{\sigma} = \sum_{c=1}^{D}\# \sigma_c$. For $D=1$, we indeed recover \eqref{eq:matrix-scaling-D1}.

If $A=A'$ or $A=UA'U^\dagger$ with $U=U_1\otimes \cdots \otimes U_D$, $U_c$ Haar distributed, one has \eqref{eq:mat-deter-corr}:
\begin{equation}
\label{eq:mat-deter-scaling-ass}
    \vphim_{\bsig} (a)=  \vphim_{\bsig} (a') = \delta_{1_n, \Pi(\bsig)} \lim_{N\rightarrow \infty} N^{-\# \bsig} \Tr_{\bsig}(A')\;,
\end{equation}
and more generally for any $\bsig\in \Sym_{n}^{D}$:
\begin{equation}
\label{eq:rescaled-limit-deter}
    \lim_{N\rightarrow \infty} N^{-\# \bsig} \Tr_{\bsig}(A') =  \vphim_{\Pi(\bsig), \bsig} (a')\;.
\end{equation}

\begin{lemma}
\label{lem:large-N-facto-matrix-scaling}
Consider a LU-invariant mixed random tensor $A$ satisfying the scaling assumption \eqref{eq:matrix-scaling}.  Then the moments of $A$ factorize asymptotically. Indeed, for any $\bsig\in \Sym_{n}^{D}$:
\[
\lim_{N\rightarrow \infty} \frac 1 {N^{\# \bsig}}\E\Bigl[\Tr_\bsig(A)\Bigr] =  \vphim_{\Pi(\bsig), \bsig}(a)\;,
\]
where if $\bsig_1, \ldots, \bsig_{K(\bsig)}$ are the connected components of $\bsig$, then $\vphim_{\Pi(\bsig), \bsig}(a) = \prod_{i=1}^{K(\bsig)} \vphim_{\bsig_i}(a)$.
\end{lemma}
\proof The term corresponding to $\pi$ in the classical moment-cumulant formula \eqref{eq:classical-mom-cum} scales as $N$ to the power $2(\#(\pi) - K(\bsig)) + \#\bm{\sigma}$. The dominant term is therefore obtained for $\pi=\Pi(\bsig)$. \qed

\

\begin{lemma}
\label{lem:facto-over-cycles-tens-prod}
Consider $A_1, \ldots, A_D$ some $N\times N$ unitarily invariant random matrices behaving asymptotically as \eqref{eq:matrix-scaling-D1}. Then $A=\bigotimes_c A_c$ satisfies \eqref{eq:matrix-scaling}. If the $A_c$ are all independent and $\bsig$ is connected, then 
\[
\vphim_\bsig(a) =  \prod_{c=1}^D  \vphim_{\Pi(\sigma_c), \sigma_c}(a_c) = \prod_{c=1}^D \prod_{\eta_c\in \sigma_c} \vphim_{\eta_c}(a_c) \;,
\]
where $\vphim_{\eta_c}(a_c) = \lim_{N}\frac 1 N \E\Bigl[\Tr(A_c^{\ell(\eta_c)})\Bigr]$, in which $\ell(\eta_c)$ is the number of elements in the cycle $\eta_c$.
\end{lemma}
\proof Consider $\bsig\in \Sym_{n}^{D}$ with connected components $\bsig^{(1)},\ldots, \bsig^{(p)}$, $p=K(\bsig)\ge 1$. We have:
\[
\Phim_{ \bsig} [A] = k_p\biggl(\biggl\{
\prod_{c=1}^D \prod_{ \eta_c^{(i)} \in \sigma_c^{(i)} }
\Tr(A_c^{  \ell(\eta_c^{(i)})    }   )\biggr\}_{1\le i \le p}\biggr)\;,
\]
where the products are over the disjoint cycles $\eta_c^{(i)}$ of $\sigma_c^{(i)}$.
Applying the formula of Leonov and Shiryaev \cite{Leo-Shir}:
\[
\Phim_{ \bsig} [A] = \sum_{\substack{{\pi \in \mathcal{P}(\bsig)}\\{\pi\vee \pi_0= 1}}} \prod_{B\in \pi} k_{\lvert B\rvert}\biggl(\biggl\{
\Tr(A_c^{  \ell(\eta_c^{(i)})    }   )\biggr\}_{\eta_c^{(i)} \in B_c }\biggr)\;,
\]
where $\mathcal{P}(\bsig)$ is the set of partitions of the $\sum_{c,i}\#(\sigma_c^{(i)}) = \#\bm{\sigma}$ disjoint cycles of the $\sigma_c^{(i)}$, $\pi_0\in \mathcal{P}(\bsig)$ has $p$ blocks, each gathering the cycles of the $\sigma_c^{(i)}$ for $i$ fixed and for all $1\le c \le D$, and $1$ is the one block partition in $\mathcal{P}(\bsig)$, and if $\pi\in \mathcal{P}(\bsig)$ and $^B\in\pi$, $B_c$ is $B$ restricted to the cycles of $\sigma_c^{(i)}$.  Since the $A_c$ satisfy \eqref{eq:matrix-scaling-D1}, each term in the sum scales as $N$ to the power $2\#(\pi) - \#\bm{\sigma} $. Under the constraint $\pi\vee \pi_0= 1$, one has applying \eqref{eq:def-of-L}:
\begin{equation}
   L\bigl[\pi, \pi_0; 0_\bsig\bigr]  =   \#\bm{\sigma} - p - \#(\pi) + 1\ge 0\;,
\end{equation}
where $0_\bsig$ is the trivial  partition in $\mathcal{P}(\bsig)$, and the inequality is saturated for some $\pi$. Therefore, $\Phim_{ \bsig} [A]$ scales as $ 2 (1- p + \#\bm{\sigma}) - \#\bm{\sigma} = 2(1-K(\bsig)) + \#\bm{\sigma}$.

If the $A_c$ are independent and $\bsig$ is connected, $p=1$ so that  $\#\bm{\sigma}  - \#(\pi) \ge 0$ with equality if and only if $\pi$ is the trivial partition $0_\bsig$:
\[
\Phim_{ \bsig} [A] = \Bigl(\prod_{c=1}^D \prod_{\eta_c\in \sigma_c} \E\Bigl[
\Tr A_c^{  \ell(\eta_c^{(i)})    } \Bigr]+o(1)\Bigr)\;,
\]
which concludes the proof.
\qed

\paragraph{Order of dominance.}As explained in Sec.~\ref{subsub:asympt-def-of-distrib}, we call first order or dominant the $\bsig$ for which the matrix scaling is maximal (the exponent of $N$ in \eqref{eq:matrix-scaling}), up to a rescaling of the tensors $A$. According to Thm.~\ref{thm:degree}, one has:
\begin{equation}
\label{eq:matrix-scaling-2}
    \Phim_{ \bsig} [A] = N^{1 + n(D-1) -(K(\bsig) - 1) - \Omega(\bsig)} \vphim_\bsig(a)(1+o(1))\;,
\end{equation}
so that  the first-order $\bsig$ are connected and melonic and such that the canonical pairing is the identity. We let:
\begin{equation}
    \mathbb{M}_n^D = \bigl\{ \bsig\in \Sym_{n}^{D} \ \mid \ K(\bsig)=1 \quad \mathrm{and} \quad \Omega(\bsig)=0   \bigr\}\;.
\end{equation}
The second-order $\bsig$ are either two connected melonic graphs with the same labeling condition, or one connected graph with $\Omega(\bsig)=1$
(see \cite{bonzom_diagrammatics_2017,fusy_combinatorial_2020}), and so on.

We say that the distribution of $A$ is characterized asymptotically at order $k$ by the data of the $\vphim_{\bsig}(a)$, for $\bsig\in \Sym_{n}^{D}$, $n\in\N^{*}$ satisfying:
\begin{equation}
    \Omega(\bsig) + K(\bsig) \le k\;.
\end{equation}

\subsection{Tensorial free cumulants of arbitrary order}

The non-negative combinatorial quantities $L$ and $L_D$ have been respectively defined in \eqref{eq:def-of-L} and  \eqref{eq:def-of-LD}. One has in particular\footnote{In  \cite{collins_tensor_2023}, this quantity has been introduced  as $\Box_\bsig(\bsig, \btau)$.}   for $\bsig, \btau\in \Sym_{n}^{D}$:
\begin{align}
\label{eq:special-L-HOFC}
   L\bigl[\Pi(\btau, \bsig), \pi ; \Pi(\btau)\bigr] &= K(\btau) - K(\bsig, \btau) -\#(\pi) +\#(\pi\vee \Pi(\bsig))\;,\\
\label{eq:special-LD-HOFC}
    L_D\Bigl[\bigl\{\Pi(\sigma_c, \tau_c)\bigr\}, \Pi(\bsig); \bigl\{\Pi(\sigma_c)\bigr\}\Bigr] & =\sum_{c=1}^D \bigl(\#(\sigma_c) - \#(\Pi(\sigma_c, \tau_c)) \bigr)  - K(\bsig)  +   K(\bsig, \btau)   \;.
\end{align}
The genus $g$ has been  introduced in \eqref{eq:Euler}. We define the following sets:
\begin{equation}
    \planar(\bsig)=\Bigl\{\btau \in \Sym_{n}^{D} \quad \mid \quad  \forall c,\  g(\sigma_c, \tau_c)=0 \quad \textrm{and}\quad L_D\bigl[\bigl\{\Pi(\sigma_c, \tau_c)\bigr\}, \Pi(\bsig); \{\Pi(\sigma_c)\}\bigr] = 0 \Bigr\}\;,
\end{equation}
which is not empty since for any $\bsig\in \Sym_{n}^{D}$, $\bsig\in \mathbb{G}_\mix(\bsig)$. We also define:\footnote{Note that $\mathbb{P}_\mix(\bsig, \btau)$ actually only depends on $\Pi(\bsig)$ and $\Pi(\btau)$. }
\begin{equation}
    \mathbb{P}_\mix(\bsig, \btau)=\Bigl\{\pi \in \mathcal{P}(n) \ \mid \  \pi \ge \Pi(\btau) \quad \textrm{and}\quad    L\bigl[\Pi(\btau, \bsig), \pi ; \Pi(\btau)\bigr]  = 0 \Bigr\}\;.
\end{equation}
It is also non-empty, since for any $\bsig, \btau\in \Sym_{n}^{D}$, $\Pi(\btau)\in \mathbb{P}_\mix(\bsig, \btau)$.
These two sets allows us to define the notion of \emph{forest of permutations}. At this point the definition is simple. It will be more involved in the pure case (see Definition \ref{def:forest-permutation-pure}).
\begin{definition}[Mixed forest of permutations]\label{def:forest-perm-mix}
  Let $n \in \N^{*}$ and $\bm{\sigma} \in \Sym_{n}^{D}$. The set of mixed forests
  of permutations on $\bm{\sigma}$ in the mixed case is
  \begin{equation*}
    \FSymm(\bm{\sigma}) = \Bigl\{ (\bm{\rho}, \Pi) \in \planar(\bm{\sigma}) \times \Partition(n) \colon \Pi \in \treeSet(\bm{\sigma}, \bm{\rho})\Bigr\}.
  \end{equation*}
\end{definition}

The following theorem generalizes the relations between asymptotic cumulants and higher-order free cumulants for unitarily invariant random matrices ($D=1$). More precisely, it generalizes the formulation of these relations given in \cite[Theorem 3.1]{lionni_higher_2022}. The combinatorial quantity $\Gamma$ has been defined in \eqref{eq:def-gamma}.
\begin{theorem}
\label{thm:asymptotics-mixed-tensorial}
For $A$ mixed with $D\ge 1$ inputs, LU-invariant, and that scales like \eqref{eq:matrix-scaling}, we have:
\begin{equation}
\label{eq:asympt-scaling-of-mixed-tensorial-free-cumulants}
   \mathcal{K}^{\mix}_\bsig[A] = N^{2(1-K(\bsig)) + \#\bsig - nD} \Bigl(\fcumm_\bsig (a) + o(1)\Bigr)\;,
\end{equation}
where the tensorial free cumulants of $A$ of arbitrary order are given
by the relations:
\begin{equation}
\label{eq:tens-free-cum-matrix-scaling-HOFC}
\fcumm_\bsig (a) = \sum_{\substack{{\btau\in \planar(\bsig)}}}\ \sum_{\pi \in \mathbb{P}_\mix(\bsig, \btau)} \vphim_{\pi,\btau}(a)\  \Gamma\bigl[\pi \vee \Pi(\bsig), \bsig \btau^{-1}\bigr]\;,
\end{equation}
where the sums are non-empty. The inverse relations are given by:
\begin{equation}
\label{eq:tens-free-cum-mom--matrix-scaling-HOFC}
    \vphim_\bsig (a)= \sum_{\btau \in \planar(\bsig)}\ \sum_{\substack{{\pi \in \mathbb{P}_\mix(\bsig, \btau)}\\{\pi \vee \Pi(\bsig)=1_n}}} \fcumm_{\pi,\btau}(a)\;,
  \end{equation}
  where the sums are non-empty. For connected $\bsig$, these relations read:
\begin{align}
\label{eq:tens-free-cum-mom--matrix-scaling}
    \vphim_\bsig(a) &= \sum_{\btau \leq \bsig } \fcumm_{\Pi(\btau), \btau}(a)\;,\\
\label{eq:tens-free-cum-matrix-scaling}
    \fcumm_\bsig(a) &= \sum_{\btau \leq \bsig } \vphim_{\Pi(\btau), \btau}(a) \mathsf{M}(\bsig\btau^{-1})\;.
\end{align}
\end{theorem}

If $A'$ is deterministic and if $\bsig$ is connected, these relations read:
\begin{equation}
\label{eq:tens-free-cum-mom--matrix-scaling-det}
\vphim_\bsig(a') = \sum_{\btau \leq \bsig } \lambda^\mix_{\Pi(\btau), \btau}(a')\;,\qquad     \lambda^\mix_\bsig(a') = \sum_{\btau \leq \bsig } \vphim_{\Pi(\btau), \btau}(a') \mathsf{M}(\bsig\btau^{-1})\;,
\end{equation}
where $\lambda^\mix_\bsig(a')$ is the rescaled limit of $\mathcal{L}^\mix_{\bsig} [A'] $ defined in \eqref{eq:cumulants-of-HCIZ}.

\

If $A'=\un$,  the $N^D\times N^D$ identity matrix, one has from \eqref{eq:mixed-ex-identity-matrix} for $\bsig$ connected:
\begin{equation}
\label{eq:connected-case-ID}
\vphim_\bsig(1) = 1\;, \hspace{1cm}
    \lambda^{\mix}_\bsig(1) =  \delta_{n, 1}\;,
\end{equation}
where we use the notation $a'=1$.

If $A=UA'U^\dagger$ with $U=U_1\otimes \cdots \otimes U_D$, $U_c$ Haar distributed, then $\fcumm_\bsig(a) = \lambda^\mix_\bsig(a')$.

\

The proof of the theorem relies on the following lemma:

\begin{lemma}\label{lem:sum-exp-mix}
  Let $n \in \N^{*}$, $\pi \in \Partition(n)$, and
  $\bm{\rho}, \bm{\sigma} \in \Sym_{n}^{D}$ with $\Pi(\bm{\rho}) \leq \pi$. We
  have
  \begin{equation*}
    2(\# \pi - K(\bm{\rho})) + \# \bm{\rho} \leq 2(\# \pi \vee \Pi(\bm{\sigma}) - K(\bm{\sigma})) + d(\bm{\sigma}, \bm{\rho}) + \# \bm{\sigma},
  \end{equation*}
  with equality if and only if $(\bm{\rho}, \pi) \in \FSymm(\bm{\sigma})$,
  i.e. $\bm{\rho} \in \planar(\bm{\sigma})$ and
  $\pi \in \treeSet(\bm{\sigma}, \bm{\rho})$.
\end{lemma}
\begin{proof}
  We start by noticing that by Euler's formula \eqref{eq:Euler},
  \begin{equation*}
    \begin{split}
      \# \bm{\rho} - d(\bm{\rho}, \bm{\sigma})
      &= \sum_{c=1}^{D}\Bigl( \# \rho_{c} + \# \sigma_{c} - n + \# (\rho_{c}\sigma_{c}^{-1})\Bigr) - \sum_{c = 1}^{D}\# \sigma_{c}\\
      &= 2\sum_{c=1}^{D}\Bigl( K(\rho_{c}, \sigma_{c}) - g(\rho_{c}, \sigma_{c})\Bigr) - \sum_{c = 1}^{D}\# \sigma_{c}.
    \end{split}
  \end{equation*}
  Then, we use the notation $L_{D}$ introduced in \eqref{eq:def-of-LD} and write
  \begin{equation*}
      \sum_{c = 1}^{D}\#\sigma_{c}-2L_{D}\Bigl(\{\Pi(\rho_{c}, \sigma_{c})\}, \Pi(\bm{\sigma}); \{\Pi(\sigma_{c})\}\Bigr) - 2K(\bm{\sigma}) +2K(\bm{\rho}, \bm{\sigma})
      = 2\sum_{c=1}^{D} K(\rho_{c}, \sigma_{c})-\sum_{c = 1}^{D}\#\sigma_{c}.
  \end{equation*}
  Finally, we use the notation $L$ (see \eqref{eq:def-of-L}) and get
  \begin{equation*}
    -2L\Bigl(\Pi(\bm{\rho}, \bm{\sigma}), \pi; \Pi(\bm{\rho})\Bigr) + 2 \# \Bigl(\Pi(\bm{\sigma}) \vee \pi\Bigr)  = 2\Bigl(\# \pi - K(\bm{\rho})\Bigr) + 2K(\bm{\rho}, \bm{\sigma}).
  \end{equation*}

  Putting everything together, we end up with
  \begin{equation}\label{eq:L-int-mix}
    \begin{split}
      2\Bigl(\# \pi - &K(\bm{\rho})\Bigr) + \# \bm{\rho}
      = 2 \Bigl(\# (\Pi(\bm{\sigma}) \vee \pi) - K(\bm{\sigma})\Bigr) + d(\bm{\rho}, \bm{\sigma}) + \# \bm{\sigma}\\
      &-2\sum_{c = 1}^{D}g(\rho_{c}, \sigma_{c}) -2L\Bigl(\Pi(\bm{\rho}, \bm{\sigma}), \pi; \Pi(\bm{\rho})\Bigr) -2L_{D}\Bigl(\{\Pi(\rho_{c}, \sigma_{c})\}, \Pi(\bm{\sigma}); \{\Pi(\sigma_{c})\}\Bigr).
    \end{split}
  \end{equation}
  We observe that the quantity appearing in the second line of
  \eqref{eq:L-int-mix} are non-positive. They are zero if and only if
  $\bm{\rho} \in \planar(\bm{\sigma})$ and $\pi \in \treeSet(\bm{\sigma}, \bm{\rho})$.
\end{proof}

\

\begin{proof}[Proof of Theorem \ref{thm:asymptotics-mixed-tensorial}]
  We start from the expression \eqref{eq:HOFC-def-C} of the finite $N$
  precursors $\Km_\bsig[A]$. The tensor $A$ is assumed
  to scale as in \eqref{eq:matrix-scaling}, and the asymptotics of
  $\WeingCm{N}$ are given in Theorem
  \ref{thm:asympt-cumulant-weingarten-funct}, leading to:
  \begin{equation}\label{eq:int-proof-of-HOFC-C}
    \begin{split}
      &\Km_\bsig[A]  = \sum_{\btau \in \Sym_{n}^{D}} \sum_{\pi \ge \Pi(\btau)} N^{2(\#(\pi)-K(\btau)) + \sum_{c=1}^D \#(\tau_c) + \sum_{c=1}^D \#(\sigma_c\tau_c^{-1}) - 2(\#(\pi\vee \Pi(\bsig)) - 1) - 2nD}\\ &\hspace{7cm}\times \Bigl(\vphim_{\pi, \btau}(a)\  \Gamma\bigl[\pi\vee \Pi(\bsig), \bsig\btau^{-1}\bigr]+o(1)\Bigr)\;.
    \end{split}
  \end{equation}
  We rewrite the exponent of $N$ using Lemma \ref{lem:sum-exp-mix}. We get
  \begin{equation*}
    \begin{split}
      2(\#(\pi)-K(\btau)) + \sum_{c=1}^D \#(\tau_c) &+ \sum_{c=1}^D \#(\sigma_c\tau_c^{-1}) - 2(\#(\pi\vee \Pi(\bsig)) - 1) - 2nD\\
                                                    &= 2(\#(\pi)-K(\btau)) - nD + \# \bm{\tau} - d(\bm{\sigma}, \bm{\tau}) - 2(\#(\pi\vee \Pi(\bsig)) - 1)\\
      &\leq 2(1 - K(\bm{\sigma})) - nD + \# \bm{\sigma},
    \end{split}
  \end{equation*}
  with equality if and only if $(\bm{\tau}, \Pi) \in \FSymm(\bm{\sigma})$. This yields the first claim
  \begin{equation*}
    \mathcal{K}^{\mix}_\bsig[A]  = N^{2(1 - K(\bm{\sigma})) + \# \bm{\sigma}- nD)}\sum_{(\bm{\tau}, \Pi) \in \FSymm(\bm{\sigma})}\vphim_{\pi, \btau}(a)\  \Gamma\bigl[\pi\vee \Pi(\bsig), \bsig\btau^{-1}\bigr]+o(1)\;.
  \end{equation*}
  Since for any $\bsig, \btau \in \Sym_{n}^{D}$ the sets $\planar(\bm{\sigma})$ and $\treeSet(\bsig, \btau)$ are
  non-empty, this gives
  \eqref{eq:asympt-scaling-of-mixed-tensorial-free-cumulants} and
  \eqref{eq:tens-free-cum-matrix-scaling-HOFC}.

  On the other hand, for the $ \Phim$ one has from
  \eqref{eq:Phi-in-terms-of-precursors-mixed} and
  \eqref{eq:asympt-scaling-of-mixed-tensorial-free-cumulants}:
\begin{equation}
    \Phim_{ \bsig} [A]  = \sum_{\btau \in \Sym_{n}^{D}}\  \sum_{\substack{{ \pi \ge \Pi(\btau)}\\{\Pi(\bsig)\vee\pi = 1_n}}}  \  N^{2(\#(\pi)-K(\btau)) + \sum_c\#(\tau_c) + \sum_c\#(\sigma_c\tau_c^{-1}) - nD} \Bigl(\fcumm_{\pi, \btau} (a)+o(1)\Bigr) \  \;.
\end{equation}
Subtracting $nD$ to the exponent of $N$, one recover the exponent in
\eqref{eq:int-proof-of-HOFC-C}, under the constraint
$\Pi(\bsig)\vee\pi = 1_n$. It remains to show that there exists at
least one $\btau \in \planar(\bsig)$ and
$\pi\in \treeSet(\bsig, \btau)$ such that $\pi\vee \Pi(\bsig)=1$. One
can verify that $\btau=\bsig$ does not work in general. Given
$\bsig\in \Sym_{n}^{D}$, we choose $\tau_1\in \Sym_{n}$ that connects
$\bsig$ in a minimal way: for every connected component
$B\in\Pi(\bsig)$, we choose one cycle $\eta_B\in\sigma_1$ of
$\sigma_1$ whose support is in $B$ and let $\eta$ be the disjoint
union of the $\eta_B$. One has $\#(\eta)=K(\bsig)$ We choose a cyclic
permutation $\gamma$ on the support of $\eta$, such that $\eta$ forms
a non-crossing permutation on $\gamma$:
$\#(\eta) + \#(\eta\gamma^{-1}) = 1+\ell$, where
$\ell = \sum_{B\in \Pi(\bsig)} \ell(\eta_B)$ is the total number of
elements in the cycles $\{\eta_B\}$. We then define $\tau_1$ as the
disjoint union of $\eta$ and of the cycles of $\sigma_1$ that differ
from the $\eta_B$. One has $K(\bsig, \btau)=1$, and
$K(\sigma_1,\tau_1) = \#(\tau_1) = \#(\sigma_1) - K(\bsig) + 1$, and
$\sigma_1\preceq \tau_1$:
$\#(\sigma_1) + \#(\sigma_1\tau_1^{-1})= \#(\tau_1) + n $. In
particular, $g(\sigma_1, \tau_1)=0$.
Furthermore, $L_D[\{\Pi(\sigma_c, \tau_c)\}, \Pi(\bsig); \{\Pi(\sigma_c)\}] = \#(\sigma_1) - K(\sigma_1,\tau_1) - K(\bsig) + 1 = 0$, so that $\btau\in\mathbb{G}_\mix(\bsig)$. With this choice, $\Pi(\tau) \in \mathbb{P}_\mix(\bsig, \btau)$. The dominant scaling in $N$ is therefore coherent with \eqref{eq:matrix-scaling}, and the rescaled limit is \eqref{eq:tens-free-cum-mom--matrix-scaling-HOFC}.
\end{proof}

\

\begin{remark}
\label{rk:order-contrib-micro}
In Theorem \ref{thm:asymptotics-mixed-tensorial}, it is shown that the dominant scale in $N$ of the finite size precursors $\Km_\bsig[A]$ is the same as that of $\Phim_{\bsig}[A]$ to a factor $N^{-nD}$. This shows in particular that
\begin{equation}
\label{eq:mprecursor-scaling-mixed}
    \Km_{\bsig} [A] = N^{1 -n -(K(\bsig) - 1) - \Omega(\bsig)} \Bigl(\fcumm_\bsig(a) + o(1)\Bigr)\;.
\end{equation}
In \eqref{eq:tensor-components-LU-inv-cum}, the joint cumulants of the tensor components are shown to expand on the $ \mathcal{K}^{\mix}_{\bsig} [A] $, so that the order of contribution of $\bsig$ as defined above can also be interpreted as the order of contribution in $N$ of $\bsig$ to the microscopic description of the distribution \eqref{eq:tensor-components-LU-inv-cum}.
\end{remark}

\begin{lemma}[see also Lemma 4.4 of \cite{nechita_tensor_2025}]
\label{lem:free-cum-tensor-product}
Consider $A_1, \ldots, A_D$ some $N\times N$ unitarily invariant random matrices behaving asymptotically as \eqref{eq:matrix-scaling-D1}. Then the tensorial free cumulants of  $A=\bigotimes_c A_c$ for $\bsig\in \Sym_n^D$ purely connected are given by:
\[
\fcumm_\bsig (a) = \prod_{c=1}^D \fcumm_{\Pi(\sigma_c), \sigma_c} (a_c).
\]
\end{lemma}
\begin{proof}Combining \eqref{eq:tens-free-cum-matrix-scaling} with  Lemma~\ref{lem:facto-over-cycles-tens-prod}, one has:
$
\fcumm_\bsig(a) = \prod_{c=1}^D\sum_{\tau_c \preceq \sigma_c } \vphim_{\Pi(\tau_c), \tau_c}(a) \mathsf{M}(\sigma_c\tau_c^{-1}),
$
where for each $c$ we recognize $\fcumm_{\Pi(\sigma_c), \sigma_c} (a_c)$.
\end{proof}

\subsection{Earlier results on tensorial free cumulants and the tensor HCIZ integral}

\paragraph{Purely connected tensorial free cumulants.}
As mentioned at the beginning of this section, in \cite{nechita_tensor_2025}, the authors define tensorial free cumulants for connected $\bsig$ through the relations \eqref{eq:tens-free-cum-mom--matrix-scaling} and \eqref{eq:tens-free-cum-matrix-scaling}. As far as we understand, these relations are stated in this form, and it is shown using these relations that independent LU-invariant mixed random tensors satisfying the matrix product scaling  are tensorially free, in the sense that free cumulants involving these different tensors vanish (see Thm.~7.7 of this reference for the precise statement).  This implies the additivity of tensorial free cumulants for  independent LU-invariant random tensors of this kind.

In Thm.~\ref{thm:asymptotics-mixed-tensorial}, we obtain the same relations as limits of the finite size relations, and  the additivity of tensorial free cumulants for  independent LU-invariant random tensors satisfying the matrix product ansatz is a consequence of Prop.~\ref{prop:add-finite-N}.

\paragraph{Melonic free cumulants and the tensor HCIZ integral.}
The relation on the right-hand side of \eqref{eq:finite-precursor-log-version-deterministic} was obtained in  \cite{collins_tensor_2023} for  $\bsig\in  \mathbb{M}_n^D $.\footnote{This is because \cite{collins_tensor_2023} focused on computing the large $N$ asymptotics of the cumulants of the HCIZ integral, which expand over the first order invariants only.} We formulate this in the following theorem, which generalizes the relation between the free-energy of the HCIZ integral and the primitive of the R-transform, given in the matrix case in \cite{collins_moments_2003} (see also \cite[Theorem 8.5]{collins_second_2007}, \cite{GUIONNET2002461,  GUIONNET2005435}).

\begin{theorem}[\cite{collins_tensor_2023}]
\label{thm:R-transform-matrix-scaling}
Consider $n\in \N^{*}$ and a  deterministic tensor $B$ with $D\ge 1 $ inputs such that for every $\bsig\in \Sym_{n}^{D}$,  $\lim_{N\rightarrow \infty} \Tr_\bsig(B) = t_\bsig <\infty$.
If $A$ is a mixed LU-invariant random tensor with $D$ inputs and that scales as \eqref{eq:matrix-scaling}:
\[
\lim_{N\rightarrow \infty }   \frac 1 N \frac{\partial^n}{\partial z^n} \E\biggl[\log \int [dU]\  \mathrm{e}^{zN \Tr(B^\mathsf{T} U A U^\dagger)} \biggr] \Biggr\lvert_{z=0} =
n!  \sum_{\bsig\in  \mathbb{M}_n^D / \sim_\mix } \frac{t_\bsig}{\# \Aut_{\mix}(\bm{\sigma})}\; \fcumm_\bsig(a)\;,
\]
where the sum is over elements in the quotient $\bsig\in \Sym_{n}^{D}/ \sim_\mix $ that are first order, i.e.~connected and such that $\Omega(\bsig)=0$.
If $B$ has only one non-vanishing element $B_{1,\ldots, 1 ; 1, \ldots, 1} = 1$, then for every $\bsig\in \Sym_{n}^{D}$, $t_\bsig=1$.
\end{theorem}
\begin{proof}
  If $A'$ is a deterministic tensor scaling as
  \eqref{eq:mat-deter-scaling-ass}, the term that dominates
  \eqref{eq:free-energy-HCIZ} when $N$ goes to infinity has been
  computed in Lem.~5.3 of \cite{collins_tensor_2023}, where it is
  shown that:
  \begin{equation}
    \lim_{N\rightarrow \infty }   \frac 1 N \frac{\partial^n}{\partial z^n} \log \int [dU]\  \mathrm{e}^{zN \Tr(B^\mathsf{T} U A' U^\dagger)}  \Biggr\lvert_{z=0} =  \sum_{\bsig\in  \mathbb{M}_n^D  } t_\bsig\; \lambda^\mix_\bsig(a')\;.
  \end{equation}

  Furthermore, using the automorphism groups introduced in \eqref{eq:def-automorphism}, we have for any $\bm{\sigma} \in \Sym_{n}^{D}$
  \begin{equation*}
    \# \Bigl\{\bm{\sigma'} \in \mathbb{M}_{n}^{D} \colon \bm{\sigma} \sim_{\mix} \bm{\sigma'}\Bigr\} = \frac{n!}{\# \Aut_{\mix}(\bm{\sigma})},
  \end{equation*}
  and the number $\#\Aut_{\mix}(\bm{\sigma})$ does not depend on the
representative in the equivalence class by $\sim_{\mix}$. Hence, we
get the result in the non-random case.

For the random case, the statement follows from
\eqref{eq:quantities-match-connected-mixed} which states that for
$\bsig$ connected (which is true for $\bsig\in \mathbb{M}_n^D$),
$\E[\mathcal{L}^{\mix}_\bsig[A]] = \Km_\bsig[A]$, so that their
rescaled limits match.

Finally, if $B$ has only one non-zero coefficient
$B_{1, \ldots, 1; 1, \ldots, 1} = 1$, then we immediately have that
$\Tr_{\bm{\sigma}}(B) = 1$.
\end{proof}

\

This statement formally implies that (as a formal series):
\begin{equation}
    \lim_{N\rightarrow \infty }   \frac 1 N  \E\biggl[\frac{\mathrm{d}}{\mathrm{d}z}\log \int [dU]\  \mathrm{e}^{zN \Tr(B^\mathsf{T} U A U^\dagger)} \biggr] = \sum_{n\ge 1} z^{n} \sum_{\bsig\in  \mathbb{M}_n^D / \sim_\mix }  \frac{t_\bsig}{\# \Aut_{\mix}(\bm{\sigma})}\; \fcumm_\bsig(a)\;,
\end{equation}
and similarly for the deterministic case.  For $D=1$, the first order $\sigma$ are cyclic permutations of $n$ elements, and for any such permutation one obtains the same result $\fcumm_\sigma(a) = \fcumm_n(a)$, which is the free cumulant of $A$ in the usual sense (see e.g.~\cite{nica_lectures_2006}). Since there are $(n-1)!$ such permutations, Thm.~\ref{thm:R-transform-matrix-scaling} simplifies for $B$ with only one one non-vanishing element $B_{1 ; 1} = 1$, to:
\begin{equation}
  \lim_{N\rightarrow \infty }   \frac 1 N \frac{\partial^n}{\partial z^n} \E\biggl[\log \int dU\  \mathrm{e}^{zN \Tr(B^\mathsf{T} U A U^\dagger)} \biggr] \Biggr\lvert_{z=0} = (n-1)!\; \fcumm_n(a)\;,
\end{equation}
or as  formal series:
\begin{equation}
    \lim_{N\rightarrow \infty }   \frac 1 N  \E\biggl[\frac{\mathrm{d}}{\mathrm{d}z}\log \int dU\  \mathrm{e}^{zN \Tr(B^\mathsf{T} U A U^\dagger)} \biggr] = R_A(z)\;,
\end{equation}
where $R_A(z) = \sum_{n\ge 1} z^{n-1}  \fcumm_n(a)$ is the R-transform of $A$, see for instance \cite{mingo_free_2017}.

\subsection{Relation to the tensor HCIZ integral in the disconnected case}

If $A=A'$ or $A=UA'U^\dagger$ with $U=U_1\otimes \cdots \otimes U_D$, $U_c$ Haar distributed, one has:
\begin{equation}
\label{eq:tens-free-cum-matrix-scaling-HOFC-det}
    \fcumm_\bsig (a) = \lambda^{\mix}_\bsig (a') = \sum_{\substack{{\btau\in \mathbb{G}_\mix(\bsig)}}}\  \vphim_{\Pi(\btau),\btau}(a')\  \Gamma\bigl[\Pi(\bsig, \btau), \bsig \btau^{-1}\bigr]\;.
\end{equation}
The inverse relations are given by:
\begin{equation}
\label{eq:tens-free-cum-mom--matrix-scaling-HOFC-det}
    \vphim_\bsig (a)= \vphim_\bsig (a')= \sum_{\btau \in \mathbb{G}_\mix(\bsig)}\ \sum_{\substack{{\pi \in \mathbb{P}_\mix(\bsig, \btau)}\\{\pi \vee \Pi(\bsig)=1_n}}} \lambda^\mix_{\pi,\btau}(a')\;.
\end{equation}

If $A'=\un$,  the $N^D\times N^D$ identity matrix, one has from \eqref{eq:mixed-ex-identity-matrix} for any $\bsig$:
\begin{equation}
\vphim_\bsig(1) = \delta_{1_n, \Pi(\bsig)}\;, \hspace{1cm}
    \lambda^{\mix}_\bsig(1) =  \delta_{n, 1}\;.
\end{equation}

\

For $\bsig$ satisfying $\Omega(\bsig)=0$, the expression \eqref{eq:tens-free-cum-matrix-scaling-HOFC-det} is obtained from the tensor HCIZ integral.  The following generalizes Thm.~\ref{thm:R-transform-matrix-scaling} in the case where $A$ is deterministic or uniform in a LU orbit.
\begin{theorem}
Consider two  deterministic tensors $A, B$ with $D\ge 1 $ inputs each, that scale as \eqref{eq:mat-deter-scaling-ass}, and  $n\in\N^{*}$. Then:
\[
\lim_{N\rightarrow \infty }   \frac 1 {N^2} \frac{\partial^n}{\partial z^n} \log \int [dU]\  \mathrm{e}^{zN^{2-D} \Tr(B^\mathsf{T} U A U^\dagger)}  \Biggr\lvert_{z=0} =
 \sum_{\substack{{\bsig \in \Sym_{n}^{D},}\\{\Omega(\bsig)=0}}} \vphim_{\Pi(\bsig), \bsig}(b) \; \lambda^\mix_\bsig(a)\;,
\]
where the sum is  over elements  $\bsig\in \Sym_{n}^{D}$ that are satisfy $\Omega(\bsig)=0$ but are not necessarily connected.
\end{theorem}
\proof Following the notations of \cite{collins_tensor_2023}, we introduce the following combinatorial quantity:
\begin{equation}
    \Delta(\bsig, \btau) = K(\bsig, \btau) + n(D-1) - \sum_{c=1}^D K(\sigma_c, \tau_c)\;.
\end{equation}
As stated in Prop.~4.8 of \cite{collins_tensor_2023}, one has that:
\[
\Delta(\bsig, \btau) = L_D\bigl[\bigl\{\Pi(\sigma_c, \tau_c)\bigr\}, \Pi(\bsig); \{\Pi(\sigma_c)\}\bigr]  + \Omega(\bsig)\;.
\]
Since both terms on the right-hand side are non-negative:

\begin{lemma}
\label{lem:Delta-vs-LD-Omega}
A pair $(\bsig, \btau) \in \Sym_{n}^{D} \times \Sym_{n}^{D}$ satisfies $\Delta(\bsig, \btau) = 0$  if and only if it satisfies:
\[
 \Omega(\bsig) = 0 \quad \textrm{and}\quad  L_D\bigl[\bigl\{\Pi(\sigma_c, \tau_c)\bigr\}, \Pi(\bsig); \{\Pi(\sigma_c)\}\bigr] = 0 \;,
\]
if and only if these conditions are satisfied with the roles of $\bsig$ and $\btau$ exhanged.
\end{lemma}

We define the set:
\begin{equation}
    \planar^{0, n}=\Bigl\{(\bsig, \btau) \in \Sym_{n}^{D} \times \Sym_{n}^{D} \quad \mid \quad  \forall c,\  g(\sigma_c, \tau_c)=0 \quad \textrm{and}\quad\Delta(\bsig, \btau) = 0 \Bigr\}\;.
\end{equation}

If $A, B$ scale as \eqref{eq:mat-deter-scaling-ass}, it was shown in Thm.~7 of of \cite{collins_tensor_2023} (regime S-II) that for $D\ge 1$:
\begin{equation}
\lim_{N\rightarrow \infty }   \frac 1 {N^2} \frac{\partial}{\partial z^n} \log \mathcal{I}_{D,N}\bigl[A ; zN^{2-D}B\bigr]  \Biggr\lvert_{z=0} = \sum_{(\bsig, \btau )\in  \planar^{0, n} } \vphim_{\Pi(\bsig), \bsig}(b) \vphim_{\Pi(\btau), \btau}(a)  \Gamma\bigl[\Pi(\bsig, \btau), \bsig \btau^{-1}\bigr]\;.
\end{equation}
Thanks to Lem.~\ref{lem:Delta-vs-LD-Omega}, we can split the sums over $\bsig$ and $\btau$ as:
\[
\lim_{N\rightarrow \infty }   \frac 1 {N^2} \frac{\partial}{\partial z^n} \log \mathcal{I}_{D,N}\bigl[A ; zN^{2-D}B\bigr] \Biggr\lvert_{z=0} =  \sum_{\substack{{\bsig \in \Sym_{n}^{D}}\\{\Omega(\bsig)=0}}} \vphim_{\Pi(\bsig), \bsig}(b) \sum_{\btau \in \planar(\bsig)}  \vphim_{\Pi(\btau), \btau}(a)  \Gamma\bigl[\Pi(\bsig, \btau), \bsig \btau^{-1}\bigr]\;,
\]
where we recognize $\lambda^{\mix}_\bsig (a)$ \eqref{eq:tens-free-cum-mom--matrix-scaling-HOFC-det}. \qed

\

On the right-hand side of the last formula from the proof, $ \vphim_{\Pi(\btau), \btau}(a) $ is  the rescaled limit of the trace-invariant $\Tr_\btau(A)$, see \eqref{eq:rescaled-limit-deter}.  If now $A$ is random, taking the expectation of this quantity, one has on the right-hand side the rescaled limit:
\[
    \lim_{N\rightarrow \infty} N^{-\# \bm{\tau}} \E\bigl[\Tr_{\btau}(A)] =  \vphim_{\Pi(\btau), \btau} (a)\;,
\]
thanks to the asymptotic factorization shown in Lem.~\ref{lem:large-N-facto-matrix-scaling}. The right-hand side then coincides with $\sum_{\bsig,\; \Omega(\bsig)=0} \vphim_{\Pi(\bsig), \bsig}(b)  \fcumm_\bsig(a)$ if the $\vphim_{\pi, \btau}(a)$ vanish for $\pi>\Pi(\btau)$, as in \eqref{eq:mat-deter-scaling-ass}.
Said otherwise, in that case, one can exchange the expectation and the logarithm.

\section{Tensorial free cumulants and Gaussian scalings }
\label{sec:Gaussian-scalings}

Starting from the example of the Gaussian tensor discussed in Section
\ref{sub:Gaussian-scaling}, we introduce two scaling assumptions for
the classical cumulants of pure random tensors. The two scaling
exponents, $\delta$ and $\exeps$, appear in a crucial way when computing the
free cumulants of pure tensors.

\subsection{Two scalings for pure tensors}
\label{sec:two-scalings-pure}

With the result for Gaussian tensors in mind (see Section
\ref{sub:Gaussian-scaling}), a natural scaling assumption to make is
to consider random LU-invariant pure tensors $T, \bar{T}$ whose
classical cumulants satisfy
\begin{equation}\label{eq:scaling-gaussian}
  \Phip_{\bm{\sigma}}[T, \bar{T}] = N^{n - \delta(\bm{\sigma})}\Bigl( \vphip_{\bm{\sigma}}(t, \bar{t}) + o(1) \Bigr).
\end{equation}

In Section \ref{sec:pure-random-Gaussian}, we give the leading
asymptotics of the finite-$N$ precursors to the free cumulants at
all-order. The case of the finite-$N$ precursors to the free cumulant
at first order was studied in details in \cite{collins_free_2025}.
More precisely, Collins, Gurau, and Lionni gave formulae to compute
the asymptotics of $\Kp_{\bm{\sigma}}[T, \bar{T}]$ when
$K_{\pure}(\bm{\sigma}) = 1$. We extend these results to the case of
any value of $\Kp_{\bm{\sigma}}[T, \bar{T}]$. The leading order of
$\Kp_{\bm{\sigma}}[T, \bar{T}]$ is then shown to be
$N^{n + 2 - \exeps(\bm{\sigma})}$, where the exponent
$\exeps(\bm{\sigma})$ is defined for any $n \in \N^{*}$ and
$\bm{\sigma} \in \Sym_{n}^{D}$ by
\begin{equation}\label{eq:exp-scaling-pure}
  \exeps(\bm{\sigma}) = \min_{\eta \in \Sym_{n}}\Bigl(2K_{\pure}(\bm{\sigma}, \eta) + d(\bm{\sigma}, \eta)\Bigr),
\end{equation}
and its multiplicative extensions: for any $\Pi \in \PPart(n)$
\begin{equation}\label{eq:mult-exeps}
  \exeps(\Pi, \bm{\sigma}) = \min_{\eta \in \Sym_{n}\\\Pi_{\pure}(\bm{\sigma}, \eta) \leq \Pi}\Bigl(2K_{\pure}(\bm{\sigma}, \eta) + d(\bm{\sigma}, \eta)\Bigr),
\end{equation}
The corresponding sets of minimizers are
\begin{equation}\label{eq:set-scal-pure}
  \begin{split}
    \scalp(\bm{\sigma}) &= \Bigl\{ \eta \in \Sym_{n} \colon \exeps(\bm{\sigma}) = 2K_{\pure}(\bm{\sigma}, \eta) + d(\bm{\sigma}, \eta) \Bigr\},\\
    \scalp(\Pi, \bm{\sigma}) &= \bigl\{\eta \in \Sym_{n} \colon \Pi(\bm{\sigma}, \eta) \leq \Pi, \exeps(\Pi, \bm{\sigma}) = 2K_{\pure}(\bm{\sigma}, \eta) + d(\bm{\sigma}, \eta)\bigr\}.
  \end{split}
\end{equation}
If $\Pi=1_{n, \bar n}$, we write $\scalp(\bm{\sigma}) =\scalp(1_{n, \bar n}, \bm{\sigma}) $.

We recall the definition given in
\eqref{eq:multiplicative-scaling-gaussien} for $\bsig\in \Sym_n^D$ and
$\Pi\ge \Pi_\pure(\bsig)$:
\[
  \delta(\Pi, \bm{\sigma}) = \min_{\substack{\eta \in \Sym_{n}\\\Pi_{\pure}(\bm{\sigma}, \eta) = \Pi}}d(\bm{\sigma}, \eta)
\]

The two scaling exponents $\delta$ and $\exeps$ (recall
\eqref{eq:scaling-gaussian} and \eqref{eq:scaling-pure}) of classical
cumulants agree if and only
$\delta(\bm{\sigma}) = \exeps(\bm{\sigma}) - 2$. This equality is
related to the notion of decreasing tuple in Lemma
\ref{lem:delta-eps-decreasing}.
\begin{definition}[Decreasing tuple]\label{def:decreasing-obs}
  Let $n \in \N^{*}$ and $\bm{\sigma} \in \Sym_{n}^{D}$. The $D$-tuple of
  permutations $\bm{\sigma}$ is said to be \emph{decreasing} if for all
  $\Pi \in \PPart(n)$ satisfying $\Pi_{\pure}(\bm{\sigma}) \leq \Pi$,
  \begin{equation*}
     2  + \delta(\bm{\sigma}) \leq 2\# \Pi + \delta(\Pi, \bm{\sigma}).
  \end{equation*}
  Given $\tilde{\Pi} \in \PPart(n)$ with
  $\Pi_{\pure}(\bm{\sigma}) \leq \tilde{\Pi}$, we say that $\bm{\sigma}$ is \emph{decreasing
  from} $\tilde{\Pi}$ if for every $\Pi \in \PPart(n), \Pi_{\pure}(\bm{\sigma}) \leq \Pi \leq \tilde{\Pi}$, we
  have
  \begin{equation*}
     2\# \tilde{\Pi} + \delta(\tilde{\Pi}, \bm{\sigma}) \leq 2\# \Pi + \delta(\Pi, \bm{\sigma})
  \end{equation*}
\end{definition}
\begin{lemma}\label{lem:delta-eps-decreasing}
  Consider $n\in \mathbb{N}^{*}$, $\Pi \in \PPart(n)$, and
  $\bsig\in \Sym_n^D$. We have
\begin{equation*}
  \delta(\Pi, \bsig) \ge \epsilon(\Pi, \bsig) - 2\# \Pi\;,
\end{equation*}
with equality if and only if $\bm{\sigma}$ is decreasing from $\Pi$.
\end{lemma}

A purely connected $\bm{\sigma} \in \Sym_{n}^{D}$ is automatically decreasing. This gives the following corollary.
\begin{corol}\label{corol:connected-means-decreasing}
  Consider $n\in \mathbb{N}^{*}$ and $\bsig\in \Sym_n^D$. If
  $K_\pure(\bsig)=1$, then $ \delta(\bsig) = \epsilon(\bsig) - 2$.
 \end{corol}
 \textit{Proof of the lemma.}
   We are going to show the Lemma in the case
   $\tilde{\Pi} = 1_{n, \bar{n}}$. In general, the result will follow
   from applying the result on the blocks of $\tilde{\Pi}$.

   To show the first statement, we write:
   \begin{equation*}
     \epsilon(\Pi, \bsig) =\min_{\substack{\Pi' \in \PPart(n)\\\Pi_{\pure}(\bm{\sigma}) \leq \Pi' \leq \Pi}} \min_{\substack{{\eta\in S_n}\\{\Pi_\pure(\bsig, \eta)=\Pi}}}\bigl[2\#\Pi' + d(\bsig, \eta)\bigr] =\min_{\substack{\Pi' \in \PPart(n)\\\Pi_{\pure}(\bm{\sigma}) \leq \Pi' \leq \Pi}} \bigl[2\#\Pi' + \delta(\Pi', \bsig)\bigr] \le 2\# \Pi + \delta(\Pi, \bsig) \;.
   \end{equation*}
   Equality holds if and only if the minimum is attained for
   $\Pi'=\Pi$, which concludes the proof.
 \qed

\subsection{Pure random tensors that scale like the standard complex Gaussian}
\label{sec:pure-random-Gaussian}

Let us now concentrate on pure random tensors $(T, \bar{T})$ that
satisfy the scaling assumption \eqref{eq:scaling-gaussian}.

\subsubsection{Tensorial free cumulants of arbitrary order}

The expression of the free cumulants are expressed using a
combinatorial object we call pure forest of permutations, defined
similarly as in the mixed case in Definition \ref{def:forest-perm-mix}.
\begin{definition}[Gaussian forest of
  permutations]\label{def:forest-permutation-pure} Let
  $n \in \N^{*}$, $\eta \in \Sym_{n}$, and $\bm{\sigma} \in \Sym_{n}^{D}$. A
  Gaussian forest of permutations on $(\bm{\sigma}, \eta)$ in the Gaussian
  case is a pair
  $(\bm{\rho}, \Pi_{\pure}(\bm{\rho}, \eta)) \in \Sym_{n}^{D} \times \PPart(n)$
  such that $\bm{\rho}\eta^{-1} \in \planar(\bm{\sigma}\eta^{-1})$.
  The set of Gaussian forests of permutations on $\bm{\sigma}$ is then
  defined by
  \begin{equation*}
    \FSymg(\bm{\sigma}) = \Bigl\{ (\bm{\rho}, \Pi) \in \Sym_{N}^{D} \times \Partition_{\pure}(n) \colon \exists \eta \in \scalp(\bm{\sigma}) \cap \scalg(\Pi, \bm{\rho}), \Pi = \Pi_{\pure}(\bm{\rho}, \eta), \bm{\rho}\eta^{-1} \in \planar(\bm{\sigma}\eta^{-1})\Bigr\}.
  \end{equation*}
\end{definition}
\begin{remark}[$\FSymg(\bm{\sigma})$ is not empty]
  The set $\FSymg(\bm{\sigma})$ is not empty in general: let
$\eta \in \scalp(\bm{\sigma})$ and take $\bm{\rho} = (\eta, \ldots, \eta)$. We immediately get that $\bm{\rho}\eta^{-1} = \bm{\symid_{n}} \in \planar(\bm{\sigma}\eta^{-1})$, and \(\eta \in \scalg(\Pi_{\pure}(\eta), \bm{\rho})\). Hence,
\begin{equation*}
  \Biggl\{\Bigl((\eta, \ldots, \eta), \Pi_\pure(\eta) \Bigr) \in \Sym_{n}^{D} \colon \eta \in \scalp(\bm{\sigma})\Biggr\} \subset \FSymg(\bm{\sigma}).
\end{equation*}
\end{remark}
\begin{theorem}\label{thm:finite-N-free-cumulants}
 Let $(T, \bar{T})$ be a pure random tensor such that for all $n \in \N^{*}$ and
  $\bm{\sigma} \in \Sym_{n}^{D}$ we have
  \begin{equation*}
    \Phip_{\bm{\sigma}}[T, \bar{T}] = N^{n - \delta(\bm{\sigma})}\Bigl( \vphip_{\bm{\sigma}}(t, \bar{t}) + o(1)\Bigr).
  \end{equation*}
  Then, we have for all $n \in \N^{*}$ and $\bm{\sigma} \in \Sym_{n}^{D}$
  \begin{equation*}
    \Kp_{\bm{\sigma}}[T, \bar{T}] = N^{n(1 - D) + 2 - \exeps(\bm{\sigma})}\Bigl(\fcump_{\bm{\sigma}}(t, \bar{t}) + o(1)\Bigr)
  \end{equation*}
  where
  \begin{equation*}
    \fcump_{\bm{\sigma}}(t, \bar{t}) = \sum_{(\bm{\rho}, \Pi) \in \FSymg(\bm{\sigma})}\vphip_{\Pi, \bm{\rho}}(t, \bar{t}) \ \Gamma\Bigl[\Pi_{\pure}(\bm{\sigma}) \vee \Pi, \bm{\sigma}\bm{\rho}^{-1}\Bigr].
  \end{equation*}

  Furthermore, whenever $K_{\pure}(\bm{\sigma}) = 1$, we have the inverse
  relation
  \begin{equation*}
    \vphip_{\bm{\sigma}}(t, \bar{t}) = \sum_{\substack{(\bm{\rho}, \Pi) \in \FSymg(\bm{\sigma})\\\bm{\rho} \text{ decreasing from } \Pi}}\fcump_{\Pi, \bm{\rho}}(t, \bar{t}).
  \end{equation*}
\end{theorem}
The proof relies on the following description of the exponent
$\delta(\bm{\sigma})$ using the set
$\planar(\bm{\sigma})$.
\begin{lemma}\label{lem:sum-exp-gauss}
  Let $n \in \N^{*}$, $\Pi \in \Partition_{\pure}(n)$, and
  $\bm{\rho}, \bm{\sigma} \in \Sym_{n}^{D}$ with
  $\Pi_{\pure}(\bm{\rho}) \leq \Pi$. For all
  $\eta \in \scalg(\Pi, \bm{\rho})$, we have
  \begin{equation*}
    - d(\bm{\rho}, \eta)
    \leq 2\Bigl(\#(\Pi_{\pure}(\bm{\sigma}, \eta) \vee \Pi_{\pure}(\bm{\rho}, \eta)) - K_{\pure}(\bm{\sigma}, \eta)\Bigr) - d(\bm{\sigma}, \eta) + d(\bm{\sigma}, \bm{\rho}),
  \end{equation*}
  with equality if and only if
  $\bm{\rho}\eta^{-1} \in \planar(\bm{\sigma}\eta^{-1})$.
\end{lemma}
\begin{proof}[Proof of Theorem \ref{thm:finite-N-free-cumulants}]
  Let us start by using \eqref{eq:HOFC-def-TT}, which yields
  \begin{equation*}
    \Kp_{\bm{\sigma}}[T, \bar{T}]
    = \sum_{\bm{\rho} \in \Sym_{n}^{D}}\sum_{\substack{\Pi \in \Partition_{\pure}(n)\\\Pi_{\pure}(\bm{\rho}) \leq \Pi}} \Phip_{\Pi, \bm{\rho}}[T, \bar{T}] \WeingCp{N}[\Pi \vee \Pi_{\pure}(\bm{\sigma}), \bm{\sigma}\bm{\rho}^{-1}].
  \end{equation*}
  Note that by hypothesis,
  \begin{equation*}
    \Phip_{\Pi, \bm{\rho}}[T, \bar{T}] = N^{n - \delta(\Pi, \bm{\rho})}\Bigl(\vphip_{\Pi, \bm{\rho}}[t, \bar{t}] + o(1)\Bigr),
  \end{equation*}
  and by \eqref{eq:connected-Weingarten}, since $\Pi(\bm{\rho}\bm{\sigma}^{-1}) \leq \Bigl(\Pi_{\pure}(\bm{\rho}, \eta) \vee \Pi_{\pure}(\bm{\sigma})\Bigr)_{[n]}= \Bigl(\Pi \vee \Pi_{\pure}(\bm{\sigma})\Bigr)_{[n]}$, we get
  \begin{equation*}
    \WeingCp{N}(\Pi \vee \Pi_{\pure}(\bm{\sigma}), \bm{\sigma}\bm{\rho}^{-1})
    = N^{2(1 - \#(\Pi \vee \Pi_{\pure}(\bm{\sigma}))) - d(\bm{\rho}, \bm{\sigma}) - nD}\Bigl(\Gamma[\Pi \vee \Pi_{\pure}(\bm{\sigma}), \bm{\sigma}\bm{\rho}^{-1}] +  o(1)\Bigr).
  \end{equation*}

  Let $\eta \in \scalg(\Pi, \bm{\rho})$. We have
  $\Pi = \Pi_{\pure}(\bm{\rho}, \eta)$ and $\delta(\Pi, \bm{\rho}) = d(\bm{\rho}, \eta)$, and the product
  $\Phip_{\Pi, \bm{\rho}}[T, \bar{T}] \WeingCp{N}[\Pi \vee \Pi_{\pure}(\bm{\sigma}), \bm{\sigma}\bm{\rho}^{-1}]$
  is of order
  \begin{equation*}
    N^{n + 2(1 - K_{\pure}(\bm{\sigma}, \bm{\rho}, \eta))) - d(\bm{\rho}, \eta) - d(\bm{\rho}, \bm{\sigma}) - nD}.
  \end{equation*}
  By Lemma \ref{lem:sum-exp-gauss}, we get
  \begin{equation*}
    n(1 - D) + 2(1 - K_{\pure}(\bm{\sigma}, \bm{\rho}, \eta)) - d(\bm{\rho}, \eta) - d(\bm{\rho}, \bm{\sigma}) \leq n(1 - D) + 2(1 - K_{\pure}(\bm{\sigma}, \eta)) - d(\bm{\sigma}, \eta),
  \end{equation*}
  with equality if and only if
  $(\bm{\rho}, \Pi_{\pure}(\bm{\rho}, \eta))\in \FSymg(\bm{\sigma}, \eta)$. Finally,
  the quantity
  $2(1 - K_{\pure}(\bm{\sigma}, \eta)) - d(\bm{\sigma}, \eta)$ is
  maximal when $\eta \in \scalp(\bm{\sigma})$. Together with the Definition
  \ref{def:forest-permutation-pure} of the Gaussian forests of
  permutation, this implies the result.

  Let us show the inverse relation. By Lemma \ref{lem:inverse-Kp}, we have
  \begin{equation*}
    \Phip_{ \bsig} [T, \bar T]  = \sum_{\btau \in \Sym_{n}^{D}}N^{nD - d(\bsig, \btau)}   \sum_{\substack{\Pi \in \PPart(n)\\\Pi_{\pure}(\bm{\tau}) \leq \Pi\\{\Pi_{\pure}(\bsig)\vee\Pi = 1_{n, \bar n}}}}  \  \Kp_{\Pi,\btau}[T, \bar T]  \;.
  \end{equation*}
  By using the scaling assumption \eqref{eq:scaling-gaussian} on the
  left-hand side and the first claim on the right-hand side, we have
  \begin{equation*}
    N^{n - \delta(\bm{\sigma})}\Bigl( \vphip_{\bm{\sigma}}(t, \bar{t}) + o(1) \Bigr)
    = \sum_{\bm{\rho} \in \Sym_{n}^{D}} \ \sum_{\substack{\Pi \in \PPart(n)\\\Pi_{\pure}(\bm{\rho}) \leq \Pi\\{\Pi_{\pure}(\bsig)\vee\Pi = 1_{n, \bar n}}}}  \ N^{- d(\bm{\sigma}, \bm{\rho}) + n + 2 - \exeps(\Pi, \bm{\rho})} \Bigl(\fcump_{\Pi, \bm{\rho}}(t, \bar{t}) + o(1)\Bigr) \;.
  \end{equation*}
  Lemma \ref{lem:delta-eps-decreasing} then implies that the exponent of $N$ satisfies
  \begin{equation*}
    - d(\bm{\sigma}, \bm{\rho}) + n + 2 - \exeps(\Pi, \bm{\rho})
    \leq n - \delta(\Pi, \bm{\rho}) - d(\bm{\sigma}, \bm{\rho}),
  \end{equation*}
  with equality if and only if $\bm{\rho}$ is decreasing from $\Pi$.
  Then, for every $\eta \in \scalg(\Pi, \bm{\rho})$, we have by Lemma \ref{lem:sum-exp-gauss} and $\Pi_{\pure}(\bm{\sigma}) = 1_{n, \bar{n}}$
  \begin{equation*}
    - d(\bm{\sigma}, \bm{\rho}) + n + 2 - \exeps(\Pi, \bm{\rho})
    \leq n + 2(1 - K_{\pure}(\bm{\sigma}, \eta)) - d(\bm{\sigma}, \eta)
    \leq n - \delta(\bm{\sigma}),
  \end{equation*}
  and equality if and only if $\bm{\rho}$ is decreasing from $\Pi$,
  $\bm{\rho}\eta^{-1} \in \planar(\bm{\sigma}\eta^{-1})$, and
  $\eta \in \scalg(\bm{\sigma})$.
\end{proof}
\begin{proof}[Proof of Lemma \ref{lem:sum-exp-gauss}]
  Let $\eta \in \scalg(\Pi, \bm{\rho})$. We start by noticing
  that
  \begin{equation*}
    \begin{split}
      -d(\bm{\rho}, \eta) - d(\bm{\rho}, \bm{\sigma})
      &= \sum_{c=1}^{D}\Bigl( \# (\rho_{c}\eta^{-1}) + \# (\sigma_{c}\eta^{-1}) - nD + \# (\rho_{c}\sigma_{c}^{-1})\Bigr) - nD - \sum_{c = 1}^{D}\# (\sigma_{c}\eta^{-1})\\
      &= 2\sum_{c=1}^{D}\Bigl( K(\rho_{c}\eta^{-1}, \sigma_{c}\eta^{-1}) - g(\rho_{c}\eta^{-1}, \sigma_{c}\eta^{-1})\Bigr) - nD - \sum_{c = 1}^{D}\# (\sigma_{c}\eta^{-1}).
    \end{split}
  \end{equation*}
  Then, we use the notation $L_{D}$ introduced in \eqref{eq:def-of-LD} and write
  \begin{equation*}
    \begin{split}
      \sum_{c = 1}^{D}\#(\sigma_{c}\eta^{-1})-2L_{D}\Bigl(\{\Pi(\rho_{c}\eta^{-1}, \sigma_{c}\eta^{-1})\}, &\Pi(\bm{\sigma}\eta^{-1}); \{\Pi(\sigma_{c}\eta^{-1})\}\Bigr) - 2K(\bm{\sigma}\eta^{-1}) +2K(\bm{\rho}\eta^{-1}, \bm{\sigma}\eta^{-1})\\
      &= 2\sum_{c=1}^{D} K(\rho_{c}\eta^{-1}, \sigma_{c}\eta^{-1})-\sum_{c = 1}^{D}\#(\sigma_{c}\eta^{-1}).
    \end{split}
  \end{equation*}

  Putting everything together, we end up with
  \begin{equation}\label{eq:L-int-gauss}
    \begin{split}
      - d(\bm{\rho}, \eta)
      &= 2 \Bigl(\# (\Pi_{\pure}(\bm{\sigma}, \eta) \vee \Pi_{\pure}(\bm{\rho}, \eta)) - K_{\pure}(\bm{\sigma}, \eta)\Bigr)- d(\bm{\sigma}, \eta) + d(\bm{\rho}, \bm{\sigma})\\
      &\quad\quad-2\sum_{c = 1}^{D}g(\rho_{c}\eta^{-1}, \sigma_{c}\eta^{-1}) -2L\Bigl(\Pi_{\pure}(\bm{\rho}, \bm{\sigma}, \eta), \Pi; \Pi_{\pure}(\bm{\rho}, \eta)\Bigr).
    \end{split}
  \end{equation}
  The last line in \eqref{eq:L-int-gauss} is a sum of non-positive
  terms. These terms are zero if and only if
  $\bm{\rho}\eta^{-1} \in \planar(\bm{\sigma}\eta^{-1})$.
\end{proof}

\subsection{Beyond the Gaussian scaling}
\label{sec:scaling-pure-tensors}

For random matrices or for the random tensors considered in Sec.~\ref{sec:tensors-matrix-product-scaling}, the dominant scales in $N$ of the $\Phi_\sigma^\mathrm{m}$ and $\mathcal{K}_\sigma^\mathrm{m}$ differ by a factor $N^{-nD}$ (see Thm.~\ref{thm:asymptotics-mixed-tensorial}). This, together with   Theorem \ref{thm:finite-N-free-cumulants}, is a motivation for studying pure random tensors for which the $N^{-nD}\Phip_\bsig$ scale as $\mathcal{K}_\bsig$. This is not a simple curiosity: we will show in Sec.~\ref{sec:gauss-cov} that  Gaussians whose covariances are tensor products of independent random matrices from classical random matrix ensembles (Wishart with parameters of the same order, GUE...) satisfy these scaling assumptions.

In this section, we consider pure random tensors such that for all
$n \in \N^{*}$ and $\bm{\sigma} \in \Sym_{n}^{D}$
\begin{equation}\label{eq:scaling-pure}
  \Phip_{\bm{\sigma}}[T, \bar{T}] = N^{2 + n - \exeps(\bm{\sigma})} \bigl( \vphip_{\bm{\sigma}}(t, \bar{t}) + o(1) \bigr),
\end{equation}
where $\vphip_{\bm{\sigma}}(t, \bar{t})$ does not depend on $N$ and
$\exeps$ was defined in \eqref{eq:exp-scaling-pure} by
\begin{equation*}
  \exeps(\bm{\sigma}) = \min_{\eta \in \Sym_{n}}\Bigl(2K_{\pure}(\bm{\sigma}, \eta) + d(\bm{\sigma}, \eta)\Bigr)\;.
\end{equation*}

\subsubsection{Comparison with the Gaussian scaling}
\label{subsub:comparison-scalings}

Similarly to Thm.~\ref{thm:degree}, one has the following.
\begin{prop}
\label{prop:melo-for-epsilon}
For any $n\in \N^{*}$ and $\bsig\in \Sym_n^D$, one has
\begin{equation}
   2 + n - \epsilon(\bsig) \le 1 -   (K_\pure(\bsig) - 1)\;,
\end{equation}
with equality if and only if $\bsig$ is melonic. Furthermore, if $\bsig$ is melonic, its canonical pairing is the unique element in $\scalp(\bm{\sigma})$.
\end{prop}
\begin{proof}
Combining Thm.~\ref{thm:degree} with  \eqref{eq:multiplicative-scaling-gaussien}, one sees that
\[
   \delta(\Pi, \bsig) \ge n-\#\Pi + (D-1) (K_\pure(\bsig) - \#\Pi) =  n + (D-1) K_\pure(\bsig)    - D\#\Pi \;,
\]
with equality if and only if $\bsig$ is melonic. As a consequence,
\[  2\#\Pi +  \delta(\Pi, \bsig) \ge n  + K_\pure(\bsig) + (D-2) (K_\pure(\bsig) - \#\Pi ) \ge n  + K_\pure(\bsig) \;,
\]
with equality if and only $\bsig$ is melonic and $\Pi=\Pi_\pure(\bsig)$. Therefore,
\[ \epsilon(\bsig) \ge n  + K_\pure(\bsig) \;,
\]
with equality if and only if $\bsig$ is melonic. Furthermore, if $\bsig$ is melonic,  $\scalp(\bm{\sigma}) $ contains the unique $\eta\in \Sym_n$ such that  $K_\pure(\bsig, \eta) = K_\pure(\bsig)$ and
$d(\bsig, \eta)=\delta(\Pi_\pure(\bsig), \bsig)$ (Thm.~\ref{thm:degree}) .
\end{proof}

\

We say that $\bsig\in \Sym_n^D$ factorizes at large $N$ for the pure random tensor $(T, \bar T)$ if
\begin{equation}
\label{eq:non-facto}
\E\bigl[\Tr_\bsig(T, \bar T)\bigr]\  \underset{N\rightarrow \infty}{\sim}\  \prod_{i=1}^q \E\bigl[\Tr_{\bsig_i}(T, \bar T)\bigr]
\end{equation}
where $\bsig_1, \ldots, \bsig_q$ are the pure connected components of $\bsig$. It has been shown recently that not all $\bsig$ factorize at large $N$ for the standard Gaussian \cite{Gurau:2025evo, Berthold:2026zxk}. We say on the other hand that $\bsig$ does not factorize at all at large $N$ if the dominant scale in $N$ of $\E\bigl[\Tr_\bsig(T, \bar T)\bigr]$ is the same as that of $\Phip_\bsig[T, \bar T]$. From \eqref{eq:classical-mom-cum}, this is equivalent to the fact that for every pure partition  $\Pi \in \mathcal{P_\pure}(n)$ satisfying $\Pi\ge \Pi_\pure(\bsig)$:
\begin{equation}
\label{eq:condition-non-facto}
 \delta(\bsig)   -   \delta(\Pi, \bm{\sigma}) \le 0\;.
\end{equation}

Since \eqref{eq:condition-non-facto} implies that $\bm{\sigma}$ is  decreasing (recall Definition \ref{def:decreasing-obs}), Lemma
\ref{lem:delta-eps-decreasing} gives the following.

\begin{lemma}
Consider $n\in \N^{*}$ and $\bsig\in \Sym_n^D$ that does not factorize at all at large $N$. Then $\delta(\bsig) =  \epsilon(\bsig) - 2\;.$
\end{lemma}

It is a priori not forbidden  for some $\bsig$ that factorize at large $N$ to satisfy $\delta(\bsig) = \epsilon(\bsig) - 2$, but we do not know of any example.

\

In order to compare this scaling behavior with the Gaussian scaling, we consider some pure random tensors $(T_\delta, \bar T_\delta)$ and $(T_\epsilon, \bar T_\epsilon)$ that respectively scale as \eqref{eq:scaling-gaussian} and \eqref{eq:scaling-pure}. For any purely connected $\bsig\in\Sym_n^D$,  both $\E\bigl[\Tr_\bsig(T_\delta, \bar T_\delta)\bigr]$ and $\E\bigl[\Tr_\bsig (T_\epsilon, \bar T_\epsilon)\bigr]$ scale in the same way. But more generally,  for the $\bsig\in\Sym_n^D$ with fixed number of pure connected components $K(\bsig)-1=q\ge 0$:
\begin{equation}
\label{eq:fluctuations-delta}
  \Phip_{\bm{\sigma}}[T_\delta, \bar{T_\delta}] = \left\{
    \begin{array}{ll}
   N^{1-(D-1) q }\bigl(\vphip(t_\delta, \bar t_\delta)  + o(1)  \bigr) & \textrm{ if } \bsig \textrm{ melonic,  }\\[+1ex]
   o\bigl(N^{1-(D-1) q }\bigr)&  \textrm{ otherwise. }
    \end{array}
\right.
\end{equation}
while on the other hand,
\begin{equation}
\label{eq:fluctuations-epsilon}
  \Phip_{\bm{\sigma}}[T_\epsilon \bar{T_\epsilon}] = \left\{
    \begin{array}{ll}
   N^{1- q }\bigl(\vphip(t_\epsilon, \bar t_\epsilon)  + o(1)  \bigr) & \textrm{ if } \bsig \textrm{ melonic,  }\\[+1ex]
   o\bigl(N^{1- q }\bigr)&  \textrm{ otherwise. }
    \end{array}
\right. \; .
\end{equation}

So while the moments associated to purely connected $\bsig$ have the same dominant magnitude in $N$ for both  \eqref{eq:scaling-gaussian} and \eqref{eq:scaling-pure}, some of the fluctuations corresponding to classical cumulants of a fixed number of purely connected trace-invariants are ``boosted '' for \eqref{eq:scaling-pure}\footnote{This can be shown to be the case for a much larger set of trace-invariants, but we do not include the proof here.}. This is not the case for the trace-invariants that do not factorize at all asymptotically.

\subsubsection{Tensorial free cumulants for the scaling $\exeps$}

As in the mixed and Gaussian cases, the appropriate combinatorial
object will be the forests of permutations.
\begin{definition}[Pure forests of permutations]\label{def:pure-forest-permutation}
  Let $n \in \N^{*}$, $\eta \in \Sym_{n}$, and $\bm{\sigma} \in \Sym_{n}^{D}$. A
  pure forest of permutations on $(\bm{\sigma}, \eta)$ is a pair
  $(\bm{\rho}, \Pi) \in \Sym_{n}^{D} \times \Partition_{\pure}(n)$ such that
  \begin{itemize}
    \item $\bm{\rho}\eta^{-1} \in \planar(\bm{\sigma}\eta^{-1})$;
    \item $\Pi_{\pure}(\bm{\rho}, \eta) \leq \Pi\quad$ and $\quad L\Bigl(\Pi_{\pure}(\bm{\sigma}, \eta) \vee \Pi_{\pure}(\bm{\rho}, \eta), \Pi; \Pi_{\pure}(\bm{\rho}, \eta)\Bigr) = 0$.
  \end{itemize}
  We denote by $\FSymp(\bm{\sigma}, \eta)$ the set of forests of permutations
  on $\bm{\sigma}$ in the pure case.

  The set of optimal forests of permutations on $\bm{\sigma}$ is then defined to be
  \begin{equation*}
    \FSymp(\bm{\sigma}) = \Bigl\{ (\bm{\rho}, \Pi) \in \Sym_{N}^{D} \times \Partition_{\pure}(n) \colon \exists \eta \in \scalp(\bm{\sigma}) \cap \scalp(\Pi, \bm{\rho}), (\bm{\rho}, \Pi) \in \FSymp(\bm{\sigma}, \eta)\Bigr\}.
  \end{equation*}
\end{definition}
\begin{remark}[$\FSymp(\bm{\sigma})$ is not empty]
  The set of pure forests of permutations is not empty. Given $\bm{\sigma} \in \Sym_n^D$, we notice that for all $\eta \in \scalp(\bm{\sigma})$, $\bm{\sigma}\eta^{-1} \in \planar(\bm{\sigma}\eta^{-1})$ so that we get
  \begin{equation*}
      \Bigl( \bm{\sigma}, \Pi_{\pure}(\bm{\sigma}, \eta)\Bigr) \in \FSymp(\bm{\sigma}, \eta).
  \end{equation*}
\end{remark}

The pure forests of permutations appear in the expression of the free
cumulants, as we now demonstrate.
\begin{theorem}\label{thm:free-cum-product-scaling}
  Let $(T, \bar{T})$ be a pure LU-invariant random tensor that satisfies the
  scaling assumption \eqref{eq:scaling-pure}. For all $n \in \N^{*}$
  and $\bm{\sigma} \in \Sym_{n}^{D}$, we have
  \begin{equation*}
    \Kp_{\bm{\sigma}}[T, \bar{T}] = N^{n(1 - D) + 2 - \exeps(\bm{\sigma})}\Bigl( \fcump_{\bm{\sigma}}(t, \bar{t}) + o(1) \Bigr)
  \end{equation*}
  where
  \begin{equation*}
    \fcump_{\bm{\sigma}}(t, \bar{t}) = \sum_{(\bm{\rho}, \Pi) \in \FSymp(\bm{\sigma})}\vphip_{\Pi, \bm{\rho}}(t, \bar{t}) \ \Gamma[\Pi \vee \Pi_{\pure}(\bm{\sigma}), \bm{\sigma}\bm{\rho}^{-1}].
  \end{equation*}

  Furthermore, we have the inverse relations
  \begin{equation*}
    \vphip_{\bm{\sigma}}(t, \bar{t}) = \sum_{\substack{(\bm{\rho}, \Pi) \in \FSymp(\bm{\sigma})\\\Pi \vee \Pi_{\pure}(\bm{\sigma}) = 1_{n, \bar{n}}}}\fcump_{\Pi, \bm{\rho}}(t, \bar{t}).
  \end{equation*}
\end{theorem}
The proof of Theorem \ref{thm:free-cum-product-scaling} uses a
description of the exponent $\exeps(\bm{\sigma})$ in terms of forests
of permutations on $\bm{\sigma}$, provided by Lemma \ref{lem:sum-exp-pure}.
\begin{lemma}\label{lem:sum-exp-pure}
  Let $n \in \N^{*}$, $\Pi \in \Partition_{\pure}(n)$, and
  $\bm{\rho}, \bm{\sigma} \in \Sym_{n}^{D}$ with $\Pi_{\pure}(\bm{\rho}) \leq \Pi$. We
  have for all $\eta \in \hat{\Omega}(\Pi, \bm{\rho})$
  \begin{equation*}
    2\Bigl(\# \Pi - K_{\pure}(\bm{\rho}, \eta)\Bigr) - d(\bm{\rho}, \eta)
    \leq 2\Bigl(\#(\Pi_{\pure}(\bm{\sigma}, \eta) \vee \Pi) - K_{\pure}(\bm{\sigma}, \eta)\Bigr) - d(\bm{\sigma}, \eta) + d(\bm{\sigma}, \bm{\rho}),
  \end{equation*}
  with equality if and only if $(\bm{\rho}, \Pi) \in \FSymp(\bm{\sigma}, \eta)$.
\end{lemma}

\begin{proof}[Proof of Theorem \ref{thm:free-cum-product-scaling}]
  We start by expressing the finite-$N$ precursors to the free
  cumulants in terms of the classical cumulants using
  \eqref{eq:HOFC-def-TT}:
  \begin{equation*}
    \Kp_{\bm{\sigma}}[T,\bar T]  = \sum_{\btau \in \Sym_{n}^{D}} \sum_{\substack{{\Pi \in \PPart(n)}\\\Pi_{\pure}(\bm{\tau}) \leq \Pi}} \Phip_{\Pi, \btau}[T, \bar T]\
    \WeingCp{N}\bigl[\Pi\vee\Pi_{\pure}(\bsig) , \bsig\btau^{-1}\bigr] \;,
  \end{equation*}
  Using our scaling assumption \eqref{eq:scaling-pure} and the
  asymptotic expression \eqref{eq:connected-Weingarten} of $\WeingCp{N}$ we get
  \begin{equation}
  \label{eq:proof-pure-epsilon}
    \begin{split}
      \Kp_{\bm{\sigma}}[T,\bar T]
      = \sum_{\btau \in \Sym_{n}^{D}} &\sum_{\substack{{\Pi \in \PPart(n)}\\\Pi_{\pure}(\bm{\tau}) \leq \Pi}} N^{n(1 - D) + 2\# \Pi - \exeps(\Pi, \bm{\tau}) - d(\bm{\sigma}, \bm{\tau}) + 2(1 - \# \Pi \vee \Pi_{\pure}(\bsig))}\\
      &\qquad\qquad\qquad\times\Bigl( \vphip_{\Pi, \bm{\tau}}(t, \bar{t}) \ \Gamma[\Pi \vee \Pi_{\pure}(\bm{\sigma}), \bm{\sigma}\bm{\tau}^{-1}] + o(1) \Bigr)\;.
    \end{split}
  \end{equation}

  Let $\eta \in \scalp(\Pi, \bm{\tau})$. By Lemma \ref{lem:sum-exp-pure}, the  exponent of $N$ (minus $n(1 - D)$) satisfies
  \begin{equation*}
     2\Bigl(\# \Pi - K_{\pure}(\bm{\tau}, \eta)\Bigr) - d(\bm{\tau}, \eta) - d(\bm{\sigma}, \bm{\tau}) + 2\Bigl(1 - \# \Pi \vee \Pi_{\pure}(\bsig)\Bigr)
    \leq 2(1 - K_{\pure}(\bm{\sigma}, \eta)) - d(\bm{\sigma}, \eta),
  \end{equation*}
  with equality if and only if $(\bm{\tau}, \Pi) \in \FSymp(\bm{\sigma}, \eta)$.
  The exponent of $N$  takes its maximal value
  $n(1 - D) + 2 - \exeps(\bm{\sigma})$ when $\eta \in \scalp(\bm{\sigma})$. This
  gives the first claim.

  To prove the second claim, we use
  \eqref{eq:Phi-in-terms-of-precursors-pure}:
  \begin{equation*}
    \Phip_{\bsig} [T, \bar T]  = \sum_{\btau \in \Sym_{n}^{D}}\sum_{\substack{\Pi \in \PPart(n)\\\Pi_{\pure}(\bm{\tau}) \leq \Pi\\{\Pi_\mathrm{p}(\bsig)\vee\Pi = 1_{n, \bar n}}}}  \  N^{nD - d(\bsig, \btau)}   \Kp_{\Pi,\btau}[T, \bar T]  \;.
  \end{equation*}
  The scaling assumption \eqref{eq:scaling-pure} and the first claim give
  \begin{equation*}
    N^{n + 2 - \exeps(\bm{\sigma})}\Bigl( \vphip_{\bm{\sigma}}(t, \bar{t}) + o(1)\Bigr)
    = \sum_{\btau \in \Sym_{n}^{D}}\sum_{\substack{\Pi \in \PPart(n)\\\Pi_{\pure}(\bm{\tau}) \leq \Pi\\{\Pi_\mathrm{p}(\bsig)\vee\Pi = 1_{n, \bar n}}}}  \  N^{n + 2\# \Pi - \exeps(\Pi, \bm{\tau}) - d(\bsig, \btau)} \Bigl( \fcump_{\Pi, \bm{\tau}}(t, \bar{t}) + o(1) \Bigr)  \;.
  \end{equation*}

  Let $\eta \in \scalp(\Pi, \bm{\tau})$. Using Lemma \ref{lem:sum-exp-pure},
  the exponent of $N$ satisfies
  \begin{equation*}
    n + 2\# \Pi - \exeps(\Pi, \bm{\tau}) - d(\bsig, \btau) \leq n + 2 - \exeps(\bm{\sigma}),
  \end{equation*}
  with equality if and only if $(\Pi, \bm{\tau}) \in \FSymp(\bm{\sigma}, \eta)$ and
  $\eta \in \scalp(\bm{\sigma})$. This gives the second claim.
\end{proof}
\begin{proof}[Proof of Lemma \ref{lem:sum-exp-pure}]
  Let $\eta \in \hat{\Omega}(\Pi, \bm{\rho})$. We start by noticing that
  \begin{equation*}
    \begin{split}
      -d(\bm{\rho}, \eta) - d(\bm{\rho}, \bm{\sigma})
      &= \sum_{c=1}^{D}\Bigl( \# (\rho_{c}\eta^{-1}) + \# (\sigma_{c}\eta^{-1}) - nD + \# (\rho_{c}\sigma_{c}^{-1})\Bigr) - nD - \sum_{c = 1}^{D}\# (\sigma_{c}\eta^{-1})\\
      &= 2\sum_{c=1}^{D}\Bigl( K(\rho_{c}\eta^{-1}, \sigma_{c}\eta^{-1}) - g(\rho_{c}\eta^{-1}, \sigma_{c}\eta^{-1})\Bigr) - nD - \sum_{c = 1}^{D}\# (\sigma_{c}\eta^{-1}).
    \end{split}
  \end{equation*}
  Then, we use the notation $L_{D}$ introduced in \eqref{eq:def-of-LD} and write
  \begin{equation*}
    \begin{split}
      \sum_{c = 1}^{D}\#(\sigma_{c}\eta^{-1})-2L_{D}\Bigl(\{\Pi(\rho_{c}\eta^{-1}, \sigma_{c}\eta^{-1})\}, &\Pi(\bm{\sigma}\eta^{-1}); \{\Pi(\sigma_{c}\eta^{-1})\}\Bigr) - 2K(\bm{\sigma}\eta^{-1}) +2K(\bm{\rho}\eta^{-1}, \bm{\sigma}\eta^{-1})\\
      &= 2\sum_{c=1}^{D} K(\rho_{c}\eta^{-1}, \sigma_{c}\eta^{-1})-\sum_{c = 1}^{D}\#(\sigma_{c}\eta^{-1}).
    \end{split}
  \end{equation*}
  Finally, we use the notation $L$ (see \eqref{eq:def-of-L}) and get
  \begin{equation*}
    -2L\Bigl(\Pi_{\pure}(\bm{\rho}, \eta) \vee \Pi_{\pure}(\bm{\sigma}, \eta), \Pi; \Pi_{\pure}(\bm{\rho}, \eta)\Bigr) + 2 \# \Bigl(\Pi_{\pure}(\bm{\sigma}, \eta) \vee \Pi\Bigr)  = 2\Bigl(\# \Pi - K_{\pure}(\bm{\rho}, \eta)\Bigr) + 2K(\bm{\rho}\eta^{-1}, \bm{\sigma}\eta^{-1}).
  \end{equation*}

  Putting everything together, we end up with
  \begin{equation}\label{eq:L-int-pure}
    \begin{split}
      2\Bigl(\# \Pi - K_{\pure}(\bm{\rho}, \eta)\Bigr) - d(\bm{\rho}, \eta)
      &= 2 \Bigl(\# (\Pi_{\pure}(\bm{\sigma}, \eta) \vee \Pi) - K_{\pure}(\bm{\sigma}, \eta)\Bigr)- d(\bm{\sigma}, \eta) + d(\bm{\rho}, \bm{\sigma})\\
      &\quad\quad-2\sum_{c = 1}^{D}g(\rho_{c}\eta^{-1}, \sigma_{c}\eta^{-1}) -2L\Bigl(\Pi_{\pure}(\bm{\rho}, \bm{\sigma}, \eta), \Pi; \Pi_{\pure}(\bm{\rho}, \eta)\Bigr)\\
      &\quad\quad-2L_{D}\Bigl(\{\Pi(\rho_{c}\eta^{-1}, \sigma_{c}\eta^{-1})\}, \Pi(\bm{\sigma}\eta^{-1}); \{\Pi(\sigma_{c}\eta^{-1})\}\Bigr).
    \end{split}
  \end{equation}
  The last two lines in \eqref{eq:L-int-pure} are a sum of
  non-positive terms. These terms are zero if and only if
  $(\bm{\rho}, \Pi) \in \FSymp(\bm{\sigma}, \eta)$.
\end{proof}

\subsection{Melonic free cumulants and the tensor BGW integral}

\begin{theorem}
\label{thm:R-transform-matrix-scaling-PURE}
Consider $n\in \mathbb{N}^{*}$ and a  deterministic pure tensor $J$ with $D\ge 1 $ inputs such that for every $\bsig\in \Sym_{n}^{D}$, $\lim_{N\rightarrow \infty} \Tr_\bsig(J, \bar J) = t_\bsig <\infty$.
If $(T, \bar T)$ is a pure LU-invariant random tensor with $D$ inputs and that scales as the pure standard complex Gaussian on $D$ inputs, i.e.\ satisfies \eqref{eq:scaling-gaussian}, or that scales as \eqref{eq:scaling-pure}:
\begin{equation}
\lim_{N\rightarrow \infty }   \frac 1 {N} \frac{\partial^{n}}{\partial z^n }\frac{\partial^{n}}{\partial {\bar z}^n } \E\biggl[\log \int [dU]\  \ee^{N^{D/2} ( z JU\cdot T + \bar z\bar J U^{\dagger}\cdot \bar T)} \biggr] \Biggr\lvert_{z=\bar z = 0} =
 (n!)^{2}\sum_{\bsig\in  \mathbb{M}_n^D / \sim_{\pure}} \frac{t_\bsig}{\#\Aut_{\pure}(\bm{\bsig})}\; \fcump_\bsig(t, \bar t)\;,
\end{equation}
where the sum is over elements $\bsig\in \Sym_{n}^{D}$ that are first order,
i.e.\ purely connected and melonic (which is equivalent to connected
and such that $\Omega(\bsig)=0$). If $J$ has only one non-vanishing
element $J_{1,\ldots, 1} = 1$, then for every $\bsig\in \Sym_{n}^{D}$,
$t_\bsig=1$.
\end{theorem}
\begin{proof}
  We start by differentiating the logarithm of the BGW integral
  \eqref{eq:def-BGW}:
  \begin{equation*}
    \frac{\partial^n}{\partial z^n} \frac{\partial^{n}}{\partial {\bar z}^n }\E\biggl[\log \int [\dd U]\  \ee^{N^{D/2} ( z J\cdot T + \bar z\bar J \cdot \bar T)} \biggr] \Biggr\lvert_{z=0}
    = N^{nD}\sum_{\Pi \in \PPart(n)}\mu_{\Pi}\prod_{S \in \Pi} \int [\dd U] (J \cdot T)^{\# S} (\bar{J} \cdot \bar{T})^{\# S}.
  \end{equation*}
  Using the computation of the BGW moments \eqref{eq:BGW-mom-G}, we have
  \begin{equation*}
    \begin{split}
      \frac{\partial^n}{\partial z^n} \frac{\partial^{n}}{\partial {\bar z}^n }\E\biggl[\log \int [\dd U]\  \ee^{N^{D/2}  ( z JU\cdot T + \bar z\bar J U^{\dagger} \cdot \bar T)} \biggr] \Biggr\lvert_{z=0}
      &= N^{nD}\sum_{\bm{\tau} \in \Sym_{n}^{D}}\sum_{\substack{\Pi \in \PPart(n)\\\Pi_{\pure}(\bm{\tau}) \leq \Pi}}\mu_{\Pi} \Tr_{\bm{\tau}}(J, \bar{J}) \ \E\Bigl[\G_{\Pi, \bm{\tau}}[T, \bar{T}]\Bigr],\\
      &= N^{nD}\sum_{\bm{\tau} \in \Sym_{n}^{D}} \Tr_{\bm{\tau}}(J, \bar{J}) \ \E\Bigl[\mathcal{L}^{\pure}_{\bm{\tau}}[T, \bar{T}]\Bigr].
    \end{split}
  \end{equation*}
  Using the right-hand side of \eqref{eq:cumulants-of-HCIZ-pure}, we get
  \begin{equation*}
    \begin{split}
    \E\Bigl[\mathcal{L}^{\pure}_{\bm{\tau}}[T, \bar{T}]\Bigr]
    &= \sum_{\bm{\rho} \in \Sym_{n}^{D}}\E\Bigl[ \Tr_{\bm{\rho}}(T, \bar{T})\Bigr]\WeingCp{N}[\Pi_{\pure}(\bm{\rho}, \bm{\tau}), \bm{\tau}\bm{\rho}^{-1}]\\
    &= \sum_{\bm{\rho} \in \Sym_{n}^{D}}\sum_{\substack{\Pi \in \PPart(n)\\\Pi_{\pure}(\bm{\rho}) \leq \Pi}}\Phip_{\Pi, \bm{\rho}}[T, \bar{T}]\WeingCp{N}[\Pi_{\pure}(\bm{\rho}, \bm{\tau}), \bm{\tau}\bm{\rho}^{-1}].
    \end{split}
  \end{equation*}

  If $(T, \bar{T})$ scales as a Gaussian, i.e.\ satisfies \eqref{eq:scaling-gaussian}, we get
  \begin{equation*}
    \E\Bigl[\mathcal{L}^{\pure}_{\bm{\tau}}[T, \bar{T}]\Bigr]
    = \sum_{\bm{\rho} \in \Sym_{n}^{D}}\sum_{\substack{\Pi \in \PPart(n)\\\Pi_{\pure}(\bm{\rho}) \leq \Pi}}N^{n - \delta(\Pi, \bm{\rho}) - d(\bm{\tau}, \bm{\rho}) + 2(1 - K_{\pure}(\bm{\tau}, \bm{\rho})) -nD}\Bigl(\vphip_{\Pi, \bm{\rho}}(t, \bar{t})\Gamma[\Pi_{\pure}(\bm{\rho}, \bm{\tau}), \bm{\tau}\bm{\rho}^{-1}] + o(1)\Bigr).
  \end{equation*}
  Hence, we get to use Lemma \ref{lem:sum-exp-gauss}, and get that for all $\eta \in \scalg(\Pi, \bm{\rho})$ the
  exponent of $N$ is
  \begin{equation*}
    n - \delta(\Pi, \bm{\rho}) - d(\bm{\tau}, \bm{\rho}) + 2(1 - K_{\pure}(\bm{\tau}, \bm{\rho})) -nD
    \leq n + 2(1 - K_{\pure}(\bm{\tau}, \bm{\rho}) + K_{\pure}(\bm{\rho}, \bm{\tau}, \eta) - K_{\pure}(\bm{\tau}, \eta)) -d(\bm{\tau}, \eta) - nD,
  \end{equation*}
  with equality if and only if $\bm{\rho}\eta^{-1} \in \planar(\bm{\tau}\eta^{-1})$. We then get that
  \begin{equation*}
    n - \delta(\Pi, \bm{\rho}) - d(\bm{\tau}, \bm{\rho}) + 2(1 - K_{\pure}(\bm{\tau}, \bm{\rho})) -nD
    \leq n +2 - \exeps(\bm{\tau}) - nD,
  \end{equation*}
  with equality if and only if
  $\bm{\rho}\eta^{-1} \in \planar(\bm{\tau}\eta^{-1})$,
  $\eta \in \scalp(\bm{\tau})$, and $\Pi \preceq \Pi_{\pure}(\bm{\rho},\bm{\tau})$.
  Finally, by Prop.~\ref{prop:melo-for-epsilon}, we have
  \begin{equation*}
    n +2 - \exeps(\bm{\tau}) \leq  1,
  \end{equation*}
  with equality if and only if $\bm{\tau}$ is purely connected and
  melonic. We notice that if $\bm{\tau}$ is purely connected, the
  condition $\Pi \preceq \Pi_{\pure}(\bm{\rho},\bm{\tau})$ is
  satisfied. Putting everything together and using our assumption on $(J, \bar{J})$, we get that
  \begin{equation*}
    \begin{split}
      \frac{\partial^n}{\partial z^n} \frac{\partial^{n}}{\partial {\bar z}^n }&\E\biggl[\log \int [\dd U]\  \ee^{N^{D/2}  ( z JU\cdot T + \bar z\bar J U^{\dagger} \cdot \bar T)} \biggr] \Biggr\lvert_{z=0}\\
      &= N \Biggl(\sum_{\bm{\tau} \in \mathbb{M}_{n}^{D}}\sum_{(\bm{\rho}, \Pi) \in \FSymg(\bm{\tau})} t_{\bm{\tau}} \vphip_{\Pi, \bm{\rho}}(t, \bar{t}) \ \Mnc(\bm{\tau}\bm{\rho}^{-1}) + o(1) \Biggr)
      = N \Biggl(\sum_{\bm{\tau} \in \mathbb{M}_{n}^{D}} t_{\bm{\tau}} \ \fcump_{\bm{\tau}}(t, \bar{t}) + o(1) \Biggr).
    \end{split}
  \end{equation*}
  where $\mathbb{M}_{n}^{D}$ is the set of purely connected melonic
  $D$-tuples of permutations. We then get the result by introducing
  the group of automorphism $\Aut_\pure(\bsig)$, see
  \eqref{eq:def-automorphism}.

  In the case where $(T, \bar{T})$ satisfies the second pure scaling
  assumption \eqref{eq:scaling-pure}, we use Lemma
  \ref{lem:sum-exp-pure} and proceed as before. We obtain the same result as for the Gaussian scaling:
    \begin{equation*}
    \begin{split}
      \frac{\partial^n}{\partial z^n} \frac{\partial^{n}}{\partial {\bar z}^n }&\E\biggl[\log \int [\dd U]\  \ee^{N^{D/2}  ( z JU\cdot T + \bar z\bar J U^{\dagger} \cdot \bar T)} \biggr] \Biggr\lvert_{z=0}\\
      &= N \Biggl(\sum_{\bm{\tau} \in \mathbb{M}_{n}^{D}}\sum_{(\bm{\rho}, \Pi) \in \FSymp(\bm{\tau})} t_{\bm{\tau}} \vphip_{\Pi, \bm{\rho}}(t, \bar{t}) \ \Mnc(\bm{\tau}\bm{\rho}^{-1}) + o(1) \Biggr)
      = N \Biggl(\sum_{\bm{\tau} \in \mathbb{M}_{n}^{D}} t_{\bm{\tau}} \ \fcump_{\bm{\tau}}(t, \bar{t}) + o(1) \Biggr).
    \end{split}
  \end{equation*}
  The proof is concluded similarly by introducing the $1/\# \Aut_{\pure}(\bm{\tau})$ factor.
\end{proof}

\section{Gaussians with non-trivial covariances, random or not}
\label{sec:gauss-cov}

\subsection{The random matrix case: Ginibre and Wishart ensembles}
\label{sec:random-matrix-case-GUE-Wishart}

We discuss the definition of the Wishart ensemble and the Ginibre
Ensemble with a possibly random covariance matrix.

Let $N, M \geq 1$ be two integers. Let $C$ be a random positive
definite matrix with law $\nu$. We denote the set of positive-definite
matrices of size $M \times M$ by $\Positive(M)$. We consider the
probability distribution $\mu_{N, M}$ on $\Mat_{N \times M}(\C)$
defined by
\begin{equation}\label{eq:def-wishart}
  \dd \mu_{N, M}(X) = \frac{1}{Z_{N, M, \nu}}\int_{\Positive(M)} \exp\Bigl(- N \Tr X C^{-1} X^{\dagger}\Bigr) \dd \nu(C) \dd X,
\end{equation}
where $Z_{N, M, \nu}$ is the partition function
\begin{equation*}
  Z_{N, M, \nu} = \int_{\Mat_{N \times M}(\C)}\int_{\Positive(M)} \exp\Bigl(- N \Tr X C^{-1} X^{\dagger}\Bigr) \dd \nu(C) \dd X.
\end{equation*}
Interpreting the trace $\Tr X C^{-1} X^{\dagger}$ as the Hilbert-Schmidt scalar product defined by
\begin{equation*}
  \langle A, B \rangle_{\HS} \coloneq \Tr AB^{\dagger}\quad \text{ for } A, B \in \Mat_{M\times M}(\C),
\end{equation*}
we have $\Tr X C^{-1} X^{\dagger} = \langle C^{-1}, X^{\dagger}X\rangle_{\HS}$. Hence, we
may see the integral appearing in the definition of $\mu$ as the Laplace
transform of $\nu$ evaluated in $N X^{\dagger}X$. It motivates us to
make the following assumption.
\begin{hypothesis}\label{hyp:laplace-B}
  The Laplace transform
  \begin{equation*}
    \Lambda_{C}(A) = \int_{\Positive(M)}\ee^{-\langle C^{-1}, A\rangle_{\HS}}\dd \nu(C)
  \end{equation*}
  exists for all positive matrix $A$ and decays faster than any polynomial in the entries of $A$ as $\max_{i,j}|A_{ij}| \to \infty$.
\end{hypothesis}
Hypothesis \ref{hyp:laplace-B} ensures that the measure $\mu_{N, M}$
given in \eqref{eq:def-wishart} is well-defined and that all its
moments are finite.

\begin{ex}[Deterministic matrix $C$]
  If we choose $\nu = \delta_{BB^{\dagger}}$ where $B$ is any full-rank
  matrix, then we recover the Wishart distribution with a fixed
  covariance matrix. In particular, the choice $B = \Id_{M}$ gives the
  classical Wishart ensemble.
\end{ex}

A different point of view is provided by performing the change of
variable $X' = X \sqrt{C^{-1}}$. This gives the following result.
\begin{lemma}\label{lem:matrix-chvar}
  Assume that $\nu$ satisfies Hypothesis \ref{hyp:laplace-B}.
  Denote by $\mu_{\Gau}$ the law of a matrix whose entries are
  independent complex centered Gaussian variables of variance $1/N$, and
  by $\tilde{\nu}$ the measure given by
  \begin{equation*}
    \dd\tilde{\nu}(C) = \frac{(\det C)^{-1/2}}{\int_{\Positive(M)}(\det C')^{-1/2}\dd \nu(C')}\dd \nu(C).
  \end{equation*}
  Then, for any measurable function
  $f \colon \Mat_{N \times M}(\C) \to \C$, we have
  \begin{equation*}
    \int_{\Mat_{N \times M}(\C)}f(X)\dd \mu_{N, M}(X)
    = \int_{\Positive(M)}\int_{\Mat_{N \times M}(\C)} f(G \sqrt{C}) \dd \mu_{\Gau}(G)\dd \tilde{\nu}(C).
  \end{equation*}
\end{lemma}
Lemma \ref{lem:matrix-chvar} implies that it is equivalent to
consider a random matrix of law $\mu_{N, M}$ --- a Gaussian matrix with random
covariance --- or the product of two independent matrices: a standard
Gaussian matrix $X$ and a positive-definite matrix $C^{1/2}$ with $C$
of law $\nu$.
\begin{proof}
  We perform the change of variables $X' = X \sqrt{C^{-1}}$ whose Jacobian
  is $(\det C)^{-1/2}$. We get that
  \begin{equation*}
    \int_{\Mat_{N \times M}(\C)}f(X)\dd \mu_{N, M}(X)
    = \frac{1}{Z_{N, \nu}}\int_{\Positive(M)}\int_{\Mat_{N \times M}(\C)} f(X \sqrt{C})\ee^{-N\Tr XX^{\dagger}}(\det C)^{-1/2}\dd X\dd \mu(C),
  \end{equation*}
  and
  \begin{equation*}
    Z_{N, \nu} = \int_{\Positive(M)}\int_{\Mat_{N \times M}(\C)} \ee^{-N\Tr XX^{\dagger}}(\det C)^{-1/2}\dd X\dd \nu(C) = \Bigl(\frac{2\pi}{N}\Bigr)^{MN} \int_{\Positive(M)} (\det C)^{-1/2}\dd \nu(C).
  \end{equation*}
  This yields the result as
  \begin{equation*}
    \dd \mu_{\Gau}(G) = \Bigl(\frac{2\pi}{N}\Bigr)^{-MN}\exp\Bigl(-N \Tr GG^{\dagger}\Bigr)\dd G. \qedhere
  \end{equation*}
\end{proof}

When computing moments of $XX^{\dagger}$, we have that
\begin{equation*}
  \E\Bigl[\Tr_{\sigma}(XX^{\dagger})\Bigr] = \E\Bigl[\Tr_{\sigma}(G C G^{\dagger})\Bigr] = \E\Bigl[\Tr_{\sigma}(CG^{\dagger}G)\Bigr],
\end{equation*}
and hence it is equivalent to compute the moment of the Wishart matrix
with random covariance $XX^{\dagger}$ or the moments of $CG^{\dagger}G$ -- where
$G^{\dagger}G$ is a usual Wishart matrix.

\subsection{The random tensor case}
\label{sec:random-tensor-case}

Fix an integer $D \geq 1$, the number of inputs (or colors), and dimensions
$N_{1}, \ldots, N_{D} \geq 1$. Let
$A = (A_{i_{1}, \ldots, i_{D}; j_{1}, \ldots, j_{D}})$ be a random tensor with $2D$ indices of law $\nu$.
As detailed in Sec.~\ref{sub:coarser1}, the tensor $A$ can be seen as a square matrix $A_{(1_D)}$ of dimension $\prod_{c=1}^{D}N_{c}$ and a tensor $T$ can be seen as a vector of length $\prod_{c=1}^{D}N_{c}$. In this section, we let
$N = \max_{1 \leq c \leq D} N_{c}$.

Motivated by the discussion of Section
\ref{sec:random-matrix-case-GUE-Wishart}, we make the following
assumptions.
\begin{hypothesis}\label{hyp:tensors-gaussian}
  \begin{enumerate}
    \item The random matrix $A_{(1_D)}$ obtained by grouping the indices is almost-surely a
          Hermitian positive-definite matrix;
          \item The laplace transform of $A^{-1}$,
          \begin{equation*}
            \Lambda_{A} \colon J \mapsto \int_{\Positive(\prod_{c}N_{c})} \ee^{- \left< A^{-1}_{(1_D)}, J \otimes \bar{J} \right>_{\HS}}\ \dd \nu(A)
          \end{equation*}
          is well-defined for all $J \in \C^{\prod_{c}N_{c}}$ and
          decays faster than any polynomial in the coefficients of $J$
          as $\|J\|_{\infty} = \max_{\bm{i}}|J_{\bm{i}}| \to \infty$.
  \end{enumerate}
\end{hypothesis}
These hypotheses allow us to define the law of a Gaussian tensor with random covariance $A$.
\begin{definition}[Gaussian tensor with random covariance $T_A$]\label{def:gauss-rd-cov}
  Let $A$ be a random tensor which satisfies Hypothesis \ref{hyp:tensors-gaussian}. We define the measure $\mu_{\bm{N}}$ by
  \begin{equation}\label{eq:def-tenseur-gaussien}
  \mu_{\bm{N}}(T_A) = \frac{1}{Z_{\bm{N}, \nu}} \int_{\Positive(N_{1} \times \cdots \times N_{d})}\ee^{-N^{D-1} T_A \cdot A^{-1} \cdot \bar{T}_A} \dd \nu(A)\dd T_A = \frac{1}{Z_{\bm{N}, \nu}}\Lambda_{A}\Bigl(N^{(D-1)/2}T_A\Bigr)\dd T_A,
\end{equation}
where the partition function $Z_{\bm{N}, \nu}$ is
\begin{equation*}
  Z_{\bm{N}, \nu} = \int_{\C^{N_{1} \times \cdots \times N_{d}}}\int_{\Positive(N_{1} \times \cdots \times N_{d})}\ee^{-N^{D-1} T_A \cdot A^{-1} \cdot \bar{T}_A} \dd \nu(A)\dd T_A.
\end{equation*}
We denote a random tensor with law $\mu_{\bm{N}}$ by $T_A$.
\end{definition}
Hypothesis \ref{hyp:tensors-gaussian} ensures that the
moments of $(T_A, \bar{T}_A)$ are finite.

\begin{remark}[Reduction to the tensor HCIZ integral]
  If $A = U \cdot A' \cdot U^\dagger$ with $U = U_1 \otimes \cdots \otimes U_D$ with $U_c \in \Unit(N_c), 1 \leq c \leq D$ and $A'$ a deterministic tensor, then we get that $\Lambda_A$ is the tensor HCIZ integral studied in \cite{collins_tensor_2023a,collins_tensor_2023}.
\end{remark}

Proceeding as in the matrix case and using the change of variable
$T' = \bar{A}_{(1_D)}^{-1/2} \cdot T$, where $A_{(1_D)}^{-1/2}$ is the square root
of the positive-definite matrix $A_{(1_D)}^{-1}$. We get the following result.
\begin{lemma}\label{lem:tensor-chvar}
  Assume that $A$ satisfies Hypothesis \ref{hyp:tensors-gaussian}. Denote by
  $\mu^{\Gau}_{\bm{N}}$ the law of a tensor whose entries are
  independent complex centered normal variables of variance $N^{1-D}$,
  and by $\tilde{\nu}$ the measure given by
  \begin{equation*}
    \dd\tilde{\nu}(A) = \frac{(\det A_{(1_D)})^{-1/2}}{\int_{\Positive(N_{1} \times \cdots \times N_{d})}(\det A'_{(1_D)})^{-1/2}\dd \nu(A')}\dd \nu(A).
  \end{equation*}
  Then, for any measurable function
  $f \colon \C^{N_{1} \times \cdots \times N_{d}} \to \C$, we have
  \begin{equation*}
    \int_{\C^{N_{1} \times \cdots \times N_{d}}}f(T_A)\dd \mu_{\bm{N}}(T_A)
    = \int_{\Positive(N_{1} \times \cdots \times N_{d})}\int_{\C^{N_{1} \times \cdots \times N_{d}}} f(\bar{A}^{1/2}\cdot T_\un) \dd \mu^{\Gau}_{\bm{N}}(T_\un)\dd \tilde{\nu}(A).
  \end{equation*}
\end{lemma}
\begin{proof}
  The proof is similar to the one of Lemma \ref{lem:matrix-chvar}. By
  performing the change of variable $T' = \bar{A}^{-1/2} \cdot T$, whose
  Jacobian is $\det A_{(1_D)}^{-1/2}$, we get the result.
\end{proof}

The computation of expectation of trace-invariants can be performed using
Lemma \ref{lem:tensor-chvar}. We have
\begin{equation}\label{eq:gaussian-external}
  \E\Bigl[\Tr_{\bm{\sigma}}(T_{A}, \bar{T}_{A})\Bigr] = \E\Bigl[\Tr_{\bm{\sigma}}(\overline{A^{1/2}}\cdot T_{\un}, A^{1/2}\cdot \bar{T}_{\un})\Bigr],
\end{equation}
where $A$ has law $\tilde{\nu}$ and $T_{\un}$ is a standard Gaussian tensor.

\

\begin{remark}[Reduction to the matrix case]\label{rem:reduce-matrix}
  In the case $D = 2$ and $A = \Id \otimes C$ where $C$ is
  a random positive definite matrix, we recover the matrix result. Indeed, we then get
  \begin{equation*}
    A_{i_{1}i_{2};j_{1}j_{2}} = \delta_{i_{1},j_{1}}C_{i_{2}j_{2}}
  \end{equation*}
  and
  \begin{equation*}
    A^{-1} = \Id \otimes C^{-1}.
  \end{equation*}
  Hence, for any pure tensor $T$ with $2$ indices (or colors)
  \begin{equation*}
    \langle A^{-1}, T \otimes \bar{T} \rangle_{\HS}
    = \sum_{i_{1},i_{2},j_{1},j_{2} = 1}^{N}\delta_{i_{1},j_{1}}\Bigl(C^{-1}\Bigr)_{i_{2}j_{2}}T_{i_{1}i_{2}}\bar{T}_{j_{1}j_{2}}
    = \sum_{i,j,k = 1}^{N} T_{ki}\bar{T}_{kj}\Bigl(C^{-1}\Bigr)_{ij}.
  \end{equation*}
  This means that under $\mu$, the tensor $(T_A, \bar{T}_A)$ is distributed like a Gaussian matrix with random covariance $C$.
\end{remark}

As with any Gaussian family of random variables, we may use Wick
formula to compute joint moments of entries of a random tensor $(T_{A}, \bar{T}_{A})$
with law $\mu_{\bm{N}}$ given in \eqref{eq:def-tenseur-gaussien}.
\begin{lemma}[Wick formula]\label{lem:wick}
  Let $T_{A}$ be a random Gaussian tensor with a covariance $A$ that
  satisfies Hypothesis \ref{hyp:tensors-gaussian}. Let $n \in \N^{*}$
  and $\bm{i}, \bm{j} \colon [n] \to [N]^{D}$. We have
  \begin{equation*}
    \E\Bigl[\prod_{k = 1}^{n}(T_{A})_{\bm{i}(k)}(\bar{T}_{A})_{\bm{j}(k)}\Bigr]
    = N^{n(1 - D)}\sum_{\eta \in \Sym_{n}}\E\Bigl[ \prod_{k=1}^{n}A_{\bm{i}(k);\bm{j}(k)} \Bigr].
  \end{equation*}
  Furthermore, we have for any $\bm{\sigma} \in \Sym_{n}^{D}$
  \begin{equation*}
    \E\Bigl[\Tr_{\bm{\sigma}}(T_{A}, \overline{T_{A}})\Bigr]
    = N^{n(1 - D)}\sum_{\eta \in \Sym_{n}}\E\Bigl[\Tr_{\bm{\sigma}\eta}(A)\Bigr].
  \end{equation*}
\end{lemma}
\begin{proof}
  Conditionally on $A$, the family $T_{A} = \Bigl((T_A)_{\bm{i}}\Bigr)_{\bm{i} \in [N]^{D}}$ is a
  Gaussian family with covariance $N^{1 - D}A$. We get by the
  Wick formula,
  \begin{equation*}
    \E\Bigl[\prod_{k = 1}^{n}(T_A)_{\bm{i}(k)}(\bar{T}_A)_{\bm{j}(k)} \mid A\Bigr]
    = \sum_{\eta \in \Sym_{n}} \prod_{k = 1}^{n} \E\Bigl[ (T_A)_{\bm{i}(k)}(\bar{T}_A)_{\bm{j}(\eta(k))} \mid A \Bigr].
  \end{equation*}
  The covariance is then given by
  \begin{equation*}
    \E\Bigl[ (T_A)_{\bm{i}(k)}(\bar{T}_A)_{\bm{j}(\eta(k))} \mid A \Bigr] = N^{1 - D}A_{\bm{i}(k);\bm{j}(\eta(k))}.
  \end{equation*}
  Taking the expectation with respect to $A$ gives the first result.

  To get the second claim, we only notice that
  \begin{equation*}
  \begin{split}
          \E\Bigl[\Tr_{\bm{\sigma}}(T_{A}, \overline{T_{A}})\Bigr]
    &= \sum_{\bm{i}, \bm{j} \colon [n] \to [N]^{D}}\E\Bigl[ (T_A)_{\bm{i}(k)}(\bar{T}_A)_{\bm{j}(\eta(k))} \Bigr] \prod_{c=1}^{D}\prod_{k = 1}^{N} \delta_{\bm{i}(k)_{c}, \bm{j}(\sigma_{c}(k))_{c}}\\
    &= N^{n(1 - D)}\sum_{\eta \in \Sym_{n}}\E\Bigl[\Tr_{\bm{\sigma}\eta}(A)\Bigr].
  \end{split}
  \end{equation*}
  This is the second claim.
\end{proof}

We finally note that a random Gaussian tensor with random covariance
$A$ can be expressed as a product of $A$ with a standard random
Gaussian tensor.
\begin{lemma}\label{lem:GaussianCTTb}
  Let $(T_\un, \bar{T}_\un)$ be a standard Gaussian tensor, i.e. a Gaussian tensor
  with covariance identity, and $T_{A}$ be a Gaussian tensor with
  covariance $A$ that satisfies Hypothesis \ref{hyp:tensors-gaussian}. Then, we have for all $n \in \N^{*}$ and
  $\bm{\sigma} \in \Sym_{n}^{D}$
  \begin{equation*}
    \E\Bigl[\Tr_{\bm{\sigma}}(T_{A}, \overline{T_{A}})\Bigr]
     = \E\Bigl[\Tr_{\bm{\sigma}}(A \cdot T_\un, \overline{T_\un})\Bigr] = N^{n(1 - D)}\sum_{\eta \in \Sym_{n}}\E\Bigl[\Tr_{\bm{\sigma}\eta}(A)\Bigr].
  \end{equation*}
\end{lemma}
\begin{remark}[Non-positive or non-full rank covariances]\label{rem:non-posit-covar}
  Lemma \ref{lem:tensor-chvar} leads us to interpret the random tensor $B \cdot T_{\un}$ for
  any random mixed tensor $B$ and standard Gaussian tensor $T_{\un}$
  as being a Gaussian tensor with formal covariance $BB^{\dagger}$, a tensor
  that may not be a full-rank as a matrix of dimension $\prod_{c} N_{c}$.

  In fact, in view of Lemma \ref{lem:GaussianCTTb}, a generalization of the Gaussian
  tensor with random covariance $A$ is provided by
  $(A \cdot T_\un, \bar{T}_{\un})$. Indeed, this tensor always makes sens
  whether $A$ is positive-definite or not. We will study such tensors
  starting from Section \ref{sec:gener-form-prod}.
\end{remark}
\begin{proof}
  The tensor $T_{A}$ is a Gaussian family with covariance $A$. Wick
  formula -- Lemma \ref{lem:wick} -- implies that for all $\bm{\sigma} \in \Sym_{n}^{D}$
  \begin{equation*}
    \E\Bigl[\Tr_{\bm{\sigma}}(T_{A}, \overline{T_{A}})\Bigr]
    = N^{n(1 - D)}\sum_{\eta \in \Sym_{n}}\E\Bigl[\Tr_{\bm{\sigma}\eta}(A)\Bigr].
  \end{equation*}
  On the other hand,
  \begin{equation*}
    \begin{split}
      \E\Bigl[\Tr_{\bm{\sigma}}(A \cdot T_\un, \overline{T_\un})\Bigr]
      &= \sum_{\bm{i}, \bm{j} \colon [n] \to [N]}\E\Bigl[A_{\bm{i};\bm{j}}\Bigr]\E\Bigl[(T_\un^{\otimes n})_{\bm{j}}(\overline{T_\un})^{\otimes n}_{\bm{i}\circ \bm{\sigma}^{-1}}\Bigr]\\
      &= N^{n(1 - D)}\sum_{\bm{i}, \bm{j} \colon [n] \to [N]}\E\Bigl[A_{\bm{i};\bm{j}}\Bigr]\sum_{\eta \in \Sym_{n}}\delta_{\bm{j}, \bm{i}\circ \bm{\sigma}^{-1}\circ \eta}\\
      &= N^{n(1 - D)}\sum_{\eta \in \Sym_{n}}\E\Bigl[\Tr_{\eta\bm{\sigma}}(A)\Bigr].
    \end{split}
  \end{equation*}
  The invariance by conjugation of $\Tr_{\bm{\sigma}}(A)$ gives
  $\Tr_{\eta\bm{\sigma}}(A) = \Tr_{\bm{\sigma}\eta}(A)$ and hence
  \begin{equation*}
    \E\Bigl[\Tr_{\bm{\sigma}}(A \cdot T_\un, \overline{T_\un})\Bigr] = \E\Bigl[\Tr_{\bm{\sigma}}(T_{A}, \overline{T_{A}})\Bigr].\qedhere
  \end{equation*}
\end{proof}

\subsection{Finite size precursors}

The quantities whose $N \to \infty$ limits define the asymptotic cumulants and free cumulants are  $\Phip_{\bsig} [T_A,\bar T_A]$ and $\Kp_{\bsig} [T_A,\bar T_A]$. Here we express them in terms of the analogous quantities for $A$.

\begin{prop}\label{prop:finite-size-precursors-T-vs-C}
  Let $n \in \N^{*}$, $\bm{\sigma} \in \Sym_{n}^{D}$, $A$ a random mixed
  tensor satisfying Hypothesis \ref{hyp:tensors-gaussian}, and
  $(T_{A}, \bar{T}_{A})$ with law $\mu_{N}$ (see
  \eqref{eq:def-tenseur-gaussien}). For any $N \geq 1$,
  \begin{equation}\label{eq:GPhiKp-TA}
    \begin{split}
      \Gb_{\bm{\sigma}}[T_{A}, \bar{T}_{A}] &= N^{n(1 - D)}\sum_{\eta \in \Sym_{n}}\Gb_{\bm{\sigma}\eta}[A]\\
      \Phip_{\bm{\sigma}}[T_{A}, \bar{T}_{A}] &= N^{n(1 - D)}\sum_{\eta \in \Sym_{n}}\Phim_{\bm{\sigma}\eta}[A]\\
      \Kp_{\bm{\sigma}}[T_{A}, \bar{T}_{A}] &= N^{n(1 - D)}\sum_{\eta \in \Sym_{n}}\Km_{\bm{\sigma}\eta}[A].
    \end{split}
  \end{equation}
\end{prop}
\begin{proof}
  We start by proving the first claim on $\Gb_{\bm{\sigma}}[T_{A}, \bar{T}_{A}]$, we have by definition
  \begin{equation*}
    \Gb_{\bm{\sigma}}[T_{A}, \bar{T}_{A}]
    = \sum_{\bm{\rho} \in \Sym_{n}^{D}}\E\Bigl[\Tr_{\bm{\rho}}(T_{A}, \bar{T}_{A})\Bigr]\Weingarten_{N}(\bm{\rho}\bm{\sigma}^{-1}).
  \end{equation*}
  Using Lemma \ref{lem:wick} and changing variables from $\bm{\rho}\eta$ to
  $\bm{\rho}$, we get
  \begin{equation*}
    \Gb_{\bm{\sigma}}[T_{A}, \bar{T}_{A}]
    = N^{n(1 - D)}\sum_{\bm{\rho} \in \Sym_{n}^{D}}\sum_{\eta \in \Sym_{n}}\E\Bigl[\Tr_{\bm{\rho}\eta}[A]\Bigr]\Weingarten_{N}(\bm{\rho}\bm{\sigma}^{-1})
    = N^{n(1 - D)}\sum_{\eta \in \Sym_{n}}\Gb_{\bm{\sigma}\eta}[A].
  \end{equation*}

  To prove the second claim on $\Phip_{\bm{\sigma}}[T_{A}, \bar{T}_{A}]$,
  one starts from the classical cumulant-moment formula and use Lemma
  \ref{lem:wick}:
  \begin{equation*}
    \begin{split}
      \Phip_{\bm{\sigma}} [T_{A},\bar T_{A}]
      &= \sum_{\substack{{\Pi\in \Partition_{\pure}(n)}\\\Pi_{\pure}(\bm{\sigma}) \leq \Pi}} \mu_\Pi \prod_{B\in \Pi_{[n]}} \E\Bigl[\Tr_{\bm{\sigma}\vert_{B}}(T_A,\bar{T}_A)\Bigr]\\
      &= N^{n(1-D)} \sum_{\substack{{\Pi\in \Partition_{\pure}(n)}\\\Pi_{\pure}(\bm{\sigma}) \leq \Pi}} \mu_\Pi \sum_{\substack{{\eta\in \Sym_{n}}\\\Pi_{\pure}(\eta)\leq \Pi}} \prod_{B\in \Pi_{[n]}} \E\Bigl[\Tr_{\bm{\sigma}\eta^{-1}\vert_{B}}(A)\Bigr].
    \end{split}
  \end{equation*}
  By exchanging the sums and using the bijection $\Pi(\bm{\sigma}\eta^{-1}) \simeq \Pi_{\pure}(\bm{\sigma}, \eta)$, we get
  \begin{equation*}
    \begin{split}
      \Phip_{\bm{\sigma}} [T_{A},\bar T_{A}]
      &= N^{n(1-D)}\sum_{\eta\in \Sym_{n} }\sum_{\substack{\Pi \in \Partition_{\pure}(n)\\\Pi_{\pure}(\bm{\sigma}, \eta) \leq \Pi}}  \mu_\Pi  \prod_{B\in \Pi_{[n]}} \E\left[\Tr_{\bm{\sigma}\eta^{-1}\vert_{B}}(A)\right]\\
      &= N^{n(1-D)}\sum_{\eta\in \Sym_{n}}\sum_{\substack{{\pi\in \Partition(n)}\\\Pi(\bm{\sigma}\eta^{-1}) \leq \pi}}  \mu_\pi  \prod_{B\in \pi} \E\left[\Tr_{\bm{\sigma}\eta^{-1}\vert_{B}}(A)\right].
    \end{split}
  \end{equation*}
  The second claim immediately follows.

  Finally, the third claim on $\Kp_{\bm{\sigma}}[T_{A}, \bar{T}_{A}]$ is
  obtained through the same transformation from a pure to a regular
  partition:
  \begin{equation*}
    N^{n(D - 1)}\Kp_{\bm{\sigma}}[T_{A}, \bar{T}_{A}]
    = \sum_{\eta\in \Sym_n} \sum_{\substack{\Pi \in \Partition_{\pure}(n)\\ \Pi_{\pure}(\bm{\sigma}, \eta) \leq \Pi}} \mu_\Pi\Gb_{\Pi_{[n]}, \bsig\eta^{-1}} [A]
    =  \sum_{\eta\in \Sym_n} \sum_{\substack{\pi \in \Partition(n)\\\Pi(\bm{\sigma}\eta^{-1}) \leq \pi}} \mu_\pi
    \Gb_{\pi, \bsig\eta^{-1}} [A].\qedhere
  \end{equation*}
\end{proof}

\subsection{Asymptotics and tensorial free cumulants of arbitrary order}

\paragraph{Moments}

The large $N$ factorization property has been discussed in \eqref{eq:non-facto}. From \eqref{eq:Pure-C-Gaussian}, the moment associated to $\bsig\in \Sym_n^D$ of the standard Gaussian $(T_\un, \bar T_\un)$ scales as $n-\delta^\bullet(\bsig)$, defined in  \eqref{eq:scale-gauss-moment}.
\begin{lemma} Let $n \in \N^{*}$, $\bm{\sigma} \in \Sym_{n}^{D}$, $A$ a random mixed
  tensor satisfying Hypothesis \ref{hyp:tensors-gaussian}, and
  $(T_{A}, \bar{T}_{A})$ with law $\mu_{N}$ (see
  \ref{eq:def-tenseur-gaussien}). Assume that
   \[
\lim_{N\rightarrow \infty}  {N^{-\#\bm{\sigma}}}\E\Bigl[\Tr_\bsig(A)\Bigr] =  \vphim_{\Pi(\bsig), \bsig}(a)\;.
\]
Then the moments of $(T_A, \bar T_A)$ have the following asymptotic scaling in $N$ and large $N$ limit:
  \begin{equation}
    \lim_{N\rightarrow \infty} N^{\delta^\bullet(\bsig) - n} \E\left[\Tr_\bsig(T_A, \bar T_A)\right] = \sum_{\substack{{\eta\in \Sym_{n}}\\{d(\bsig, \eta) = \delta^\bullet(\bsig)}}} \vphim_{\Pi(\bsig\eta^{-1}), \bsig\eta^{-1}}(a)\;.
\end{equation}
In particular, $\bsig$ factorizes  for $(T_A, \bar T_A)$  if and only if it does for the standard Gaussian $(T_\un, \bar T_\un)$.
\end{lemma}
\begin{proof}
From Lemma~\ref{lem:GaussianCTTb} and the hypothesis of the present lemma, one has:
\begin{equation}
    \E\left[\Tr_\bsig(T_A, \bar T_A)\right] = N^{n(1-D)}\sum_{\eta\in \Sym_{n}} N^{\sum_{c=1}^D \#(\sigma_c\eta^{-1})} \bigl(\vphim_{\Pi(\bsig\eta^{-1}), \bsig\eta^{-1}}(a)+o(1)\bigr)\;.
\end{equation}

One has large $N$ factorization  for $(T_A, \bar T_A)$  if and only if the $\eta\in \Sym_{n}$ such that $d(\bsig, \eta) = \delta^\bullet(\bsig)$ satisfy $\Pi_\mathrm{p}(\eta) \le \Pi_\mathrm{p}(\bsig)$, which is also the necessary and sufficient condition for the large $N$ factorization of $\bsig$ for $(T_\un, \bar T_\un)$.
\end{proof}

\paragraph{Case  where the covariance is deterministic or uniform in a LU orbit. } Let $\bm{\sigma} \in \Sym_{n}^{D}$. We recall the notions introduced in  \eqref{eq:exp-scaling-pure}:
\begin{equation*}
  \begin{split}
    \exeps(\bsig)&= \min_{\eta\in \Sym_{n}}(2K(\bsig\eta^{-1}) + d(\bsig, \eta))\;,\\
    \scalp(\bsig) &= \bigl\{\eta\in \Sym_{n},\ 2K(\bsig\eta^{-1}) + d(\bsig, \eta) = \varepsilon (\bsig)\bigr\}\;.
  \end{split}
\end{equation*}

\begin{theorem}
\label{thm:free-cum-luorb}
  Consider a mixed deterministic tensor $A'$ with $D$ inputs and such that for every $\bsig\in \Sym_n^D$, $ \lim_{N\rightarrow \infty} N^{-\#\bm{\sigma}} \Tr_{\bsig}(A')<\infty$. If $A=A'$ or $A=UA'U^\dagger$ with $U=U_1\otimes \cdots \otimes U_D$, $U_c$ Haar distributed, one has  for any $\bsig\in \Sym_{n}^{D}$:
\begin{align}
\label{eq:scaling-phi-luorb}
    \Phip_{\bsig} [T_A,\bar T_A] &= N^{n - \delta(\bsig)}   \bigl( \vphip_\bsig(t_a, \bar t_a) +o(1)\bigr)\;,\\
    \Kp_{\bsig} [T_A,\bar T_A] &= N^{ 2 - \varepsilon (\bsig) + n(1-D) }    \bigl(\fcump_\bsig(t_a, \bar t_a)+o(1)\bigr)\;.
    \label{eq:scaling-K-luorb}
\end{align}
and the asymptotics cumulants  and  tensorial free cumulants of $(T_A, \bar T_A)$ are expressed as:
\begin{equation}
\label{eq:asympt-mom-T-vs-C-deter}
   \vphip_\bsig(t_a, \bar t_a)   =\sum_{\substack{{\eta\in \Sym_{n}}\\{K_p(\bsig, \eta)=1}\\{ d(\bsig, \eta) =  \delta(\bsig)}
   }}  \vphim_{\bsig \eta^{-1}}(a')
\end{equation}
\begin{equation}
\label{eq:asympt-cum-T-vs-C-deter}
    \fcump_\bsig(t_a, \bar t_a)   =\sum_{\substack{{\eta\in \scalp(\bsig)}
    }}  \lambda^{\mix}_{\bsig \eta^{-1}}(a')\;.
\end{equation}
\end{theorem}

\begin{proof}
With these assumptions, one has from \eqref{eq:GPhiKp-TA} and  \eqref{eq:mat-deter-corr}
\begin{equation}
    \Phip_{\bsig} [T_A,\bar T_A] = N^{n(1-D)} \sum_{\substack{{\eta\in \Sym_{n}}\\{\Pi(\bsig\eta^{-1})=1_n}}} \Tr_{\bsig\eta^{-1}} (A')\;,
\end{equation}
which together with \eqref{eq:mat-deter-scaling-ass} imply \eqref{eq:scaling-phi-luorb} and \eqref{eq:asympt-mom-T-vs-C-deter}.  On the other hand, from \eqref{eq:GPhiKp-TA} and  \eqref{eq:finite-precursor-log-version-deterministic}:
\begin{equation}
    \Kp_{\bsig} [T_A,\bar T_A] = N^{n(1-D)} \sum_{\eta\in \Sym_{n}}  \mathcal{L}^{\mix}_{\bsig\eta} [A']
\end{equation}
From Thm.~\ref{thm:asymptotics-mixed-tensorial}:
\begin{equation}
    \Kp_{\bsig} [T_A,\bar T_A] = N^{n(1 - D) } \sum_{\eta\in \Sym_{n}}  N^{2(1-K(\bsig\eta^{-1})) + \sum_c\#(\sigma_c\eta^{-1}) - nD} \lambda^{\mix}_{\bsig\eta^{-1}} (a') (1+o(1))\;,
\end{equation}
which imply \eqref{eq:scaling-K-luorb} and \eqref{eq:asympt-cum-T-vs-C-deter}.
\end{proof}

If $A'=\un$,  the $N^D\times N^D$ identity matrix, the values of $\varphi\mix_{\bsig\eta^{-1}}(1)$ and $\lambda_{\bsig\eta^{-1}}^\mix(1)$ are given  \eqref{eq:connected-case-ID}.  We recover as expected the value of $\varphi^\pure_\bsig(t_1, \bar t_1)$ given in Sec.~\ref{sub:Gaussian-scaling}, Eq.~\eqref{eq:asympt-cum-gaussienne}, and for $\bsig$ connected, the value of $\kappa^\pure_\bsig(t_1, \bar t_1)$ given in Sec.~6.1 of \cite{collins_free_2025}:
\begin{equation}
      \vphip_\bsig(t_1, \bar t_1)=  \# \bigl\{\eta\in \Sym_{n} \ \mid \ d(\bsig, \eta) =  \delta(\bsig)\bigr\}  \;,
      \qquad
      \fcump_\bsig(t_1, \bar t_1)= \delta_{n, 1}\;.
\end{equation}

\paragraph{General case of a random covariance with the matrix product scaling.}

\begin{theorem}
\label{thm:free-cum-gen-case}
Let $n \in \N^{*}$, $\bm{\sigma} \in \Sym_{n}^{D}$, $A$ a random mixed
  tensor satisfying Hypothesis \ref{hyp:tensors-gaussian} as well as the matrix product scaling hypothesis \eqref{eq:matrix-scaling}, and
  $(T_{A}, \bar{T}_{A})$ with law $\mu_{N}$ (see
  \ref{eq:def-tenseur-gaussien}).  Then:
  \begin{equation}
    \begin{split}
      \Phip_{\bsig} [T_A,\bar T_A] &= N^{2+n - \exeps (\bsig)}  \Bigl( \vphip_\bsig(t_a, \bar t_a) + o(1)\Bigr)\\
      \Kp_{\bsig} [T_A,\bar T_A] &= N^{ 2 + n(1-D) - \exeps (\bsig)}\Bigl(\fcump_\bsig(t_a, \bar t_a) + o(1)\Bigr)\;,
      \end{split}
  \end{equation}
and the asymptotics cumulants  and  tensorial free cumulants of $(T_A, \bar T_A)$ are expressed in terms of those of $A$ as:
\begin{equation}
\label{eq:asympt-mom-T-vs-C-pure-conn-HO}
   \vphip_\bsig(t_a, \bar t_a) 
   = \sum_{\substack{{\eta\in \scalp(\bsig)}}}  \vphim_{\bsig \eta^{-1}}(a)
\end{equation}
\begin{equation}
    \fcump_\bsig(t_a, \bar t_a)   =\sum_{\substack{{\eta\in \scalp(\bsig)} }}  \fcumm_{\bsig \eta^{-1}}(a)\;.
\end{equation}
\end{theorem}

The proof is analogous to that of \eqref{eq:scaling-K-luorb} and \eqref{eq:asympt-cum-T-vs-C-deter}.

\paragraph{Comments on the scaling and connectivity.}
Consider a deterministic mixed tensor $A'$ such that for every $\bsig\in \Sym_n^D$, $ \lim_{N\rightarrow \infty} N^{-\#\bm{\sigma}} \Tr_{\bsig}(A')<\infty$.  We have shown that under Hypothesis \ref{hyp:tensors-gaussian} (that guarantees that the definition is well-defined), a pure Gaussian random tensor whose covariance is either $A'$, or taken uniformly at random in the LU orbit of $A'$, scales as a standard complex Gaussian tensor. Furthermore,  \eqref{eq:asympt-mom-T-vs-C-deter} involves a connectivity condition which is due to the lack of randomness (in the sense that the classical cumulants of more than one connected trace-invariants on the mixed side vanish). The role of connexifying then falls down on  the permutation $\eta$.

On the other hand, still under Hypothesis \ref{hyp:tensors-gaussian} but in the general case where the random covariance of $(T_A, \bar{T_A})$ is only assumed to scale as \eqref{eq:matrix-scaling}, there is no explicit connectivity condition in \eqref{eq:asympt-mom-T-vs-C-pure-conn-HO}. To be more precise, this question depends on the elements of $\scalp(\bsig)$ (the $\eta\in \Sym$ that are  solutions of $2K(\bsig\eta^{-1}) + d(\bsig, \eta) = \varepsilon (\bsig)$): if the only elements $\eta\in \scalp(\bsig)$ are such that $\Pi_\pure(\eta)\le \Pi_\pure(\bsig)$, then the $\eta$ do not connexify, and this role falls on  $A$. From Prop.~\ref{prop:melo-for-epsilon}, this is for instance the case for $\bsig$ melonic (the $\bsig$ that scale the strongest when the number of pure connected components is fixed). In that case, the the additional randomness of $A$ is responsible for the ``boosted'' fluctuations, discussed in Sec.~\ref{subsub:comparison-scalings}, see \eqref{eq:fluctuations-delta} and \eqref{eq:fluctuations-epsilon}.

\subsection{Example: tensor products of random matrices as covariance}

The relations of Thm.~\ref{thm:free-cum-gen-case} and Thm.~\ref{thm:free-cum-luorb} agree for $\bsig $ purely connected. In this section, we focus on the purely connected $\bsig$, and for a covariance that can be written as a tensor product of matrices, one for each color $c\in[D]$.
\begin{prop}
\label{prop:purely-connected-and-tensor-product}
  If $\bsig\in \Sym_n^D$ is purely connected and $A=A_1\otimes \cdots \otimes A_D$, where $A_1, \ldots, A_D$ are  $N\times N$ unitarily invariant random matrices behaving asymptotically as \eqref{eq:matrix-scaling-D1}, one has:
\begin{equation}
\begin{split}
\label{eq:purely-connected-and-tensor-product}
   \vphip_\bsig(t_a, \bar t_a)   &=\sum_{\substack{{\eta\in \Sym_{n}}\\{ d(\bsig, \eta) =  \delta(\bsig)}
   }}  \prod_{c=1}^D \  \vphim_{\Pi(\sigma_c\eta^{-1}), \sigma_c\eta^{-1} }(a_c)\;,\\
    \fcump_\bsig(t_a, \bar t_a)   &=  \sum_{\substack{{\eta\in \Sym_{n}}\\{ d(\bsig, \eta) =  \delta(\bsig)}
   }} \prod_{c=1}^D \  \fcumm_{\Pi(\sigma_c\eta^{-1}), \sigma_c\eta^{-1}}(a_c)\;.
   \end{split}
\end{equation}
\end{prop}
\begin{proof}
The first formula is a consequence of Lemma~\ref{lem:facto-over-cycles-tens-prod}, and the second of Lemma~\ref{lem:free-cum-tensor-product}.
\end{proof}

\

From Theorem~\ref{thm:degree},  purely connected melonic $\bsig$ have the property that  there is a unique $\eta\in \Sym_n$ such that $d(\bsig, \eta)=\delta(\bsig)$.  For every $D\ge 3$, there are many other tuples $\bsig$ satisfying this property. For such $\bsig$, up to a relabeling, one can always choose this minimizer to be the identity, $\mathrm{id}_n$, and therefore:
\begin{equation}
\begin{split}
\vphip_\bsig(t_a, \bar t_a)  & =  \vphim_{\bsig}(a)\;,\\
 \fcump_\bsig(t_a, \bar t_a)  & =      \fcumm_\bsig(a) \;.
 \end{split}
\end{equation}
If in addition the hypothesis of Prop.~\ref{prop:purely-connected-and-tensor-product} are satisfied, \eqref{eq:purely-connected-and-tensor-product} simplifies to:
\begin{equation}
\begin{split}
\vphip_\bsig(t_a, \bar t_a)  & = \prod_{c=1}^D \   \vphim_{\Pi(\sigma_c), \sigma_c }(a_c)\;,\\
 \fcump_\bsig(t_a, \bar t_a)  & =   \prod_{c=1}^D \     \fcumm_{\Pi(\sigma_c), \sigma_c }(a_c) \;.
 \end{split}
\end{equation}

\begin{ex}
If $\bsig\in \Sym_n^D$ is purely connected and $\eta=\mathrm{id}_n$ is the only permutation minimizing $d(\bsig, \eta)$,  and if  $A=A_1\otimes \cdots \otimes A_D$, with either $A_c$ a Wishart random matrix of parameters $(N, r_c N)$, or $A_c = U_cA'_c U_c^\dagger$ with $U_c\in \Unit(N)$ Haar distributed and $A'c$ a deterministic matrix whose spectrum is that of  Wishart random matrix of parameters $(N, r_c N)$. 
Such matrices have free cumulants given by $r_c$, see \cite{mingo_free_2017} .
We then have:
\begin{equation}
\begin{split}
   \vphip_\bsig(t_a, \bar t_a)  & =  \prod_{c=1}^D \sum_{\rho_c\leq \sigma_c} r_c^{\#(\rho_c)}\;,\\
    \fcump_\bsig(t_a, \bar t_a)  & =      \prod_{c=1} ^ D  r_c^{\#(\sigma_c)} \;.
    \end{split}
\end{equation}
In the square case, that is, if for all $c\in [D]$, $r_c=1$, then $   \fcump_\bsig(t_a, \bar t_a) =1 $.

If $\bm{\sigma}$ is melonic and connected and $r_1 = \cdots = r_D = r$, we have $\#\bsig = 1 + n(D - 1)$ and
\begin{equation}
    \fcump_\bsig(t_a, \bar t_a)  = r^{1 + n(D-1)}
\end{equation}
Written otherwise, if we rescale $A$ by $r^{1 - D}$, setting
$\tilde{A} = r^{1 - D}A$, we get
\begin{equation}\label{eq:free-Poisson}
  \fcump_{\bm{\sigma}}(t_{\tilde{a}}, \bar t_{\tilde{a}}) = r.
\end{equation}
This gives an asymptotic sequence of free cumulants reminiscent of the
free Poisson law of parameter $r$ appearing in (random matrix) free probability.

Finally, we notice that such random tensor distribution can be realized
with a deterministic $A$: it suffices to take a tensor product of
deterministic matrices whose eigenvalues are the zeroes of an
appropriate Laguerre polynomial. Indeed, in the large dimension limit,
the eigenvalues of a Wishart matrices are well-approximated by zeroes of Laguerre polynomials, see for instance \cite{dette_strong_2002} and references therein.
\end{ex}

This provides  examples of distributions whose melonic  tensorial free cumulants are one, which was a question raised in
\cite{collins_free_2025}, Sec.~6.2. More generally, Gaussians with LU-invariant random covariances are the first examples of LU-invariant random tensors with no-trivial tensorial free cumulants (all explicit examples previously studied were global unitary invariant).

\section{General formulae for  products of tensors}
\label{sec:gener-form-prod}

Motivated by Lemmata \ref{lem:tensor-chvar} and
\ref{lem:GaussianCTTb}, and the discussion in Remark
\ref{rem:non-posit-covar}, we study the product of a mixed and a pure
tensor: $(B \cdot T, \bar{T})$. It is also natural to consider the
product $A \cdot B$ for $A$ and $B$ two mixed tensors. This section is
dedicated to give general formulae at finite $N$ for the moments,
classical cumulants, precursors $\Gb$, and finite-$N$ precursors to the free
cumulants of such product tensors.

Notice that in general, we will not require that $B$ is LU-invariant.
Indeed, it suffices that the tensors it is multiplied to are
LU-invariant to write general formulae.

\subsection{Moments of products of tensors}

We start by studying the moments of a tensor $(B \cdot T, \bar{T})$ in the pure
case, and $B \cdot A$ in the mixed case. Actually, we give a proof in
the more general case of families of tensors
$(\bm{B} \cdot \bm{T}, \overline{\bm{T}}) = ((B^{(1)} \cdot T^{(1)}, \bar{T}^{(1)}), \ldots, (B^{(l)} \cdot T^{(l)}, \bar{T}^{(l)}))$ and
$\bm{B} \cdot \bm{A} = (B^{(1)} \cdot A^{(1)}, \ldots, B^{(l)} \cdot A^{(l)})$.
\begin{prop}[Moments of product of tensors]\label{prop:moments-tensors}
  Let $n \in \N^{*}$ and $\bm{B} = (B^{(1)}, \ldots, B^{(n)})$ be a family of random mixed tensor.
  \begin{itemize}
    \item Let $\bm{A} = (A^{(1)}, \ldots, A^{(n)})$ be a LU-invariant
          family of mixed tensors independent from $\bm{B}$, and
          $\bm{\sigma} \in \Sym_{n}^{D}$, then
          \begin{equation*}
            \E\Bigl[ \Tr_{\bm{\sigma}}(\bm{B} \cdot \bm{A})\Bigr]
            = \sum_{\bm{\rho}, \bm{\tau} \in \Sym_{n}^{D}}\E\Bigl[ \Tr_{\bm{\rho}}(\bm{B})\Bigr]\Weingarten_{\bm{N}}(\bm{\rho}\bm{\tau}\bm{\sigma}^{-1})\E\Bigl[ \Tr_{\bm{\tau}}(\bm{A})\Bigr].
          \end{equation*}
    \item Let $\bm{T} = (T^{(1)}, \ldots, T^{(n)})$ be a LU-invariant
          family of pure tensors independent from $\bm{B}$, and
          $\bm{\sigma} \in \Sym_{n}^{D}$. We have
          \begin{equation*}
            \E\Bigl[ \Tr_{\bm{\sigma}}(\bm{B} \cdot \bm{T}, \overline{\bm{T}})\Bigr]
            = \sum_{\bm{\rho}, \bm{\tau} \in \Sym_{n}^{D}}\E\Bigl[ \Tr_{\bm{\rho}}(\bm{B})\Bigr]\Weingarten_{\bm{N}}(\bm{\rho}\bm{\tau}\bm{\sigma}^{-1})\E\Bigl[ \Tr_{\bm{\tau}}(\bm{T}, \overline{\bm{T}})\Bigr].
          \end{equation*}
  \end{itemize}
\end{prop}
\begin{proof}
  We prove the result in the pure case, the computation is similar in
  the mixed case. The invariance by conjugation by local-unitary elements yields
  \begin{equation*}
    \E \Bigl[\Tr_{\bm{\sigma}}(\bm{B} \cdot \bm{T}, \overline{\bm{T}}) \Bigr]
    = \E \Bigl[\Tr_{\bm{\sigma}}(\bm{B} \cdot U \cdot \bm{T}, \bar{U} \cdot \overline{\bm{T}}) \Bigr],
  \end{equation*}
  where $U = U_{1} \otimes \cdots U_{D} \in \Unit(N)^{\otimes D}$ is
  Haar-distributed. By expanding the trace-invariant and using the independence,
  we get
  \begin{equation*}
    \E \Bigl[\Tr_{\bm{\sigma}}(\bm{B} \cdot \bm{T}, \overline{\bm{T}}) \Bigr]
    = \sum_{\bm{i}, \bm{j}, \bm{j'}, \bm{k}, \bm{k'} \colon [n] \to [N]^{D}} \E \Bigl[\bm{B}_{\bm{j'} \circ \bm{\sigma};\bm{j}}\Bigr] \E\Bigl[(U^{\otimes n})_{\bm{j}, \bm{k}}(\bar{U}^{\otimes n})_{\bm{j'}, \bm{k'}}\Bigr]\E \Bigl[\bm{T}_{\bm{k}}\overline{\bm{T}}_{\bm{k'}}\Bigr],
  \end{equation*}
  where we recall that we use the notation
  $\bm{B}_{\bm{i};\bm{j}} = \prod_{k=1}^{n}B^{(k)}_{\bm{i}(k);\bm{j}(k)}$ and
  $\bm{T}_{\bm{i}} = \prod_{k = 1}^{n}T^{(k)}_{\bm{i}(k)}$.

  The Weingarten formula --- Theorem \ref{thm:Weingarten-formula} ---
  allows us to integrate over the unitary variables
  \begin{equation*}
    \begin{split}
    \E \Bigl[\Tr_{\bm{\sigma}}(\bm{B} \cdot \bm{T}, \overline{\bm{T}}) \Bigr]
    &= \sum_{\bm{\rho}, \bm{\tau} \in \Sym_{n}^{D}} \ \sum_{\bm{j}, \bm{j'}, \bm{k}, \bm{k'} \colon [n] \to [N]^{D}} \delta_{\bm{j}, \bm{j'} \circ \bm{\rho}}\delta_{\bm{k}, \bm{k'} \circ \bm{\tau}}\E \Bigl[\bm{B}_{\bm{j'}\circ \bm{\sigma}; \bm{j}}\Bigr] \E \Bigl[\bm{T}_{\bm{k}}\overline{\bm{T}}_{\bm{k'}}\Bigr]\Weingarten_{\bm{N}}(\bm{\rho}^{-1}\bm{\tau})\\
    &= \sum_{\bm{\rho}, \bm{\tau} \in \Sym_{n}^{D}} \ \sum_{\bm{j}, \bm{k} \colon [n] \to [N]^{D}}\E \Bigl[\bm{B}_{\bm{j}\circ \bm{\sigma}\bm{\rho}^{-1}; \bm{j}}\Bigr] \E \Bigl[\bm{T}_{\bm{k}}\overline{\bm{T}}_{\bm{k} \circ \bm{\tau}^{-1}}\Bigr]\Weingarten_{\bm{N}}(\bm{\rho}^{-1}\bm{\tau})\\
    &= \sum_{\bm{\rho}, \bm{\tau} \in \Sym_{n}^{D}}\E \Bigl[\Tr_{\bm{\sigma\rho^{-1}}}(\bm{B})\Bigr] \E \Bigl[\Tr_{\bm{\tau}}(\bm{T}, \overline{\bm{T}})\Bigr]\Weingarten_{\bm{N}}(\bm{\rho}^{-1}\bm{\tau}).
    \end{split}
  \end{equation*}
  The invariance by conjugation of the Weingarten function
  $\Weingarten_{\bm{N}}$ allows us to conclude:
  \begin{equation*}
    \E \Bigl[\Tr_{\bm{\sigma}}(\bm{B} \cdot \bm{T}, \overline{\bm{T}}) \Bigr]
    = \sum_{\bm{\rho}, \bm{\tau} \in \Sym_{n}^{D}}\E \Bigl[\Tr_{\bm{\rho}}(\bm{B})\Bigr] \Weingarten_{\bm{N}}(\bm{\rho}\bm{\tau}\bm{\sigma}^{-1}) \E \Bigl[\Tr_{\bm{\tau}}(\bm{T}, \overline{\bm{T}})\Bigr].   \qedhere
  \end{equation*}
\end{proof}

\subsection{Classical cumulants for products of tensors}

Let us now give formulae to express the classical cumulants for
products of tensors. Actually, we give formulae for a more general
version of the cumulants defined as follows. Let $n \in \N^{*}$ and
$\bm{\sigma} \in \Sym_{n}^{D}$. Given a family of $n$ mixed tensors
$\bm{A}$ and $\pi_{1}, \pi_{2} \in \Partition(n)$, we set in the mixed
case
\begin{equation}\label{eq:def-super-cumulant-mixed}
  \Phim_{\pi_{2}, \bm{\sigma}}[\pi_{1}; \bm{A}] = \sum_{\substack{\pi \in \Partition(n)\\\pi_{1} \vee \Pi(\bm{\sigma}) \leq \pi \leq \pi_{2}}}\Biggl(\prod_{S \in \pi_{2}}\mu_{\pi\vert_{S}} \E\Bigl[ \Tr_{\bm{\sigma}\vert_{S}}(\bm{A})\Bigr]\Biggr).
\end{equation}
Given a family of $n$ pure tensors $(\bm{T}, \overline{\bm{T}})$ and
$\Pi_{1}, \Pi_{2} \in \PPart(n)$, we set in the pure case
\begin{equation}\label{eq:def-super-cumulant-pure}
  \Phip_{\Pi_{2}, \bm{\sigma}}[\Pi_{1}; \bm{T}, \overline{\bm{T}}] = \sum_{\substack{\Pi \in \PPart(n)\\\Pi_{1} \vee \Pi_{\pure}(\bm{\sigma}) \leq \Pi \leq \Pi_{2}}}\Biggl(\prod_{S \in \Pi_{2}}\mu_{\Pi\vert_{S}} \E\Bigl[ \Tr_{\bm{\sigma}\vert_{S}}(\bm{T}, \overline{\bm{T}})\Bigr]\Biggr).
\end{equation}
The formulae for $\pi_{1} \neq 0_{n}$ and $\Pi_{1} \neq 0_{n, \bar{n}}$
are provided for the sake of completeness: for instance, if
$\bm{\sigma^{(1)}}, \ldots, \bm{\sigma^{(K)}}$ are the connected components of
$\bm{\sigma}$ (in the sense that $\bm{\sigma} = \bm{\sigma^{(1)}} \cdots \bm{\sigma^{(K)}}$ and
$K(\bm{\sigma^{(p)}}) = 1$), having $\pi_{1}$ or $\Pi_{1}$ non-zero and coarser than $\Pi(\bm{\sigma})$ or $\Pi_{\pure}(\bm{\sigma})$ means that
we consider cumulants of products of trace-invariants associated to
the $\bm{\sigma^{(p)}}$'s. We give later a similar generalizations of the
finite-$N$ precursors to the free cumulants.

\begin{prop}[Classical cumulants for products of tensors]\label{prop:classical-cum-product}
  Let $n \in \N^{*}$, and $\bm{B} = (B^{(1)}, \ldots, B^{(n)})$ be a family of random mixed tensors.
  \begin{itemize}
    \item Let $\bm{A} = (A^{(1)}, \ldots, A^{(n)})$ be a LU-invariant
          family of mixed tensors independent from $\bm{B}$,
          $\pi_{1}, \pi_{2} \in \Partition(n)$, and
          $\bm{\sigma} \in \Sym_{n}^{D}$, then
          \begin{equation*}
            \begin{split}
              \Phim_{\pi_{2}, \bm{\sigma}}\Bigl[\pi_{1}; \bm{B} \cdot \bm{A}\Bigr]
= \sum_{\bm{\rho}, \bm{\tau} \in \Sym_{N}^{D}}&\sum_{\substack{\pi', \pi'' \in \Partition(n)\\\Pi(\bm{\rho}) \leq \pi' \leq \pi_{2} \\ \Pi(\bm{\tau}) \leq \pi'' \leq \pi_{2}}} \Phim_{\pi', \bm{\rho}}[\bm{B}]\Phim_{\pi'', \bm{\tau}}[\bm{A}]\\
              &\quad\quad\times\prod_{S \in \pi_{2}}\WeingCm{N}\Bigl[\pi_{1} \vee \Pi(\bm{\sigma}) \vee \pi' \vee \pi''\vert_{S}, \bm{\rho}\bm{\tau}\bm{\sigma}^{-1}\vert_{S}\Bigr].
            \end{split}
          \end{equation*}
    \item Let
          $(\bm{T}, \overline{\bm{T}}) = \Bigl((T^{(1)}, \bar{T}^{(1)}), \ldots, (T^{(n)}, \bar{T}^{(n)})\Bigr)$
          be a LU-invariant family of pure tensors independent from
          $\bm{B}$, $\Pi_{1}, \Pi_{2} \in \Partition_{\pure}(n)$, and
          $\bm{\sigma} \in \Sym_{n}^{D}$. We have
          \begin{equation*}
            \begin{split}
              \Phip_{\Pi_{2}, \bm{\sigma}}\Bigl[\Pi_{1}; \bm{B} \cdot \bm{T}, \overline{\bm{T}}\Bigr]
              = \sum_{\bm{\rho}, \bm{\tau} \in \Sym_{N}^{D}}&\sum_{\substack{\pi' \in \Partition(n)\\\Pi(\bm{\rho}) \leq \pi' \leq (\Pi_{2})_{[n]}}}\sum_{\substack{\Pi'' \in \Partition_{\pure}(n)\\ \Pi_{\pure}(\bm{\tau}) \leq \Pi'' \leq \Pi_{2}}} \Phim_{\pi', \bm{\rho}}[\bm{B}]\Phip_{\Pi'', \bm{\tau}}[\bm{T}, \bar{\bm{T}}]\\
              &\quad\quad\times\prod_{S \in \Pi_{2}}\WeingCp{N}\Bigl[\Pi_{1} \vee \Pi_{\pure}(\bm{\sigma}) \vee \pi' \vee \Pi''\vert_{S}, \bm{\rho}\bm{\tau}\bm{\sigma}^{-1}\vert_{S}\Bigr].
            \end{split}
          \end{equation*}
  \end{itemize}
\end{prop}
\begin{proof}
  We give the proof for pure case, the mixed case being very similar.
  We consider the case $\Pi_{2} = 1_{n, \bar{n}}$. The general case is
  obtained by taking a product of cumulants of this form. Proposition
  \ref{prop:moments-tensors} gives
  \begin{equation*}
    \begin{split}
      \Phip_{1_{n, \bar{n}}, \bm{\sigma}}\Bigl[\Pi_{1}; \bm{B} \cdot \bm{T}, \overline{\bm{T}}\Bigr]
      &= \sum_{\substack{\Pi \in \Partition_{\pure}(n)\\\Pi_{1} \vee \Pi_{\pure}(\bm{\sigma}) \leq \Pi}} \mu_{\Pi} \prod_{S \in \Pi_{[n]}} \E \Bigl[\Tr_{\bm{\sigma}\vert_{S}}(\bm{B} \cdot \bm{T}, \overline{\bm{T}})\Bigr]\\
      &= \sum_{\bm{\rho}, \bm{\tau} \in \Sym_{N}^{D}}\sum_{\substack{\Pi \in \Partition_{\pure}(n)\\\Pi_{1} \vee \Pi_{\pure}(\bm{\sigma}, \bm{\tau})\leq \Pi \\ \Pi(\bm{\rho}) \leq \Pi_{[n]}}} \mu_{\Pi} \prod_{S \in \Pi_{[n]}} \E \Bigl[\Tr_{\bm{\rho}\vert_{S}}(\bm{B})\Bigr] \Weingarten_{N}(\bm{\rho}\bm{\tau}\bm{\sigma}^{-1}\vert_{S}) \E \Bigl[\Tr_{\bm{\tau}\vert_{S}}(\bm{T}, \overline{\bm{T}})\Bigr].
    \end{split}
  \end{equation*}
  We rewrite the two expectations as products of classical cumulants and get
  \begin{equation*}
    \begin{split}
      &\Phip_{1_{n, \bar{n}}, \bm{\sigma}}[\Pi_{1}; \bm{B} \cdot \bm{T}, \overline{\bm{T}}]\\
      &\quad= \sum_{\bm{\rho}, \bm{\tau} \in \Sym_{N}^{D}}\sum_{\substack{\pi' \in \Partition(n)\\\Pi(\bm{\rho}) \leq \pi'}}\sum_{\substack{\Pi, \Pi'' \in \Partition_{\pure}(n)\\\Pi_{1} \vee \Pi_{\pure}(\bm{\sigma}) \vee \pi' \vee \Pi'' \leq \Pi\\ \Pi_{\pure}(\bm{\tau}) \leq \Pi''}} \mu_{\Pi} \Phi_{\pi', \bm{\rho}}[\bm{B}]\Phip_{\Pi'', \bm{\tau}}[T, \bar{T}]\prod_{S \in \Pi_{[n]}} \Weingarten_{N}(\bm{\rho}\bm{\tau}\bm{\sigma}^{-1}\vert_{S})\\
      &\quad= \sum_{\bm{\rho}, \bm{\tau} \in \Sym_{N}^{D}}\sum_{\substack{\pi' \in \Partition(n)\\\Pi(\bm{\rho}) \leq \pi'}}\sum_{\substack{\Pi'' \in \Partition_{\pure}(n)\\ \Pi_{\pure}(\bm{\tau}) \leq \Pi''}} \Phi_{\pi', \bm{\rho}}[\bm{B}]\Phip_{\Pi'', \bm{\tau}}[T, \bar{T}]\WeingCp{N}[\Pi_{1} \vee \Pi_{\pure}(\bm{\sigma}) \vee \pi' \vee \Pi'', \bm{\rho}\bm{\tau}\bm{\sigma}^{-1}]. \qedhere
    \end{split}
  \end{equation*}
\end{proof}

\subsection{Quantities \texorpdfstring{$\Gb$}{G} for products of tensors}
We can slightly modify Proposition \ref{prop:moments-tensors} to treat
the quantities $\Gb_{\bm{\sigma}}$ in the mixed and pure case.
\begin{prop}\label{prop:Gbar-pure}
  Let $n \in \N^{*}$ and $\bm{B} = (B^{(1)}, \ldots, B^{(n)})$ be a
  family of random mixed tensors.
  \begin{itemize}
    \item Let $\bm{A} = (A^{(1)}, \ldots, A^{(n)})$ be a LU-invariant family of
          random mixed tensors independent from $\bm{B}$, and
          $\bm{\sigma} \in \Sym_{n}^{D}$. We have
          \begin{equation*}
            \Gb_{\bm{\sigma}}[B \cdot \bm{A}]
            = \sum_{\bm{\rho} \in \Sym_{n}^{D}}\Gb_{\bm{\sigma}\bm{\rho}^{-1}}[\bm{B}] \ \Gb_{\bm{\rho}}[\bm{A}].
          \end{equation*}
    \item Let $(\bm{T}, \overline{\bm{T}}) = \Bigl((T^{(1)}, \bar{T}^{(1)}), \ldots, (T^{(n)}, \bar{T}^{(n)})\Bigr)$ be a LU-invariant
          family of random pure tensors independent from $\bm{B}$. We have
          \begin{equation*}
            \Gb_{\bm{\sigma}}[\bm{B} \cdot \bm{T}, \overline{\bm{T}}]
            = \sum_{\bm{\rho} \in \Sym_{n}^{D}}\Gb_{\bm{\sigma}\bm{\rho}^{-1}}[\bm{B}] \ \Gb_{\bm{\rho}}[\bm{T}, \overline{\bm{T}}].
          \end{equation*}
  \end{itemize}
\end{prop}
\begin{proof}
  We only do the proof in the pure case as the proof in the
  mixed case is very similar. Using Proposition \ref{prop:moments-tensors}, we have
  \begin{equation*}
    \begin{split}
      \Gb_{\bm{\sigma}}[\bm{B}\cdot \bm{T}, \overline{\bm{T}}]
      &= \sum_{\bm{\rho} \in \Sym_{n}^{D}}\E\Bigl[\Tr_{\bm{\rho}}(\bm{B}\cdot \bm{T}, \overline{\bm{T}})\Bigr]\Weingarten_{\bm{N}}(\bm{\rho}\bm{\sigma}^{-1})\\
      &= \sum_{\bm{\rho}, \bm{\tau}, \bm{\mu} \in \Sym_{n}^{D}}\E\Bigl[\Tr_{\bm{\tau}}(\bm{B})\Bigr]\Weingarten_{N}(\bm{\tau}\bm{\mu}\bm{\rho}^{-1})\E\Bigl[\Tr_{\bm{\mu}}(\bm{T}, \overline{\bm{T}})\Bigr]\Weingarten_{\bm{N}}(\bm{\rho}\bm{\sigma}^{-1}).
    \end{split}
  \end{equation*}
  Using the definition of $\Gb$, this can be rewritten as
  \begin{equation*}
    \sum_{\bm{\rho}, \bm{\tau} \in \Sym_{n,\bar{n}}^{D}}\E\Bigl[\Tr_{\bm{\tau}}(\bm{B})\Bigr]\Gb_{\bm{\tau}^{-1}\bm{\rho}}[\bm{T}, \overline{\bm{T}}]\Weingarten_{N}(\bm{\rho}\bm{\sigma}^{-1}).
  \end{equation*}
  We perform the change of variable $\bm{\rho}' = \bm{\tau}^{-1}\bm{\rho}$ and get
  \begin{equation*}
    \Gb_{\bm{\sigma}}[\bm{B}\cdot \bm{T}, \overline{\bm{T}}]
    = \sum_{\bm{\rho}, \bm{\tau} \in \Sym_{n}^{D}}\E\Bigl[\Tr_{\bm{\tau}}(\bm{B})\Bigr]\Gb_{\bm{\rho}}[\bm{T}, \overline{\bm{T}}]\Weingarten_{\bm{N}}(\bm{\rho}\bm{\sigma}^{-1}\bm{\tau})
    = \sum_{\bm{\rho} \in \Sym_{n}^{D}}\Gb_{\bm{\sigma}\bm{\rho}^{-1}}[\bm{B}]\Gb_{\bm{\rho}}[\bm{T}, \overline{\bm{T}}]. \qedhere
  \end{equation*}
\end{proof}

\subsection{Finite-\texorpdfstring{$N$}{N} precursors to the free cumulants for products of tensors}

We use the extension of the finite-$N$ precursors to the free
cumulants defined as follows. In the mixed case, given
$\pi_{1}, \pi_{2} \in \Partition(n)$ and $\bm{\sigma} \in \Sym_{n}$,
we set
\begin{equation}\label{eq:finite-N-free-mix-ext}
  \Km_{\pi_{2}, \bm{\sigma}}[\pi_{1}; \bm{A}]
  = \sum_{\substack{\pi \in \Partition(n)\\\pi_{1} \vee \Pi(\bm{\sigma}) \leq \pi \leq \pi_{2}}}\mu_{\pi} \ \Gb_{\pi, \bm{\sigma}}[\bm{A}].
\end{equation}
In the pure case, given $\Pi_{1}, \Pi_{2} \in \Partition_{\pure}(n)$
and $\bm{\sigma} \in \Sym_{n}$, we set
\begin{equation}\label{eq:finite-N-free-pure-ext}
  \Kp_{\Pi_{2}, \bm{\sigma}}[\Pi_{1}; \bm{T}, \overline{\bm{T}}]
  = \sum_{\substack{\Pi \in \Partition_{n}\\\Pi_{1} \vee \Pi_{\pure}(\bm{\sigma}) \leq \Pi \leq \Pi_{2}}}\mu_{\Pi} \ \Gb_{\Pi, \bm{\sigma}}[\bm{T}, \overline{\bm{T}}].
\end{equation}
Notice that we may replace $\pi_{1}$ by $\pi_{1}\vee \pi(\bm{\sigma})$
(respectively $\Pi_{1}$ by $\Pi_{1}\vee \Pi_{\pure}(\bm{\sigma})$) in the expression
above without changing their value. 

\begin{prop}\label{prop:finite-N-cumulants-tensors}
  Let $n \in \N^{*}$ and $\bm{B} = (B^{(1)}, \ldots, B^{(n)})$ be a
  family of random mixed tensors.
  \begin{itemize}
    \item Let $\bm{A} = (A^{(1)}, \ldots, A^{(n)})$ be a LU-invariant
          family of random mixed tensors independent from $\bm{B}$,
          $\bm{\sigma} \in \Sym_{n}^{D}$, and
          $\pi_{1}, \pi_{2} \in \Partition(n)$. We have
          \begin{equation*}
            \begin{split}
              \Km_{\pi_{2}, \bm{\sigma}}[\pi_{1}; \bm{B} \cdot \bm{A}]
              = \sum_{\bm{\rho} \in \Sym_{n}^{D}}\sum_{\substack{\pi' \in \Partition(n)\\\Pi(\bm{\sigma}\bm{\rho}^{-1}) \leq \pi'}}\sum_{\substack{\pi'' \in \Partition(n)\\\Pi(\bm{\rho}) \leq \pi''\\\pi_{1} \vee \pi' \vee \pi'' = \pi_{2}}}\Km_{\pi', \bm{\sigma}\bm{\rho}^{-1}}[\bm{B}] \ \Km_{\Pi'', \bm{\rho}}[\bm{A}].
      \end{split}
    \end{equation*}
    \item Let
          $(\bm{T}, \overline{\bm{T}}) = \Bigl((T^{(1)}, \bar{T}^{(1)}), \ldots, (T^{(n)}, \bar{T}^{(n)})\Bigr)$
          be a LU-invariant family of random pure tensors independent
          from $\bm{B}$, $\bm{\sigma} \in \Sym_{n}^{D}$, and
          $\Pi_{1}, \Pi_{2} \in \Partition_{\pure}(n)$. We have
          \begin{equation*}
            \begin{split}
              \Kp_{\Pi_{2}, \bm{\sigma}}[\Pi_{1}; \bm{B}\cdot \bm{T}, \overline{\bm{T}}]
              = \sum_{\bm{\rho} \in \Sym_{n}^{D}}\sum_{\substack{\pi' \in \Partition(n)\\\Pi(\bm{\sigma}\bm{\rho}^{-1}) \leq \pi'}}\sum_{\substack{\Pi'' \in \PPart(n)\\\Pi_{\pure}(\bm{\rho}) \leq \Pi''\\\Pi_{1} \vee \pi' \vee \Pi'' = \Pi_{2}}}\Km_{\pi', \bm{\sigma}\bm{\rho}^{-1}}[\bm{B}] \ \Kp_{\Pi'', \bm{\rho}}[T, \bar{T}].
            \end{split}
          \end{equation*}
  \end{itemize}
\end{prop}
\begin{proof}
  We give the proof in the pure case only, as the mixed case and
  second pure case are very similar. The finite-$N$ precursors of the
  free cumulant
  $\Kp_{\bm{\sigma}, \Pi_{1}, \Pi_{2}}[\bm{B}\cdot \bm{T}, \overline{\bm{T}}]$
  can be written as
  \begin{equation*}
    \Kp_{\Pi_{2}, \bm{\sigma}}[\Pi_{1}; \bm{B}\cdot \bm{T}, \overline{\bm{T}}]
    = \sum_{\substack{\Pi\in \Partition_{\pure}(n)\\\Pi_{1}\vee \Pi_{\pure}(\bm{\sigma}) \leq \Pi \leq \Pi_{2}}}\Bigl(\prod_{S \in \Pi_{2}}\mu_{\Pi\vert_{S}}\Bigr) \ \Gb_{\Pi, \bm{\sigma}}[\bm{B}\cdot \bm{T}, \overline{\bm{T}}].
  \end{equation*}
  We use Proposition \ref{prop:Gbar-pure} to get
  \begin{equation*}
    \Kp_{\Pi_{2}, \bm{\sigma}}[\Pi_{1}; \bm{B}\cdot \bm{T}, \overline{\bm{T}}]
    = \sum_{\bm{\rho} \in \Sym_{n}^{D}}\sum_{\substack{\Pi\in \Partition_{\pure}(n)\\\Pi_{1}\vee \Pi_{\pure}(\bm{\sigma}) \vee \Pi_{p}(\bm{\rho}) \leq \Pi \leq \Pi_{2}}}\Bigl(\prod_{S \in \Pi_{2}}\mu_{\Pi\vert_{S}}\Bigr) \ \Gb_{ \Pi_{[n]}, \bm{\sigma}\bm{\rho}^{-1}}[\bm{B}] \  \Gb_{\Pi, \bm{\rho}}[\bm{T}, \overline{\bm{T}}].
  \end{equation*}
  We can then rewrite the quantities $\Gb$ in terms of finite-$N$ precursors of the free
  cumulants using \eqref{eq:inv-G-K-mix} and \eqref{eq:inv-G-K-pure}:
  \begin{equation*}
    \Gb_{ \Pi_{[n]}, \bm{\sigma}\bm{\rho}^{-1}}[\bm{B}] = \sum_{\substack{\pi' \in \Partition(n)\\\Pi(\bm{\sigma}\bm{\rho}^{-1}) \leq \pi' \leq \Pi_{[n]}}} \Km_{\pi', \bm{\sigma}\bm{\rho}^{-1}}[\bm{B}] \quad \text{ and } \quad \Gb_{ \Pi, \bm{\rho}}[\bm{T}, \overline{\bm{T}}] = \sum_{\substack{\Pi'' \in \PPart(n)\\\Pi_{\pure}(\bm{\rho}) \leq \Pi'' \leq \Pi}} \Kp_{\Pi'', \bm{\rho}}[\bm{T}, \overline{\bm{T}}].
  \end{equation*}
  Notice that
  \(\Pi_{\pure}(\bm{\sigma}) \leq \Pi(\bm{\sigma}\bm{\rho}^{-1}) \vee \Pi_{\pure}(\bm{\rho})\).
  Hence, after exchanging the sums, the sum on $\Pi$ becomes
  \begin{equation*}
    \sum_{\substack{\Pi \in \PPart(n)\\\Pi_{1} \vee \pi' \vee \Pi'' \leq \Pi \leq \Pi_{2}}}\Bigl(\prod_{S \in \Pi_{2}}\mu_{\Pi\vert_{S}}\Bigr) = \delta_{\Pi_{1} \vee \pi' \vee \Pi'', \Pi_{2}},
  \end{equation*}
  where we used Remark \ref{rem:lattice-Pp}. We get
  \begin{equation*}
    \Kp_{\Pi_{2}, \bm{\sigma}}[\Pi_{1}; \bm{B}\cdot \bm{T}, \overline{\bm{T}}]
    = \sum_{\bm{\rho} \in \Sym_{n}^{D}}\sum_{\substack{\pi' \in \Partition(n)\\\Pi(\bm{\sigma}\bm{\rho}^{-1}) \leq \pi'}}\sum_{\substack{\Pi'' \in \PPart(n)\\\Pi_{\pure}(\bm{\rho}) \leq \Pi''\\\Pi_{1} \vee \pi' \vee \Pi'' = \Pi_{2}}}\Km_{\pi', \bm{\sigma}\bm{\rho}^{-1}}[\bm{B}]\Kp_{\Pi'', \bm{\rho}}[\bm{T}, \overline{\bm{T}}].\qedhere
  \end{equation*}
\end{proof}

\subsection{The product of tensor \texorpdfstring{$(B_1 \cdot T, \bar{B}_2 \cdot \bar{T})$}{(B1 T, B2 T)}}
\label{sec:product-tensor-b}

In the pure case, a natural product of tensor is to consider
$(B_{1} \cdot T, \bar{B}_{2} \cdot \bar{T})$, where $(T, \bar{T})$ is a pure
tensor and $B$ and $B'$ are mixed tensor. Whenever both $B_{1}$ and $B_{2}$
are tensor products of random matrices, we recover the case discussed
in the preceding sections.
\begin{lemma}\label{lem:prod-BT-BT}
  Let $n \in \N^{*}$ and $\bm{\sigma} \in \Sym_{n}^{D}$. Let
  $(\bm{T}, \bm{\bar{T}}) = \Bigl((T^{(1)}, \bar{T}^{(1)}), \ldots, (T^{(n)}, \bar{T}^{(n)})\Bigr)$
  be a pure tensor, and $B_{1}$ and $B_{2}$ be two mixed tensors, with $B_{1}$ a tensor product of random matrices. We have
  \begin{equation*}
    \Tr_{\bm{\sigma}}(B_{1} \cdot \bm{T}, \bar{B}_{2} \cdot \overline{\bm{T}}) = \Tr_{\bm{\sigma}}\Bigl((B_{2}^{\dagger}B_{1})  \cdot \overline{\bm{T}}, \bm{T}\Bigr).
  \end{equation*}

  Furthermore, we have
  \begin{equation*}
    \begin{split}
      \Gb_{\bm{\sigma}}[B_{1} \cdot T, \bar{B}_{2} \cdot \bar{T}] &= \Gb_{\bm{\sigma}}[(B_{2}^{\dagger}B_{1}) \cdot \bm{T}, \overline{\bm{T}}]\\
      \Phip_{\bm{\sigma}}[B_{1} \cdot \bm{T}, \bar{B}_{2} \cdot \overline{\bm{T}}] &= \Phip_{\bm{\sigma}}[(B_{2}^{\dagger}B_{1}) \cdot \bm{T}, \overline{\bm{T}}]\\
      \Kp_{\bm{\sigma}}[B_{1} \cdot \bm{T}, \bar{B}_{2} \cdot \overline{\bm{T}}] &= \Kp_{\bm{\sigma}}[(B_{2}^{\dagger}B_{1}) \cdot \bm{T}, \overline{\bm{T}}].
    \end{split}
  \end{equation*}
\end{lemma}
\begin{proof}
  Let us prove the first claim. We have
  \begin{equation*}
    \Tr_{\bm{\sigma}}(B_{1} \cdot \bm{T}, \bar{B}_{2} \cdot \overline{\bm{T}})
    = \sum_{\bm{i}, \bm{j}, \bm{k} \colon [n] \to [N]^{D}}(B_{1})_{\bm{i}\circ \bm{\sigma}; \bm{j}}(\bar{B}_{2})_{\bm{i}; \bm{k}}\bm{T}_{\bm{j}}\overline{\bm{T}}_{\bm{k}}.
  \end{equation*}
  Since $B_{1} = B_{1}^{(1)} \otimes \cdots \otimes B_{1}^{(D)}$, we have
  \begin{equation*}
    (\bar{B}_{2})_{\bm{i}; \bm{k}}
    = \prod_{c=1}^{D} \prod_{p=1}^{n}(\bar{B}_{2}^{(c)})_{\bm{i}_{c}(p); \bm{k}_{c}(p)}
    = \prod_{c=1}^{D} \prod_{p=1}^{n}(\bar{B}_{2}^{(c)})_{\bm{i}_{c}\circ \sigma_{c}(p); \bm{k}_{c} \circ \sigma_{c}(p)}
    = (B_{2}^{\dagger})_{\bm{k} \circ \bm{\sigma}; \bm{i}\circ\bm{\sigma}}.
  \end{equation*}
  This allows us to sum on $\bm{i}$ and get the first claim:
  \begin{equation*}
    \Tr_{\bm{\sigma}}(B_{1} \cdot \bm{T}, \bar{B}_{2} \cdot \overline{\bm{T}})
    = \sum_{\bm{j}, \bm{k} \colon [n] \to [N]^{D}}(B_{2}^{\dagger} \cdot B_{1})_{\bm{k}\circ\bm{\sigma}; \bm{j}}\bm{T}_{\bm{j}}\bm{\bar{T}}_{\bm{k}}
    = \Tr_{\bm{\sigma}}\Bigl((B_{2}^{\dagger}\cdot B_{1}) \cdot \overline{\bm{T}}, \bm{T}\Bigr).
  \end{equation*}

  The second claim is then a direct consequence of definitions
  \eqref{eq:BGW-mom-G}, \eqref{eq:def-Phip}, and
  \eqref{eq:HOFC-def-TT-short}.
\end{proof}

\section{Asymptotics for products of tensors}
\label{sec:large-n-limit}

When considering random matrices, and free variables obtained as
$N \to \infty$ limits of such objects, the free cumulants of products
of free variables satisfy remarkable relations. An example is the
following theorem.
\begin{theorem}[Restatement of {\cite[Theorem 14.4]{nica_lectures_2006}}]\label{thm:matrix-free-cumulant-product}
  Let $a$ and $b$ be two free variables. For all $n \in \N^{*}$, and
  $\sigma \in \Sym_{n}$ with $K(\sigma) = 1$,
  \begin{equation*}
    \fcumm_{\sigma}(b \cdot a) = \sum_{\tau \in \planar(\sigma)}\fcumm_{\tau}(b)\fcumm_{\sigma\tau^{-1}}(a).
  \end{equation*}
\end{theorem}
\begin{remark}\label{rem:translate-free}
  Theorem \ref{thm:matrix-free-cumulant-product} is stated using
  slightly different object that in \cite{nica_lectures_2006}. Let us explain the
  correspondence. Let $\gamma = \cycle{1, \ldots, n}$. Given a permutation
  $\sigma$ with one cycle, i.e.\ with $K(\sigma) = 1$, and $\mu$ such that $\gamma = \mu \sigma\mu^{-1}$, there is a bijection
  between the sets of non-crossing partitions on $n$ elements $\NC(n)$
  and the set $\planar(\sigma)$ given by
  \begin{equation*}
    \begin{cases}
      \planar(\sigma) &\to \NC(n)\\
      \rho &\mapsto \Pi(\mu\rho\mu^{-1}).
    \end{cases}
  \end{equation*}
  Through this bijection, $\sigma\tau^{-1}$ is then sent to the Kreweras
  complement of $\Pi(\mu\tau\mu^{-1})$.
\end{remark}

We will see in the sequel (in Theorems \ref{thm:cum-product-mix} and \ref{thm:free-cum-pure-mix}) analogue formulae to Theorem
\ref{thm:matrix-free-cumulant-product} for the free cumulants of
products of tensors.

\subsection{Classical cumulants for products of tensors}
\label{sec:class-cumul-prod}

To describe the asymptotics of the classical cumulants for products of
tensors, we use pairs of forests of permutations, together with an additional
condition, which we now introduce

\begin{definition}[Intertwined pairs of forests of permutations]\label{def:comp-forest-perm}
  Let $n \in \N^{*}$ and $\bm{\sigma} \in \Sym_{n}^{D}$. A pair of mixed
  forest of permutations $(\bm{\tau}, \pi_{2}) \in \FSymm(\bm{\sigma})$ and
  $(\bm{\rho}, \pi_{1}) \in \FSymm(\bm{\sigma}\bm{\tau}^{-1})$ is said to be
  \emph{intertwined} if
  \begin{equation*}
    L\Bigl(\pi_{1} \vee \Pi(\bm{\sigma}\bm{\tau}^{-1}), \pi_{2} \vee \Pi(\bm{\sigma}); \Pi(\bm{\sigma}\bm{\tau}^{-1})\Bigr) = 0.
  \end{equation*}
  We denote by $\ISymm(\bm{\sigma})$ the set of such pairs $\Bigl((\bm{\rho}, \pi_{1}), (\bm{\tau}, \pi_{2})\Bigr)$.

  Similarly, given $\eta \in \Sym_{n}$, a pair of a pure and a mixed
  forest of permutations
  $(\bm{\tau}, \Pi_{2}) \in \FSymp(\bm{\sigma}, \eta)$
  and $(\bm{\rho}, \pi_{1}) \in \FSymm(\bm{\sigma}\bm{\tau}^{-1})$ is said
  to be \emph{intertwined} if
  \begin{equation*}
    L\biggl(\pi_{1} \vee \Pi(\bm{\sigma}\bm{\tau}^{-1}), \Bigl(\Pi_{2} \vee \Pi_{\pure}(\bm{\sigma}, \eta)\Bigr)_{[n]}; \Pi(\bm{\sigma}\bm{\tau}^{-1})\biggr) = 0.
  \end{equation*}
  We denote by $\ISymp(\bm{\sigma}, \eta)$ the set of such pairs $\Bigl((\bm{\rho}, \pi_{1}), (\bm{\tau}, \Pi_{2})\Bigr)$. The set of intertwined pure pairs is then
  \begin{equation*}
    \ISymp(\bm{\sigma}) = \Biggl\{ \Bigl((\bm{\rho}, \pi_{1}), (\bm{\tau}, \Pi_{2})\Bigr) \colon \exists \eta \in \scalp(\bm{\sigma}) \cap \scalp(\Pi_{2}, \bm{\tau}), \Bigl((\bm{\rho}, \pi_{1}), (\bm{\tau}, \Pi_{2})\Bigr) \in \ISymp(\bm{\sigma}, \eta)\Biggr\}.
  \end{equation*}

  Finally, the set of intertwined Gaussian pairs is
  \begin{equation*}
    \ISymg(\bm{\sigma}) = \Biggl\{ \Bigl((\bm{\rho}, \pi_{1}), (\bm{\tau}, \Pi_{2})\Bigr) \colon \exists \eta \in \scalp(\bm{\sigma}) \cap \scalg(\Pi_{2}, \bm{\tau}), \Bigl((\bm{\rho}, \pi_{1}), (\bm{\tau}, \Pi_{2})\Bigr) \in \ISymp(\bm{\sigma}, \eta)\Biggr\}.
  \end{equation*}
\end{definition}
\begin{remark}\label{rem:IS-nonempty}
  The sets $\ISymm, \ISymp$, and $\ISymg$ are non-empty. Indeed, if we let $\eta \in \scalp(\bm{\sigma})$ and take
  $\bm{\rho} = \bm{\symid_{n}}$, $\pi_{1}  = 0_{n}$,
  $\bm{\tau} = \bm{\sigma}$, $\pi_{2} = \Partition(\bm{\sigma})$, and $\Pi_{2}= \PPart(\bm{\sigma}, \eta)$, we get that
  \begin{equation*}
  \Bigl((\bm{\rho}, \pi_{1}), (\bm{\tau}, \pi_{2})\Bigr) \in \ISymm(\bm{\sigma}),
    \quad \text{ and } \quad
  \Bigl((\bm{\rho}, \pi_{1}), (\bm{\tau}, \Pi_{2})\Bigr) \in \ISymp(\bm{\sigma}) \cap \ISymg(\bm{\sigma}).
\end{equation*}
\end{remark}

\begin{prop}\label{prop:class-cum-prod-mixte}
  Let $n \in \N^{*}$, $\bm{\sigma} \in \Sym_{n}^{D}$, and $A, B$ be two
  independent random mixed tensor with $A$ LU-invariant. We assume
  that their classical cumulants satisfy \eqref{eq:matrix-scaling}. Then the classical cumulants of $B\cdot A$ also satisfies \eqref{eq:matrix-scaling} and we have:
  \begin{equation*}
    \Phim_{\bm{\sigma}}[B \cdot A] = N^{2(1 - K(\bm{\sigma})) + \# \bm{\sigma} } \Bigl( \vphim_{\bm{\sigma}}(b \cdot a) + o(1)\Bigr)
  \end{equation*}
  with
  \begin{equation*}
      \vphim_{\bm{\sigma}}(b \cdot a)
      = \sum_{\bigl((\bm{\rho}, \pi_{1}), (\bm{\tau}, \pi_{2})\bigr) \in \ISymm(\bm{\sigma})}\vphim_{\pi_{1}, \bm{\rho}}(b)\vphim_{\pi_{2}, \bm{\tau}}(a)\Gamma\Bigl[\Pi(\bm{\sigma}) \vee \pi_{1} \vee \pi_{2}, \bm{\rho}\bm{\tau}\bm{\sigma}^{-1}\Bigr].
  \end{equation*}
\end{prop}
Note that the sum appearing in the definition of $\vphim_{\bm{\sigma}}(b\cdot a)$ is non-trivial by Remark \ref{rem:IS-nonempty}.
\begin{remark}[Connected $\bm{\sigma}$ in the mixed case]\label{rem:prod-mix-cum-connected}
  If $\bm{\sigma}$ is connected, i.e.\ $K(\bm{\sigma}) = 1$, the condition
  $L(\pi_{1} \vee \Pi(\bm{\sigma}\bm{\tau}^{-1}), \pi_{2} \vee \Pi(\bm{\sigma}); \Pi(\bm{\sigma}\bm{\tau}^{-1})) = 0$
  simplifies and becomes $\pi_{1} \leq \Pi(\bm{\sigma}\bm{\tau}^{-1})$. We get that
  \begin{equation*}
    \vphim_{\bm{\sigma}}(b \cdot a)
    = \sum_{(\bm{\tau}, \pi_{2}) \in \FSymm(\bm{\sigma})}\vphim_{\pi_{2}, \bm{\tau}}(a)\sum_{\substack{(\bm{\rho}, \pi_{1}) \in \FSymm(\bm{\sigma}\bm{\tau}^{-1})\\\pi_{1} \leq \Pi(\bm{\sigma}\bm{\tau}^{-1})}}\vphim_{\pi_{1}, \bm{\rho}}(b)\Mnc(\bm{\rho}\bm{\tau}\bm{\sigma}^{-1}).
  \end{equation*}
  Furthermore, $(\bm{\rho}, \pi_{1}) \in \FSymm(\bm{\sigma}\bm{\tau}^{-1})$ and $\pi_{1} \leq \Pi(\bm{\sigma}\bm{\tau}^{-1})$, are equivalent to $\bm{\rho} \in \planar(\bm{\sigma}\bm{\tau}^{-1})$, $\pi_{1} = \Pi(\bm{\rho}) \leq \Pi(\bm{\sigma}\bm{\tau}^{-1})$. We get
  \begin{equation*}
    \begin{split}
    \vphim_{\bm{\sigma}}(b \cdot a)
      &= \sum_{(\bm{\tau}, \pi_{2}) \in \FSymm(\bm{\sigma})}\vphim_{\pi_{2}, \bm{\tau}}(a)\sum_{\substack{\bm{\rho} \in \planar(\bm{\sigma}\bm{\tau}^{-1})\\\Pi(\bm{\rho}) \leq \Pi(\bm{\sigma}\bm{\tau}^{-1})}}\vphim_{\Pi(\bm{\rho}), \bm{\rho}}(b)\Mnc(\bm{\rho}\bm{\tau}\bm{\sigma}^{-1})\\
      &= \sum_{(\bm{\tau}, \pi_{2}) \in \FSymm(\bm{\sigma})}\vphim_{\pi_{2} \bm{\tau}}(a)\fcumm_{\Pi(\bm{\sigma}\bm{\tau}^{-1}), \bm{\sigma}\bm{\tau}^{-1}}(b).
    \end{split}
  \end{equation*}
\end{remark}
\begin{proof}
  Proposition \ref{prop:classical-cum-product} gives us
  \begin{equation*}
    \Phim_{\bm{\sigma}}[B \cdot A]
    = \sum_{\bm{\rho}, \bm{\tau} \in \Sym_{N}^{D}}\sum_{\substack{\pi', \pi'' \in \Partition(n)\\\Pi(\bm{\rho}) \leq \pi'\\ \Pi(\bm{\tau}) \leq \pi''}} \Phim_{\pi', \bm{\rho}}[B]\Phi_{\pi'', \bm{\tau}}[A]\WeingCm{N}\Bigl[\Pi(\bm{\sigma}) \vee \pi' \vee \pi'', \bm{\rho}\bm{\tau}\bm{\sigma}^{-1}\Bigr].
  \end{equation*}
  Using \eqref{eq:matrix-scaling} and Lemma \ref{lem:sum-exp-mix} we have
  \begin{equation*}
    \Phim_{\pi', \bm{\rho}}[B]
    = N^{2(\# \pi' \vee \Pi(\bm{\sigma}\bm{\tau}^{-1}) - K(\bm{\sigma}\bm{\tau}^{-1})) - d(\bm{\sigma}, \bm{\tau}) + d(\bm{\rho}, \bm{\sigma}\bm{\tau}^{-1}) + nD - \alpha_{1}}\Bigl(\vphim_{\pi', \bm{\rho}}(b) + o(1)\Bigr)
  \end{equation*}
  where $\alpha_{1} \geq 0$ with equality if and only if $(\bm{\rho}, \pi') \in \FSymm(\bm{\sigma}\bm{\tau}^{-1})$. Similarly, we have
  \begin{equation*}
    \Phim_{\pi'', \bm{\tau}}[A]
    = N^{2(\# \pi'' \vee \Pi(\bm{\sigma}) - K(\bm{\sigma})) + \# \bm{\sigma} + d(\bm{\tau}, \bm{\sigma}) - \alpha_{2}}\Bigl(\vphim_{\pi'', \bm{\tau}}(a) + o(1)\Bigr),
  \end{equation*}
  where $\alpha_{2} \geq 0$ with equality if and only if
  $(\bm{\tau}, \pi'') \in \FSymm(\bm{\sigma})$.
  Finally, we have by Theorem \ref{thm:asympt-cumulant-weingarten-funct}
  \begin{equation*}
    \begin{split}
    \WeingCm{N}\Bigl[\Pi(\bm{\sigma}) &\vee \pi' \vee \pi'', \bm{\rho}\bm{\tau}\bm{\sigma}^{-1}\Bigr]\\
    &= N^{2(1 - \# \Pi(\bm{\sigma})\vee \pi' \vee \pi'') - nD - d(\bm{\tau}, \bm{\sigma}\bm{\rho}^{-1})}\Bigl(\Gamma\Bigl[\Pi(\bm{\sigma}) \vee \pi' \vee \pi'', \bm{\rho}\bm{\tau}\bm{\sigma}^{-1}\Bigr] + o(1)\Bigr).
    \end{split}
  \end{equation*}

  Putting everything together, we get that the product
  $\Phim_{\pi', \bm{\rho}}[B]\Phim_{\pi'', \bm{\tau}}[A]\WeingCm{N}\Bigl[\Pi(\bm{\sigma}) \vee \pi' \vee \pi'', \bm{\rho}\bm{\tau}\bm{\sigma}^{-1}\Bigr]$
  is of order
  \begin{equation*}
    N^{2(1 + \# \pi' \vee \Pi(\bm{\sigma}\bm{\tau}^{-1}) + \# \pi'' \vee \Pi(\bm{\sigma}) - K(\bm{\sigma}\bm{\tau}^{-1}) - K(\bm{\sigma})-\# \Pi(\bm{\sigma})\vee \pi' \vee \pi'') + \# \bm{\sigma} - \alpha_{1} - \alpha_{2}}.
  \end{equation*}
  We notice that since $\Pi(\bm{\tau}) \leq \pi''$ we have
  $\Pi(\bm{\sigma}) \vee \pi'' = \Pi(\bm{\sigma}\bm{\tau}^{-1}) \vee \pi''$ and
  \begin{equation*}
    \begin{split}
    -2L(\Pi(\bm{\sigma}\bm{\tau}^{-1}) \vee \pi', &\Pi(\bm{\sigma}) \vee \pi_{2}; \Pi(\bm{\sigma}\bm{\tau}^{-1}))\\
    & = 2 \# \Pi(\bm{\sigma}\bm{\tau}^{-1}) \vee \pi' + 2\# \Pi(\bm{\sigma}) \vee \pi'' - 2\# \Pi(\bm{\sigma}) \vee \pi' \vee \pi'' - 2 K(\bm{\sigma}\bm{\tau}^{-1}).
    \end{split}
  \end{equation*}
  This quantity is non-positive. We thus get that the power of $N$ is at most
  \begin{equation*}
    2(1 - K(\bm{\sigma})) + \# \bm{\sigma}
  \end{equation*}
  with equality if and only if $\Bigl((\bm{\rho}, \pi'), (\bm{\tau}, \pi'')\Bigr) \in \ISymm(\bm{\sigma})$. Note that by Remark \ref{rem:IS-nonempty}, this set is non-empty.
\end{proof}

\begin{prop}\label{prop:class-cum-prod-pure}
  Let $n \in \N^{*}$, $\bm{\sigma} \in \Sym_{n}^{D}$, $(T, \bar{T})$
  be a LU-invariant random pure tensor that satisfies either the
  scaling assumption \eqref{eq:scaling-pure} for pure tensors or the
  Gaussian scaling assumption \eqref{eq:scaling-gaussian}, and $B$ a
  random mixed tensor independent from $(T, \bar{T})$, that satisfies
  the scaling assumption \eqref{eq:matrix-scaling} for mixed tensors.
  Then, $(B \cdot T, \bar{T})$ satisfies the scaling assumption
  \eqref{eq:scaling-pure}:
  \begin{equation*}
    \Phip_{\bm{\sigma}}[B \cdot T, \bar{T}] = N^{n+2-\exeps(\bm{\sigma})} \Bigl( \vphip_{\bm{\sigma}}(b \cdot t, \bar{t}) + o(1)\Bigr)
  \end{equation*}
  with
  \begin{equation*}
    \vphip_{\bm{\sigma}}(b \cdot t, \bar{t}) = \sum_{\bigl((\bm{\rho}, \pi_{1}), (\bm{\tau}, \Pi_{2})\bigr) \in \ISymp(\bm{\sigma})}\vphim_{\pi_{1}, \bm{\rho}}(b)\vphim_{\Pi_{2}, \bm{\tau}}(a)\Gamma\Bigl[\Pi_{\pure}(\bm{\sigma}) \vee \pi_{1} \vee \Pi_{2}, \bm{\rho}\bm{\tau}\bm{\sigma}^{-1}\Bigr]
  \end{equation*}
  if $(T, \bar{T})$ satisfies \eqref{eq:scaling-pure}, and
  \begin{equation*}
    \vphip(b \cdot t, \bar{t}) = \sum_{\bigl((\bm{\rho}, \pi_{1}), (\bm{\tau}, \Pi_{2})\bigr) \in \ISymg(\bm{\sigma})}\vphim_{\pi_{1}, \bm{\rho}}(b)\vphim_{\Pi_{2}, \bm{\tau}}(a)\Gamma\Bigl[\Pi_{\pure}(\bm{\sigma}) \vee \pi_{1} \vee \Pi_{2}, \bm{\rho}\bm{\tau}\bm{\sigma}^{-1}\Bigr]
  \end{equation*}
  if $(T, \bar{T})$ satisfies \eqref{eq:scaling-gaussian}.
\end{prop}
Note that the sums appearing in the two definitions of
$\vphip_{\bm{\sigma}}(b\cdot t, \bar{t})$ are non-trivial by Remark
\ref{rem:IS-nonempty}.
\begin{remark}[Connected $\bm{\sigma}$ in the pure case]\label{rem:prod-pure-cum-connected}
  If $\bm{\sigma}$ is purely connected, i.e.\ $K_{\pure}(\bm{\sigma}) = 1$, the condition $L(\pi_{1} \vee \Pi(\bm{\sigma}\bm{\tau}^{-1}), \Bigl(\Pi_{2} \vee \Pi_{\pure}(\bm{\sigma}, \eta)\Bigr)_{[n]}; \Pi(\bm{\sigma}\bm{\tau}^{-1})) = 0$
  simplifies and becomes $\pi_{1} \leq \Pi(\bm{\sigma}\bm{\tau}^{-1})$. We get that
  \begin{equation*}
    \vphip_{\bm{\sigma}}(b \cdot t, \bar{t})
    = \sum_{(\bm{\tau}, \Pi_{2}) \in \FSymp(\bm{\sigma})}\vphip_{\Pi_{2}, \bm{\tau}}(t, \bar{t})\sum_{\substack{(\bm{\rho}, \pi_{1}) \in \FSymm(\bm{\sigma}\bm{\tau}^{-1})\\\pi_{1} \leq \Pi(\bm{\sigma}\bm{\tau}^{-1})}}\vphim_{\pi_{1}, \bm{\rho}}(b)\Mnc(\bm{\rho}\bm{\tau}\bm{\sigma}^{-1}).
  \end{equation*}
  Furthermore, $(\bm{\rho}, \pi_{1}) \in \FSymm(\bm{\sigma}\bm{\tau}^{-1})$ and $\pi_{1} \leq \Pi(\bm{\sigma}\bm{\tau}^{-1})$, are equivalent to $\bm{\rho} \in \planar(\bm{\sigma}\bm{\tau}^{-1})$, $\pi_{1} = \Pi(\bm{\rho}) \leq \Pi(\bm{\sigma}\bm{\tau}^{-1})$. We get
  \begin{equation*}
    \vphim_{\bm{\sigma}}(b \cdot a)
    = \sum_{(\bm{\tau}, \Pi_{2}) \in \FSymp(\bm{\sigma})}\vphip_{\Pi_{2}, \bm{\tau}}(t, \bar{t})\fcumm_{\Pi(\bm{\sigma}\bm{\tau}^{-1}), \bm{\sigma}\bm{\tau}^{-1}}(b).
  \end{equation*}
\end{remark}
\begin{proof}
  Proposition \ref{prop:classical-cum-product} immediately gives
  \begin{equation*}
    \begin{split}
      \Phip_{\bm{\sigma}}[B \cdot \bm{T}, \overline{\bm{T}}]
      &= \sum_{\bm{\rho}, \bm{\tau} \in \Sym_{N}^{D}}\sum_{\substack{\pi' \in \Partition(n)\\\Pi(\bm{\rho}) \leq \pi'}}\sum_{\substack{\Pi'' \in \Partition_{\pure}(n)\\\Pi_{\pure}(\bm{\tau}) \leq \Pi''}} \Phi_{\pi', \bm{\rho}}[B]\Phip_{\Pi'', \bm{\tau}}[\bm{T}, \bar{\bm{T}}]
      \WeingCp{N}\Bigl[\Pi_{\pure}(\bm{\sigma}) \vee \pi' \vee \Pi'', \bm{\rho}\bm{\tau}\bm{\sigma}^{-1}\Bigr].
    \end{split}
  \end{equation*}

  Let $\eta \in \scalp(\Pi'', \bm{\tau})$. Note that we have in particular
  $\Pi_{\pure}(\bm{\tau}, \eta) \leq \Pi''$. When considering the classical
  cumulant of $B$, the scaling assumption \eqref{eq:matrix-scaling} and Lemma \ref{lem:sum-exp-mix} give
  \begin{equation*}
    \Phi_{\pi', \bm{\rho}}[B]
    = N^{2(\# \pi' \vee \Pi(\bm{\sigma}\bm{\tau}^{-1}) - K(\bm{\sigma}\bm{\tau}^{-1})) - d(\bm{\sigma}, \bm{\tau}) + d(\bm{\rho}, \bm{\sigma}\bm{\tau}^{-1}) + nD - \alpha_{1}}\Bigl(\vphim_{\Pi', \bm{\rho}}(b) + o(1)\Bigr)
  \end{equation*}
  where $\alpha_{1} \geq 0$ with equality if and only if
  $(\bm{\rho}, \pi') \in \FSymm(\bm{\sigma}\bm{\tau}^{-1})$.
  Similarly, in the pure case, we use the scaling assumption \eqref{eq:scaling-pure} and Lemma
  \ref{lem:sum-exp-pure} to get
  \begin{equation*}
    \Phip_{\Pi'', \bm{\tau}}[T, \bar{T}]
    = N^{n + 2(\# \Pi_{\pure}(\bm{\sigma}) \vee \Pi'' - K_{\pure}(\bm{\sigma}, \eta)) - d(\bm{\sigma}, \eta) + d(\bm{\sigma}, \bm{\tau}) - \alpha_{2}}\Bigl(\vphip_{\Pi'', \bm{\tau}}(t, \bar{t}) + o(1)\Bigr),
  \end{equation*}
  where $\alpha_{2} \geq 0$ with equality if and only if
  $(\bm{\tau}, \Pi'') \in \FSymp(\bm{\sigma}, \eta)$. Finally, we have
  \begin{equation*}
    \begin{split}
    \WeingCp{N}\Bigl[\Pi_{\pure}(\bm{\sigma}) &\vee \pi' \vee \Pi'', \bm{\rho}\bm{\tau}\bm{\sigma}^{-1}\Bigr]\\
    &= N^{2(1 - \# \Pi_{\pure}(\bm{\sigma})\vee \pi' \vee \Pi'') - nD - d(\bm{\tau}, \bm{\sigma}\bm{\rho}^{-1})}\Biggl(\Gamma\Bigl[\Bigl(\Pi_{\pure}(\bm{\sigma}) \vee \Pi' \vee \Pi''\Bigr)_{[n]}, \bm{\rho}\bm{\tau}\bm{\sigma}^{-1}\Bigr] + o(1)\Biggr).
    \end{split}
  \end{equation*}

  We thus get that the product $\Phi_{\pi', \bm{\rho}}[B]\Phip_{\Pi'', \bm{\tau}}[\bm{T}, \bar{\bm{T}}]\WeingCp{N}\Bigl[\Pi_{\pure}(\bm{\sigma}) \vee \pi' \vee \Pi'', \bm{\rho}\bm{\tau}\bm{\sigma}^{-1}\Bigr]$ is of order
  \begin{equation*}
    N^{2(1 + \# \pi' \vee \Pi(\bm{\sigma}\bm{\tau}^{-1}) + \# (\Pi'' \vee \Pi_{\pure}(\bm{\sigma}, \eta)) - K(\bm{\sigma}\bm{\tau}^{-1})- K_{\pure}(\bm{\sigma}, \eta) - \# \Pi_{\pure}(\bm{\sigma})\vee \pi' \vee \Pi'') - d(\bm{\sigma}, \bm{\tau}) + n - \alpha_{1} - \alpha_{2}}.
  \end{equation*}
  We notice that since $\Pi_{\pure}(\bm{\tau}) \leq \Pi''$,
  $\Pi(\bm{\sigma}\bm{\tau}) \leq (\Pi'' \vee \Pi_{\pure}(\bm{\sigma}, \eta))_{[n]}$.
  Hence, we can introduce the quantity
  \begin{equation*}
    \begin{split}
      -2 L\Bigl(\pi' \vee \Pi(\bm{\sigma}\bm{\tau}^{-1}), &(\Pi'' \vee \Pi_{\pure}(\bm{\sigma}, \eta)); \Pi(\bm{\sigma}\bm{\tau}^{-1})\Bigr)\\
      &= 2 \# \pi' \vee \Pi(\bm{\sigma}\bm{\tau}^{-1}) + 2 \# \Pi'' \vee \Pi_{\pure}(\bm{\sigma}, \eta) - 2K_{\pure}(\bm{\sigma}\bm{\tau}^{-1}) - 2 \# \pi' \vee \Pi'' \vee \Pi_{\pure}(\bm{\sigma}).
    \end{split}
  \end{equation*}
  This quantity is non-positive, so that the maximum exponent of $N$ for fixed $\eta$ is
  \begin{equation*}
    n + 2(1 - K_{\pure}(\bm{\sigma}, \eta)) - d(\bm{\sigma}, \eta).
  \end{equation*}
  This is obtained whenever
  $\Bigl((\bm{\rho}, \pi'), (\bm{\tau}, \Pi'')\Bigr) \in \ISymp(\bm{\sigma}, \eta)$. Note that by Remark \ref{rem:IS-nonempty}, this set is non-empty. The
  maximal exponent is then obtained whenever $\eta \in \scalp(\bm{\sigma})$.
  This gives the result.

  In the case where $(T, \bar{T})$ satisfies the scaling assumption
  \eqref{eq:scaling-gaussian}, the proof is identical, except that
  $\eta \in \scalg(\Pi'', \bm{\tau})$ instead of
  $\eta \in \scalp(\Pi'', \bm{\tau})$.
\end{proof}

\subsection{Free cumulants of products of tensors}
\label{sec:free-cumul-prod}

We now compute the free cumulants of products of tensors in the case
of a product of mixed tensors, and then of a product of a mixed and
pure tensor.
\begin{theorem}\label{thm:cum-product-mix}
  Let $A$ be a random LU-invariant mixed tensor, and $B$ be a random
  mixed tensor independent from $A$. Assume that $A$ and $B$ satisfy the scaling assumption \eqref{eq:matrix-scaling}. For all $n \in \N^{*}$ and
  $\bm{\sigma} \in \Sym_{n}^{D}$, we have
  \begin{equation*}
    \Km_{\bm{\sigma}}[B \cdot A] = N^{2(1 - K(\bm{\sigma})) + \# \bm{\sigma}- nD} \Bigl( \fcumm_{\bm{\sigma}}(b \cdot a) + o(1) \Bigr),
  \end{equation*}
  where
  \begin{equation*}
    \fcumm_{\bm{\sigma}}(b \cdot a) = \sum_{\substack{\pi_{1}, \pi_{2} \in \Partition(n)\\\pi_{1} \vee \pi_{2} = 1_{n}}}\sum_{\substack{\bm{\rho} \in \Sym_{n}^{D}\\\bigl((\bm{\sigma}\bm{\rho}^{-1}, \pi_{1}), (\bm{\rho}, \pi_{2})\bigr) \in \ISymm(\bm{\sigma})}}\fcumm_{\pi_{1}, \bm{\sigma}\bm{\rho}^{-1}}(b) \fcumm_{\pi_{2}, \bm{\rho}}(a).
  \end{equation*}
\end{theorem}
\begin{remark}[Mixed connected case]\label{rem:connected-intertwined}
  In the mixed case, when $K(\bm{\sigma}) = 1$, the definition of $\ISymm(\bm{\sigma})$ simplifies. Indeed, the condition
  \begin{equation*}
    L\Bigl(\pi_{1} \vee \Pi(\bm{\sigma}\bm{\rho}^{-1}), \pi_{2} \vee \Pi(\bm{\sigma}); \Pi(\bm{\sigma}\bm{\rho}^{-1})\Bigr) = 0
  \end{equation*}
  immediately implies that
  $\pi_{1} \leq \Pi(\bm{\sigma}\bm{\rho}^{-1})$. Together with the
  fact that
  $(\bm{\sigma}\bm{\rho}^{-1}, \pi_{1}) \in \FSymm(\bm{\sigma}\bm{\rho}^{-1})$,
  and thus that $\Pi(\bm{\sigma}\bm{\rho}^{-1}) \leq \pi_{1}$, we
  get that $\pi_{1} = \Pi(\bm{\sigma}\bm{\rho}^{-1})$. Hence, if
  $K(\bm{\sigma}) = 1$
  \begin{equation*}
    \fcumm_{\bm{\sigma}}(b \cdot a) = \sum_{(\bm{\rho}, \pi) \in \FSymm(\bm{\sigma})}\fcumm_{\Pi(\bm{\sigma}\bm{\rho}^{-1}), \bm{\sigma}\bm{\rho}^{-1}}(b) \fcumm_{\pi, \bm{\rho}}(a).
  \end{equation*}
  We now notice that if $(\bm{\rho}, \pi) \in \FSymm(\bm{\sigma})$ then we
  immediately get that $\pi = \Pi(\bm{\rho})$. We finally get a
  formula analogous to the one of Theorem \ref{thm:matrix-free-cumulant-product}: if
  $K(\bm{\sigma}) = 1$,
  \begin{equation*}
    \fcumm_{\bm{\sigma}}(b \cdot a) = \sum_{\bm{\rho} \in \planar(\bm{\sigma})} \fcumm_{\Pi(\bm{\sigma}\bm{\rho}^{-1}), \bm{\sigma}\bm{\rho}^{-1}}(b)\fcumm_{\Pi(\bm{\rho}), \bm{\rho}}(a).
  \end{equation*}
\end{remark}

\begin{proof}
  By Proposition \ref{prop:finite-N-cumulants-tensors}, we have
  \begin{equation*}
    \Km_{\bm{\sigma}}[B \cdot A]
    = \sum_{\bm{\rho} \in \Sym_{n}^{D}}\sum_{\substack{\pi' \in \Partition(n)\\\Pi(\bm{\sigma}\bm{\rho}^{-1}) \leq \pi'}}\sum_{\substack{\pi'' \in \Partition(n)\\\Pi(\bm{\rho}) \leq \pi''\\\pi' \vee \pi'' = 1_{n}}}\Km_{\pi', \bm{\sigma}\bm{\rho}^{-1}}[B]\Km_{\pi'', \bm{\rho}}[A].
  \end{equation*}

  Using Theorem \ref{thm:asymptotics-mixed-tensorial}, we get
  \begin{equation*}
  \begin{split}
    \Km_{\bm{\sigma}}[B \cdot A]
    = \sum_{\bm{\rho} \in \Sym_{n}^{D}}&\sum_{\substack{\pi' \in \Partition(n)\\\Pi(\bm{\sigma}\bm{\rho}^{-1}) \leq \pi'}}\sum_{\substack{\pi'' \in \Partition(n)\\\Pi(\bm{\rho}) \leq \pi''\\\pi' \vee \pi'' = 1_{n}}}N^{2(\#\pi' + \# \pi'' - K(\bm{\rho}) - K(\bm{\sigma}\bm{\rho}^{-1})) - d(\bm{\sigma}, \bm{\rho}) + \# \bm{\rho}- nD}\\
    &\times\Bigl( \fcumm_{\pi', \bm{\sigma}\bm{\rho}^{-1}}(b)\fcumm_{\pi'', \bm{\rho}}(a) + o(1) \Bigr).
  \end{split}
  \end{equation*}

  The proof is the very similar to the end of the one of Proposition
  \ref{prop:class-cum-prod-mixte}. We use Lemma \ref{lem:sum-exp-mix}
  to get
  \begin{equation*}
    \begin{split}
    2(\#\pi' + \# \pi'' - K(\bm{\rho}) - &K(\bm{\sigma}\bm{\rho}^{-1})) - d(\bm{\sigma}, \bm{\rho}) + \# \bm{\rho}- nD\\
    &\leq 2( \# \pi' + \# \pi'' \vee \Pi(\bm{\sigma})  - K(\bm{\sigma}\bm{\rho}^{-1}) - K(\bm{\sigma})) + \# \bm{\sigma}-nD,
    \end{split}
  \end{equation*}
  with equality if and only if $(\pi'', \bm{\rho}) \in \FSymm(\bm{\sigma})$. We then introduce the quantity $L$ defined in \eqref{eq:def-of-L}:
  \begin{equation*}
    -2L\Bigl(\pi', \pi'' \vee \Pi(\bm{\sigma}); \Pi(\bm{\sigma}\bm{\rho}^{-1})\Bigr) = 2( \# \pi' + \# \pi'' \vee \Pi(\bm{\sigma})  - K(\bm{\sigma}\bm{\rho}^{-1}) - \# \pi' \vee \pi'' \vee \Pi(\bm{\sigma})).
  \end{equation*}
  Since $\pi' \vee \pi'' = 1_{n}$, we get
  \begin{equation*}
    2(\#\pi' + \# \pi'' - K(\bm{\rho}) - K(\bm{\sigma}\bm{\rho}^{-1})) - d(\bm{\sigma}, \bm{\rho}) + \# \bm{\rho}- nD
    \leq 2( 1 - K(\bm{\sigma})) + \# \bm{\sigma} - nD,
  \end{equation*}
  with equality if and only if
  $(\pi'', \bm{\rho}) \in \FSymm(\bm{\sigma})$ and
  $L\Bigl(\pi', \pi'' \vee \Pi(\bm{\sigma}); \Pi(\bm{\sigma}\bm{\rho}^{-1})\Bigr) = 0$.

  We then notice that this is equivalent to having
  \begin{equation*}
    \Bigl((\bm{\sigma}\bm{\rho}^{-1}, \pi'), (\bm{\sigma}, \pi'')\Bigr) \in \ISymm(\bm{\sigma}).
  \end{equation*}
  This gives the claim.
\end{proof}

\begin{theorem}\label{thm:free-cum-pure-mix}
  Let $(T, \bar{T})$ be a random LU-invariant pure tensor, and $B$ be
  a random mixed tensor independent from $(T, \bar{T})$. Assume that
  $(T, \bar{T})$ satisfy the scaling assumption
  \eqref{eq:scaling-pure} and that $A$ satisfies
  \eqref{eq:matrix-scaling}. For all $n \in \N^{*}$ and
  $\bm{\sigma} \in \Sym_{n}^{D}$, we have
  \begin{equation*}
    \Kp_{\bm{\sigma}}[B \cdot T, \bar{T}] = N^{n(1 - D) + 2 - \exeps(\bm{\sigma})} \Bigl( \fcump_{\bm{\sigma}}(b \cdot t, \bar{t})\Bigr)
  \end{equation*}
  where
  \begin{equation*}
    \fcump_{\bm{\sigma}}(b \cdot t, \bar{t})
    = \sum_{\substack{\pi_{1} \in \Partition(n), \Pi_{2} \in \PPart(n)\\\pi_{1} \vee \Pi_{2} = 1_{n, \bar{n}}}}\sum_{\substack{\bm{\rho} \in \Sym_{n}^{D}\\\bigl((\bm{\sigma}\bm{\rho}^{-1}, \pi_{1}), (\bm{\rho}, \Pi_{2})\bigr) \in \ISymp(\bm{\sigma})}}\fcumm_{\pi_{1}, \bm{\sigma}\bm{\rho}^{-1}}(b)\fcump_{\Pi_{2}, \bm{\rho}}(t, \bar{t}).
  \end{equation*}
\end{theorem}
\begin{remark}[Connected case]\label{rem:connected-intertwined-pure}
  In the pure case, when $K_{\pure}(\bm{\sigma}) = 1$, the definition
  of $\ISymp(\bm{\sigma})$ simplifies as in the mixed case. Indeed, for all $\eta \in \Sym_{n}$,
  the condition
  \begin{equation*}
    L\Bigl(\pi_{1} \vee \Pi(\bm{\sigma}\bm{\rho}^{-1}), (\Pi_{2} \vee \Pi_{\pure}(\bm{\sigma}, \eta))_{[n]}; \Pi(\bm{\sigma}\bm{\rho}^{-1})\Bigr) = 0
  \end{equation*}
  immediately implies that
  $\pi_{1} \leq \Pi(\bm{\sigma}\bm{\rho}^{-1})$. Together with the fact
  that
  $(\bm{\sigma}\bm{\rho}^{-1}, \pi_{1}) \in \FSymm(\bm{\sigma}\bm{\rho}^{-1})$,
  which implies $\Pi(\bm{\sigma}\bm{\rho}^{-1}) \leq \pi_{1}$, we get
  that $\pi_{1} = \Pi(\bm{\sigma}\bm{\rho}^{-1})$. Thus, if
  $K(\bm{\sigma}) = 1$
  \begin{equation*}
      \fcump_{\bm{\sigma}}(b \cdot t, \bar{t})
      = \sum_{(\bm{\rho}, \Pi) \in \FSymp(\bm{\sigma})}\fcumm_{\Pi(\bm{\sigma}\bm{\rho}^{-1}), \bm{\sigma}\bm{\rho}^{-1}}(b) \fcump_{\Pi, \bm{\rho}}(t, \bar{t}).
    \end{equation*}
\end{remark}

\begin{proof}
  By Proposition \ref{prop:finite-N-cumulants-tensors}, we have
  \begin{equation*}
    \Kp_{\bm{\sigma}}[B \cdot T, \bar{T}]
    = \sum_{\bm{\rho} \in \Sym_{n}^{D}}\sum_{\substack{\pi' \in \Partition(n)\\\Pi(\bm{\sigma}\bm{\rho}^{-1}) \leq \pi'}}\sum_{\substack{\Pi'' \in \PPart(n)\\\Pi_{\pure}(\bm{\rho}) \leq \Pi''\\\pi' \vee \Pi'' = 1_{n, \bar{n}}}}\Km_{\pi', \bm{\sigma}\bm{\rho}^{-1}}[B]\Kp_{\Pi'', \bm{\rho}}[A].
  \end{equation*}

  Using Theorem \ref{thm:asymptotics-mixed-tensorial} and Theorem \ref{thm:free-cum-product-scaling}, we get
  \begin{equation*}
    \begin{split}
      \Kp_{\bm{\sigma}}[B \cdot T, \bar{T}]
      = \sum_{\bm{\rho}\in \Sym_{n}^{D}}&\sum_{\substack{\pi' \in \Partition(n)\\\Pi(\bm{\sigma}\bm{\rho}^{-1}) \leq \pi'}}\sum_{\substack{\Pi'' \in \PPart(n)\\\Pi_{\pure}(\bm{\rho}) \leq \pi''\\\pi' \vee \Pi'' = 1_{n}}}N^{2(\#\pi' + \# \Pi'' - K(\bm{\sigma}\bm{\rho}^{-1})) - d(\bm{\sigma}, \bm{\rho}) + n(1 - D) - \exeps(\Pi'', \bm{\rho})}\\
      &\quad\times\Bigl( \fcumm_{\pi', \bm{\sigma}\bm{\rho}^{-1}}(b)\fcump_{\Pi'', \bm{\rho}}(t, \bar{t}) + o(1) \Bigr).
    \end{split}
  \end{equation*}

  Let $\eta \in \scalp(\Pi'', \bm{\rho})$. Using Lemma \ref{lem:sum-exp-pure}, the exponent of $N$ satisfies
  \begin{equation*}
    \begin{split}
      2(\#\pi' + \# \Pi'' - &K_{\pure}(\bm{\rho}, \eta) - K(\bm{\sigma}\bm{\rho}^{-1})) - d(\bm{\sigma}, \bm{\rho}) + n(1 - D) - d(\bm{\rho}, \eta)\\
                            &\leq 2(\# \pi' + \# \Pi'' \vee \Pi_{\pure}(\bm{\sigma}, \eta) - K(\bm{\sigma}\bm{\rho}^{-1}) - K_{\pure}(\bm{\sigma}, \eta)) - d(\bm{\sigma}, \eta) + n(1 - D),
    \end{split}
  \end{equation*}
  with equality if and only $(\bm{\rho}, \Pi'') \in \FSymp(\bm{\sigma}, \eta)$.

  We then introduce the quantity $L$ defined in \eqref{eq:def-of-L}:
  \begin{equation*}
    -2L\Biggl(\pi', \Bigr(\Pi'' \vee \Pi_{\pure}(\bm{\sigma}, \eta)\Bigl)_{[n]}; \Pi(\bm{\sigma}\bm{\rho}^{-1})\Biggr) = 2\#  + 2 \# \Pi''\pi' \vee \Pi_{\pure}(\bm{\sigma}, \eta) - 2 - 2K(\bm{\sigma}\bm{\rho}^{-1})
  \end{equation*}
  where we used that $\# \pi' \vee \Pi'' = \# 1_{n, \bar{n}} = 1$, since
  in the sum above we only consider partitions that satisfy
  $\pi' \vee \Pi'' = 1_{n, \bar{n}}$. Hence, we get
  \begin{equation*}
   \begin{split}
     2(\#\pi' + \# \Pi'' - &K_{\pure}(\bm{\rho}, \eta) - K(\bm{\sigma}\bm{\rho}^{-1})) - d(\bm{\sigma}, \bm{\rho}) + n(1 - D) - d(\bm{\rho}, \eta)\\
     &\leq 2 - \exeps(\bm{\sigma}) + n(1 - D)
   \end{split}
 \end{equation*}
 with equality if and only if
 \begin{itemize}
   \item $(\bm{\rho}, \Pi'') \in \FSymp(\bm{\sigma}, \eta)$;
   \item $L\Biggl(\pi', \Bigr(\Pi'' \vee \Pi_{\pure}(\bm{\sigma}, \eta)\Bigl)_{[n]}; \Pi(\bm{\sigma}\bm{\rho}^{-1})\Biggr)= 0$;
   \item $\eta \in \scalp(\bm{\sigma})$.
 \end{itemize}
 These conditions are equivalent to having
 \begin{equation*}
   \Bigl((\bm{\sigma}\bm{\rho}^{-1}, \pi'), (\bm{\rho}, \Pi'')\Bigr) \in \ISymp(\bm{\sigma}).
 \end{equation*}
 This entails the result.
\end{proof}

\appendix

\section{Proof of Theorem \ref{thm:asympt-cumulant-weingarten-funct}}
\label{sec:proof-theorem-ref}

By the definition \eqref{eq:cumulant-weingarten}, we have
\begin{equation*}
  \WeingCm{N}[\pi, \bnu] = \sum_{\substack{\pi' \in \Partition(n)\\\pi \leq \pi'}} \mu_{\pi'}  \prod_{B\in \pi'}\Weingarten_{N}(\bm{\nu}{_{|_B}}).
\end{equation*}
It was shown by Novak \cite[Theorem 3.1]{novak_jucysmurphy_2010} that the Weingarten function
admits an expansion in powers of $1/N$ as $N \to \infty$ $1/N$ in
terms of monotone walks: for all $\rho \in \Sym_{n}$ and $1 \leq c \leq D$
\begin{equation*}
  \Weingarten_{N}(\rho) = N^{-n}\sum_{r \geq 0} \frac{\# \mwalks^{r}(\rho)}{N^{r}}.
\end{equation*}

We get by exchanging the sums,
\begin{equation*}
  \WeingCm{N}[\pi, \bnu] = \sum_{\substack{\pi' \in \Partition(n)\\\pi \leq \pi'}}\sum_{\{r_{c, B} \colon 1 \leq c \leq D,  B \in  \pi'\}} (-1)^{\sum_{i} r_{c, B}} N^{-n -\sum_{i}r_{c, B}} \mu_{\pi'}  \prod_{B\in \pi'}\prod_{c = 1}^{D}\# \mwalks^{r_{c, B}}(\nu{_{|_B}}).
\end{equation*}
Introduce the sets
\begin{equation*}
  \mwalks^{r}(\pi, \nu) = \Bigl\{\bm{\tau} \in \mwalks^{r}(\nu) \colon \Pi(\nu) \vee \Pi(\bm{\tau}) = \pi\Bigr\}.
\end{equation*}
We have
\begin{equation*}
  \begin{split}
    \WeingCm{N}&[\pi, \bnu]\\
    &= \sum_{\substack{\pi', \pi_{1}, \ldots, \pi_{D} \in \Partition(n)\\\pi \vee  \pi_{1} \vee \cdots \vee \pi_{D} \leq \pi'\\\forall c, \Pi(\nu_{c}) \leq \pi_{c}}}\sum_{\{r_{c, B} \colon 1 \leq c \leq D, B \in \pi'\}} (-1)^{\sum_{i} r_{c, B}} N^{-n -\sum_{i}r_{c, B}} \mu_{\pi'}  \prod_{c = 1}^{D}\prod_{B\in \pi_{c}}\# \mwalks^{r_{c, B}}(1_{B}, \nu{_{|_B}}).
  \end{split}
\end{equation*}
The Moebius formula for partitions \eqref{eq:moebius-inversion-partition} then gives
\begin{equation*}
  \WeingCm{N}[\pi, \bnu]
  = \sum_{\substack{\pi_{1}, \ldots, \pi_{D} \in \Partition(n)\\\pi \vee \pi_{1} \vee \cdots \pi_{D} = 1_{n}\\\forall c, \Pi(\nu_{c}) \leq \pi_{c}}}\sum_{\{r_{c, B} \colon 1 \leq c \leq D, B \in \pi_{c}\}} (-1)^{\sum_{c, B} r_{c, B}} N^{-n -\sum_{c, B}r_{c, B}} \prod_{c = 1}^{D}\prod_{B\in \pi_{c}}\# \mwalks^{r_{c, B}}(1_{B}, \nu{_{|_B}}).
\end{equation*}
By Remark \ref{rem:Riemann-Hurwitz}, the exponent of $N$ is maximal
whenever for all $1 \leq c \leq D$ and $B \in \pi_{c}$, $r_{c, B} = 2(\# B - 1) - |\nu_{c}{_{\vert_{B}}}|$. This gives
the exponent of $N$
\begin{equation*}
  2(\sum_{c}\# \pi_{c} - nD) - \sum_{c}\# \nu_{c}.
\end{equation*}
Finally, we introduce the quantity $L_{D}$ defined in \eqref{eq:def-of-LD}:
\begin{equation*}
  -2L_{D}\Bigl(\{\pi_{c}\}, \pi'; \{\Pi(\nu_{c})\}\Bigr) = 2\sum_{c = 1}^{D}\Bigl(\# \pi_{c} - \# \nu_{c}\Bigr) - 2( 1 - \# \pi' ).
\end{equation*}
Hence, we have
\begin{equation*}
  2\Bigl(\sum_{c}\# \pi_{c} - nD\Bigr) - \sum_{c}\# \nu_{c} \leq \sum_{c} \# \nu_{c} - 2nD + 2(1 - \# \pi')
\end{equation*}
with equality if and only if
$L_{D}\Bigl(\{\pi_{c}\}, \pi'; \{\Pi(\nu_{c})\}\Bigr) = 0$. We end up with
\begin{equation*}
  \WeingCm{N}[\pi, \bnu]
  = N^{2(1 - \# \pi') + \sum_{c} \# \nu_{c} - 2nD}\Biggl(\sum_{\substack{\pi_{1}, \ldots, \pi_{D} \in \Partition(n)\\\pi \vee \pi_{1} \vee \cdots \vee \pi_{D} = 1_{n}\\\forall c, \Pi(\nu_{c}) \leq \pi_{c}}}(-1)^{d(\bm{\nu})}\prod_{c = 1}^{D}\prod_{B\in \pi_{c}}\gamma(\nu{_{|_B}}) + \order{N^{-1}}\Biggr),
\end{equation*}
as wanted.

To treat the pure case, it suffices to use Remark \ref{rem:WgC-mix-to-pure} which gives
\begin{equation*}
  \WeingCp{N}[\Pi, \bnu] = \WeingCm{N}[\Pi_{[n]}, \bnu] = N^{2(1 - \# \Pi) + \sum_{c} \# \nu_{c} - 2nD}\Bigl( \Gamma[\Pi_{[n]}, \bm{\nu}] + \order{N^{-1}}\Bigr). \qedhere
\end{equation*}

\printbibliography

\end{document}